\documentclass[10pt,reqno]{amsart}

\usepackage{verbatim}

\usepackage{amsmath, amssymb}
\usepackage[latin1]{inputenc}
\usepackage{xcolor}
\usepackage{tikz}
\usepackage{color}
\def\b{\mathbb{B}}
\def\B{\mathcal{B}}
\def\H{\mathcal{H}}
\def\h{\mathbb{H}}
\def\det{\text{det}}
\def\ds{\displaystyle}
\newtheorem{thm}{Theorem}[section]
\newtheorem{lem}[thm]{Lemma}
\newtheorem{prop}[thm]{Proposition}
\newtheorem{cor}[thm]{Corollary}
\newtheorem{rem}[thm]{Remark}
\newtheorem{defi}[thm]{Definition}
\newcounter{defi}

\numberwithin{equation}{section}

\renewcommand{\Re}{{\rm Re }\;}
\renewcommand{\Im}{{\rm Im }\;}

\newcommand{\be}{\begin{equation}}
\newcommand{\ee}{\end{equation}}
\newcommand{\bea}{\begin{eqnarray}}
\newcommand{\eea}{\end{eqnarray}}
\newcommand{\beas}{\begin{eqnarray*}}
\newcommand{\eeas}{\end{eqnarray*}}

\newcommand{\C}{{\mathbb C}}
\newcommand{\R}{{\mathbb R}}
\newcommand{\bN}{{\mathbb N}}

\newcommand{\vC}{{\mathcal C}}

\newcommand{\vE}{{\mathcal E}}

\newcommand{\vL}{{\mathcal L}}

\newcommand{\vP}{{\mathcal P}}

\newcommand{\vU}{{\mathcal U}}

\newcommand{\vX}{{\mathcal X}}
\newcommand{\vW}{{\mathcal W}}

\newcommand{\vZ}{{\mathcal Z}}

\newcommand{\Ran}{\mbox{\rm Ran }}

\def\la{\lambda}

\def\q{\qquad}
\def\ds{\displaystyle}

\newcommand{\beq}{\begin{equation}}
\newcommand{\eeq}{\end{equation}}
\newcommand{\beqs}{\begin{equation*}}
\newcommand{\eeqs}{\end{equation*}}
\newcommand{\bal}{\begin{aligned}}
\newcommand{\eal}{\end{aligned}}

\newcommand{\inner}[2]{\left\langle #1, #2 \right\rangle}

\newcommand\inp[2][]{#1 \langle #2#1\rangle}
\begin{document}

%
%
%
%
%
%
%
%
%

\title[Large time behavior of solutions to Schr\"odinger equation]{Large time behavior of Solutions to
Schr\"odinger equation with
complex-valued potential}

\author[M. Aafarani]{Maha Aafarani}

\address{Laboratoire de Math\'ematiques Jean Leray\\
Universit\'e de Nantes \\
44322 Nantes Cedex 3  France \\
E-mail: maha.aafarani@univ-nantes.fr }
\date{\today}



\begin{abstract}
We study the large-time behavior of the solutions to the Schr\"odinger equation  associated with a quickly decaying potential in dimension three. We establish the resolvent expansions at threshold zero and near positive resonances. The large-time expansions of solutions are obtained under different conditions, including the existence of positive resonances and zero resonance or/and zero eigenvalue.
\end{abstract}

\maketitle
\tableofcontents
\section{Introduction} \label{sec:intro}
 In this work, we are interested in the large-time behavior  of the solution $u(t) = e^{-itH}u_0$ as $t\to +\infty$ of the time-dependent Schr\"odinger equation
\begin{equation}\label{Eq:Schr}
  \left \{ \begin{array}{cccc}
  i \partial_t u(t,x)   & =& H u(t,x) &, \ x\in \R^3, \ t>0\\
    u(0,x)&=&u_0(x),& 
 \end{array}
 \right.
 \end{equation}
where $H = -\Delta +V$ is a perturbation of $-\Delta$ by a complexe-valued potential supposed to satisfy the decay condition 
\begin{equation}\label{V:decay}
|V(x)|\leq C_v \langle x\rangle^{-\rho}, \quad \forall x\in \R^3,
\end{equation}
where $\rho>2$ and $\langle x\rangle=(1+|x|^2)^{1/2}$.\\

This equation is a fundamental dynamical equation for the wave-function $u(t,x)$ describing the motion of particles in non relativistic quantum mechanics. %
An interesting model in nuclear physics where this non-seladjoint operator arises, is the optical nuclear model \cite{nuclear1954}.  
This model describes the dynamic of a compound elastic neutron scattering from a heavy nucleus. In this example, the interaction between the neutron and the nucleus is modeled by a complex-valued potential with negative imaginary part. \\

It turns out that the behavior of non-selfadjoint Schr\"odinger operators may differ from selfadjoint ones (see \cite{schwartz1960some,davies2002}). 
Previously, Jensen and Kato \cite{jensen1979spectral} have studied the three dimensional selfadjoint  operator $H=-\Delta +V$  with   rapidly decreasing real potential $V$ ($V$ satisfies the decay condition (\ref{V:decay})).
Obviously, in the selfadjoint case the decay in time of the solution is strongly linked to the analysis at low energies part of the resolvent  depending  on the presence of the zero eigenvalue or/and zero resonance. In \cite{jensen1979spectral} the low energy asymptotics of the resolvent was obtained. Also local decay in time of the solution to the time-dependent Schrodinger equation is 
found. In an extension work
\cite{komech2013dispersive} of the latter, a similar results has been obtained for Schr\"odinger operator with a magnetic potential under the assumption that zero is a regular point i.e. it is not an eigenvalue nor a resonance in the sense of \cite{jensen1979spectral,newton1977noncentral}. For studies of the non-selfadjoint Schr\"odinger operator, we  refer for example to \cite{wang2017gevrey} on Gevrey estimates of the resolvent for slowly decreasing potentials  and sub-exponential time-decay of solutions, to \cite{wang2012time,zhu:tel-01427736} for time-decay of solutions to dissipative Schr\"odinger equation and to \cite{goldberg2006dispersive,goldberg2010dispersive} for dispersive estimates.\\

Our goal is to extend the results of \cite{jensen1979spectral}  to non-selfadjoint Schr\"odinger operator  with a rapidly decreasing complex-valued potential $V$. We are interested in the spectral analysis of the operator $H$ and the large time behavior of $e^{-itH}$ for some model operator   having real resonances. Here, we mean by the latter a real number $\lambda_0 \geq 0$ for which the equation $\displaystyle  -\Delta u +Vu -\lambda_0 u=0$ has a non trivial solution $\psi\in L^{2,-s}(\R^3) \setminus L^2(\R^3)$, $ \forall s>1/2$. In particular, for $\lambda_0>0$ it can be seen that this solution satisfies the Sommerfeld radiation condition 
\begin{equation}\label{Behavior:psi}
\psi(x) = \frac{e^{\pm i\sqrt{\lambda_0}\vert x \vert }}{\vert x \vert } w(\frac{x}{\vert x\vert }) + o(\frac{1}{\vert x\vert}), \quad \text{as} \, \vert x\vert \rightarrow +\infty, 
\end{equation}
 where $w\in L^2(\mathbb{S}^2)$, $w\neq 0$, with $\mathbb{S}^2= \{ x\in \R^3: \vert x \vert =1\}$.
\\

It is well known that selfadjoint Schr\"odinger operator can have zero resonance only but no positive resonances (see \cite{ASNSP_1975_4_2_2_151_0,ikebe1972limiting, komech2013dispersive}). In \cite{jensen1979spectral} positive resonances are absent and  intermediate energies do not contribute to the large time behavior of $e^{-itH}$. While these numbers are the main difficulty to study the spectral properties of non selfadjoint operators near the positive real axis. Wang in \cite{wang2017gevrey} has included the presence of positive resonances for a compactly supported perturbation of the Schr\"odinger operator with a slowly decaying potential satisfying a condition of analyticity. However,
usually these real numbers are supposed to be absent (see \cite{goldberg2010dispersive}). 
See \cite{wang2012time} for an example of a positive resonance of a dissipative operator in dimension three and \cite{pavlov1968nonself} in dimension one.\\

Real resonances are responsible for most remarkable physical phenomena and many problems arising from analysis of spectral properties for non selfadjoint operators. Effectively, the zero resonance is responsible for Efimov effect for N-body quantum systems (see  \cite{amado1973there,tamura1991efimov,tamura1993efimov,wang2004existence} and see also the physical review \cite{naidon2017efimov}). In addition, the presence of positive resonances affects essentially the asymptotic completeness of wave operators (cf. \cite{faupin2018asymptotic}). They are further the only points to which complex eigenvalues may eventually accumulate, as well as the boundary values of the resolvent on the cut along the continuous spectrum, known as limiting absorption principle, does not exist globally (cf. \cite{royer2010limiting,saito1974principle,wang2012time}).\\ 

In order to obtain the large-time behavior of $e^{-itH}$, we establish the expansions of $R(z)$ at threshold zero and positive resonances in the sense of bounded operators between two suitable weighted Sobolev spaces (see Section \ref{sec:results}).  The  novelties in our results are the following.  
 We establish the asymptotic expansions of $R(z)$ near positive resonances (Theorem \ref{Thm:realRes1}) which are assumed to be singularities of finite order  with a specific hypothesis $(H3)$ on the behavior of solutions defined by (\ref{Behavior:psi}) (see \cite{faupin2018asymptotic,schwartz1960some} for more hypotheses concerning positive resonances). In addition,  in  
theorems \ref{thm:cas2}-\ref{thm:cas3} we obtain  the asymptotic 
expansions of the resolvent at low energies  for non-selfadjoint operator. In particular, explicit calculation of lower order coefficients has been performed. The difficulty is that the algebraic and geometric eigenspaces need not coincide, i.e not  only eigenvectors  but also Jordan chain may occur.  Our main approach extends the method of Lidskii \cite{lidskii1966perturbation} to study  the giant matrices representation that we find, where there exist many Jordan blocks corresponding to the eigenvectors associated with zero eigenvalue.
Under our hypothesis $(H1)$ and $(H2)$ we get same singularities (negative powers of $z^{1/2}$) that have appeared in \cite{jensen1979spectral} because the presence of  zero eigenvalue or/and zero resonance. More recent results can be followed in \cite{goldberg2010dispersive,wang2017gevrey}. 
\\

The obtained results are useful for the study of the asymptotic behavior in time of wave functions that would depend on the nature of the threshold energy and the characteristic of positive resonances, as well as the asymptotic expansions of the resolvent and the semigroup $e^{-itH}$ as $t\to +\infty$ have many applications in the scattering theory (see for example \cite{faupin2018asymptotic,faupin2018scattering}).
\\

This paper is organized as follows. In Section \ref{sec:results}, we introduce our hypothesis and we state the main results. In Section \ref{sec:zero resonance} we study the asymptotic expansions of the resolvent at zero energy. We prove Theorem \ref{thm:cas1} when zero is an eigenvalue of arbitrary geometric multiplicity, then we prove Theorem \ref{thm:cas2} in the case of zero resonance and Theorem \ref{thm:cas3} in the more complicated case when zero is both an eigenvalue and a resonance of $H$. Section \ref{section:real resonance} is devoted to the study of outgoing positive resonances, we establish the proof of Theorem \ref{Thm:realRes1} by assuming that the set of positive resonances is finite and  their associated eigenvectors satisfy an appropriate  assumption. Moreover,  we obtain another result (Theorem \ref{Thm:realRes1}) in more general situation. Finally, in Section \ref{sec:Representation} we establish a representation formula for the semigroup $e^{-itH}$ as $t\to +\infty$ which allows us to prove  Theorem \ref{thm:estim}. 
\\
\paragraph{Notation} Let $\mathcal{X}$ and $\mathcal{X}'$ be two Banach  spaces. We denote $\B(\mathcal{X},\mathcal{X}')$ the set of linear bounded operators from $\mathcal{X}$ to $\mathcal{X}'$. For simplicity, $\B(\mathcal{X})= \B(\mathcal{X},\mathcal{X})$.
For all $m,m',s,s' \in \mathbb{R}$, we denote by $\mathbb{H}^{m,s}$ the weighted
Sobolev space on $\mathbb{R}^3$ 
$$\mathbb{H}^{m,s}= \{ u\in \mathcal{S}^{'}(\R^3): \,\Vert u \Vert_{m,s}= \Vert \langle x\rangle^s(1-\Delta)^{m/2}u \Vert_{L^2}<\infty\}, $$ 
such that for $m<0$, $\h^{m,s}$ is defined as the dual of $\h^{-m,-s}$ with dual product identified with the scalar product of $L^2$ $\langle\cdot,\cdot\rangle$. The index  $s$ is omitted for standard Sobolev spaces, i.e. $\mathbb{H}^{m}$ denotes $\mathbb{H}^{m,0}$.
In particular $\mathbb{H}^0= L^2$ with the  associated norm $\Vert \cdot \Vert_0$.
Let $\mathbb{B}(m,s,m',s') =\mathcal{B} \left(\mathbb{H}^{m,s},\mathbb{H}^{m',s'}\right)$. For linear operator $T$, we denote by $\text{Ran}\,T$ the range  of $T$ and by $\text{rank}\,T$ its rank.
We also define the following subsets: $\R_+=[0,+\infty[, \ \R_-=]-\infty,0], \ \C_{\pm}=\lbrace z\in \C, \ \pm \Im (z) >0\rbrace$ and $ \bar{\C}_{\pm}=\lbrace z\in \C, \ \pm \Im (z) \geq 0\rbrace.$

\section{Assumptions and formulation of the main results}

\label{sec:results}
\subsection{The operator}\label{sub:model} We consider the Schr\"odinger operator $H= -\Delta +V$ in $\R^3$, where
$\Delta$ denotes the Laplacian  and $V$ is a complex-valued potential which will be assumed to satisfy the following decay condition
\begin{equation}\label{V}
| V(x)| \leq C_v \langle x\rangle^{-\rho},\quad \forall x\in \R^3,
\end{equation}
where $ \rho>2$ throughout the paper and  will depend on the results to obtain. Under the previous assumptions, $H$ is a closed non-selfadjoint operator on $L^2$ with domain the standard Sobolev space $\mathbb{H}^2$. 
Moreover, the condition (\ref{V}) implies that the operator $V$ of multiplication by $V(x)$ is relatively compact with respect to
$-\Delta$. It is then known that the essential spectrum of $H$  denoted by $\sigma_e(H)$ coincides with that of the non perturbed operator $-\Delta$ (cf. \cite{Hislop1996}).
Thus $\sigma_e(H)$ covers the positive real axis
$[0,+\infty[$. 
In addition,  the operator $H$ has no eigenvalues along the half real axis $]0,+\infty[$ (cf. \cite{k59}). 
Hence, the spectrum of $H$ denoted by $\sigma(H)$
is the disjoint union  of $\sigma_e(H)$ and a countable set denoted by $\sigma_d(H)$, with  
\begin{equation}\label{set:eigenvalues}
\sigma_d(H):=
\lbrace z \in \mathbb{C}\setminus[0,+\infty[:\, \exists\,  0\neq u  \in D(H), Hu = z u\rbrace
\end{equation}
 consisting of discrete eigenvalues with finite algebraic multiplicities. For $z\in \sigma_d(H)$, the associated Riesz projection of $H$ is defined by 
\begin{equation}\label{Riesz}
 \Pi_{z} = -\frac{1}{2 i\pi} \int_{\vert w-z\vert =\epsilon} (H-wId)^{-1}dw,
 \end{equation}
for $ \epsilon>0$ small enough. 
These eigenvalues can accumulate only on the half axis $[0,+\infty[$ at zero or at positive resonances (see Definition \ref{def:reson}).
\\

It should be noted that in this work it is sufficient to assume the existence of some constants $\rho>2$ and $C_v, R>0$ such that the assumption (\ref{V}) on the potential $V$ is replaced by the following:
$$\left \{ \begin{array}{cc}
| V(x)| \leq C_v \langle x\rangle^{-\rho},& x \in \R^3 \ \text{with} \ |x|>R,\\
\\
V \, \text{is}\, -\Delta-\text{compact}.
\end{array}
\right.$$

\paragraph{\textbf{Resolvent}}
Denote $\displaystyle R_0(z)=(-\Delta-z Id)^{-1}$ for $z\in \C \setminus \R_+$  
and $ R(z)=(H-z Id)^{-1}$ for 
$z\in \C\setminus \sigma(H)$. 
In order to obtain the asymptotic expansion of $R(z)$, we use the following relation between the resolvents
\begin{equation}\label{ResEq}
R(z)= (Id+ R_0(z)V)^{-1}R_0(z), \ \forall z \notin \sigma(H),
\end{equation}
and we need to recall some well-known facts about $R_0(z)$.\\

For $z \in \C\setminus \R_+$, $R_0(z)$ is a convolution operator from $L^2$ to itself with integral kernel 
$$ R_0(z)(x,y) = \frac{e^{+i\sqrt{z}\vert x-y \vert} }{4\pi \vert x-y \vert}, \quad \Im \sqrt{z}>0.$$ 
Here the branch of $\sqrt{z}$ is holomorphic in  $\C\setminus \R_+$ such that $\displaystyle \lim\limits_{\epsilon\to 0^+} \sqrt{\lambda\pm i\epsilon}= \pm \sqrt{\lambda},\ \forall \lambda>0$.
Moreover, the boundary values of the free resolvent on $ \R_+$  are defined by the following limits
\begin{equation}\label{Lim ResLibre}
  R_0^{\pm}(\lambda): = s-\lim\limits_{\epsilon  \to 0}
  R_0(\lambda \pm i\epsilon), \ \textit{for} \ \lambda>0, 
\end{equation}
 which  exist in the uniform operator topology of  $ \mathbb{B}(0,s,0,-s')$, for $s,s'>1/2$  (see  \cite[Theorem 4.1]{ASNSP_1975_4_2_2_151_0}). 
 In addition, for all $\ell\in \bN$ 
\begin{equation}\label{ResLibre:dev}
 R_{0}(z)=\sum_{j=0}^{\ell} (i\sqrt{z})^{j}G_{j} + o(|z|^{\ell/2}),
\end{equation}
where $G_{j}$ is an integral operator with integral kernel 
$G_{j}(x,y)=|x-y|^{j-1}/4\pi j!,\  j=0,1,\cdots,\ell,$ such that
\begin{align}
G_{0}\in \mathbb{B}(-1,s,1,-s'),& \quad s,s'>1/2 ,\quad s+s' >2, \label{G0} \\
G_{j}\in \mathbb{B}(-1,s,1,-s'),& \quad s,s'>j+1/2 , \quad j=0,1,\cdots,\ell\label{Gj}.
\end{align}
In particular, $G_0  := s -\lim\limits_{ z \to 0, z\in \C\setminus\R_+}
R_0(z)$ is formally  inverse to $-\Delta$.
See \cite[Section 2.]{jensen1979spectral}. We will often work with the variable a square root, $\eta$, $z=\eta^2$ with $\Im \eta>0$, in order to use analyticity arguments. \\

Let $\mathbb{C}\setminus \R_+ \ni z \mapsto K(z):=R_0(z)V :L^{2} \longrightarrow L^{2}$ be an analytic operator valued function. As mentioned before,  we see that $K(z)$ is a compact operator for all $z\in \mathbb{C}\setminus \R_+$, and that   $\lbrace Id+K(z), z \in \C\setminus \R_+\rbrace$ is a holomorphic family of Fredholm operators (\cite{dyatlovmathematical}, Annexe C.2).  By (\ref{Lim ResLibre}), the latter can be continuously extended to a  family of operators in $\B(L^{2,-s})$ for $1/2<s<\rho -1/2$ in the two closed half-planes $\Bar{\C}_{\pm}$. Therefore,
applying analytic Fredholm theory with respect to $z$, it follows that $(Id+R_0(z)V)^{-1}$ is a meromorphic operator valued function in $\C\setminus \R_+$ with values in $\B(L^{2,-s})$, whose poles are discreet eigenvalues  of $H$ in $\C\setminus \R_+$. Moreover, for $\lambda>0$, the  limits
\begin{equation*}
 \lim\limits_{\epsilon \to 0^+} (Id+R_0(\lambda\pm i\epsilon)V)^{-1}= (Id+R_0^{\pm}(\lambda)V)^{-1},  
\end{equation*}
exist in $\b(0,-s,0,-s)$, for every $1/2<s<\rho -1/2$, if and only if $Id+R_0^{\pm}(\lambda)V$ is one to one. In other terms, the above limits do not exist if there exists a non trivial solution $\psi\in \H^{1,-s}$, $\forall s>1/2$, of $R_0(\lambda\pm i0)V g= -g$. And, it
can be easily proved that $\psi\in \H^{1,-s}$,
$\forall s>1/2$, is a solution of  $R_0(\lambda\pm i0)V g= -g$ if and only if $ (H-\lambda)\psi =0$ and $\psi$ satisfies the radiation condition (\ref{Behavior:psi}). Similarly, in view of (\ref{ResLibre:dev}), we see that 
\begin{equation*}
 \lim\limits_{z\in \C\setminus\R_+, z\to 0} (Id+R_0(z)V)^{-1}= (Id+G_0V)^{-1},  
\end{equation*}
exists in $\b(0,-s,0,-s)$, for $1/2<s<\rho -1/2$, if and only if $Id+G_0V$ is one to one.
In the following we shall use the notations $K^+(\lambda):= R_0^+(\lambda)V$ and $K_0:= G_0V$. \\

\begin{defi}\label{def:reson}
A positive number $\lambda_0> 0$ is called an  \textbf{outgoing positive  resonance} of $H$  if  $-1\in\sigma(K^{+}(\lambda_0))$ and it is called an \textbf{ incoming positive resonance}
if $-1 \in \sigma(K^{-}(\lambda_0))$. Moreover, $zero$ is said to be a \textbf{resonance} of $H$ if  Ker$_{L^{2,-s}}(Id+K_0)
/\text{Ker}_{L^2}(Id+K_0)\neq \{0\}$, $\forall s>1/2$. Let $\sigma_r^+(H)$ denotes the set of all outgoing positive  resonances.
\end{defi}

This work is mainly concerned with the singularities of the resolvent at zero and at outgoing  positive resonances.\\

Note that zero may be an embedded eigenvalue or/and a resonance of the non-selfadjoint operator $H$ if the decay condition $\rho>2$ is satisfied. In the selfadjoint case, it is known that the resonance at zero, if it occurs, is simple, i.e.  dim Ker$_{\h^{1,-s}}H / \text{Ker}_{L^2} H= 1$, $\forall s>1/2$ (cf. \cite[Theorem 3.6]{jensen1979spectral}). Similarly, in the non-selfadjoint case we can show that zero resonance is geometrically simple. Indeed, let
$s=1/2+ \epsilon$, $\rho = 2+ \epsilon_0$, $ 0<\epsilon< \epsilon_0$ and $\psi \in \h^{1,-s}$ such that 
$ (Id+G_0V)\psi = 0.$ Then $\psi(x)$ behaves as $|x| \to +\infty$ like 
\begin{equation}\label{Charac:Psi}
 \psi(x) \sim \frac{C}{\vert x \vert} + \frac{1}{\vert x \vert^{1+\epsilon_0 -\epsilon}}\phi(x)\ \text{with}\  C = \frac{-1}{4\pi}\int_{\mathbb{R}^3} V(y) \psi (y) dy,
 \end{equation}
where  $\phi$ is some bounded function on $ \lbrace x\in \R^3,\ \vert x
\vert > 1 \rbrace$. The same argument must be reiterated as necessary in showing that
\begin{equation}\label{Char:resonState}
\int_{\R^3} V(y)\psi (y) dy= 0 \text{ if and only if }
\psi \in  L^{2}.
\end{equation}
Furthermore, zero is an eigenvalue of $H$ if and only if $-1$ is an eigenvalue of the compact operator $K_0$ on $L^{2,-s}$ and the associated eigenfunctions belong to the orthogonal space of $1$ defined by $\{ \psi\in L^{2,-s}; \inp{\psi,\overline{ V1}}  =0\}$ (see (\ref{Char:resonState})). If this occurs, then their associated eigenspaces coincide. In particular, they have the same geometric multiplicity, i.e.
$\text{dim Ker}_{L^2}(H)= \text{dim Ker}_{L^2}(Id+K_0)$. 
  Let $\Pi_1:L^{2,-s}\to L^{2,-s}$ (resp., $\Pi_1^{\lambda}:L^{2,-s}\to L^{2,-s}$) be the well defined Riesz projection associated with the eigenvalue $-1$ of the compact operator $K_0$ (resp., $K^+(\lambda)$) (see (\ref{Riesz})). Noting that the Riesz projection corresponding to the embedded eigenvalue 0 of $H$ cannot be defined. 
We denote $m:=\text{rank} \, \Pi_1$ the algebraic multiplicity of $-1$ as eigenvalue of $K_0$.\\

Our study covers all the situations of zero energy. We will use the following terminology introduced in \cite{newton1977noncentral}: If zero is a resonance and not an embedded eigenvalue of $H$ zero is said to be a singularity for $H$ of  the \textbf{first kind}. If zero is an embedded eigenvalue and not a resonance of $H$ zero is said to be a singularity for $H$ of  the \textbf{second kind}. 
Finally, if zero is both an embedded eigenvalue  and a resonance of $H$ zero is said to be a singularity for $H$ of the \textbf{third kind}.
In the last two cases, we look at the eigenvalue $-1$ of $K_0$ of geometric multiplicity $k\geq 2$. Since algebraic and geometric eigenspaces need not be equal, generalized eigenvectors occur in $\Ran \Pi_1$.  We decompose $\Ran \Pi_1$ to $k$ invariant subspaces of $K_0$, $E_1,\cdots, E_k$. The subspaces $E_i$, $1\leq i \leq k$, are spanned by  Jordan chains of length $m_i$, i.e 
\[ E_i= \text{Span}\{ u_r^{(i)}= (Id+K_0)^{m_i-r}u_{m_i}^{(i)},\ 1\leq r\leq m_i\}\] 
for some vectors $u_{m_i}^{(i)}\in\text{Ker}(Id+K_0)^{m_i}\setminus
\text{Ker}(Id+K_0)^{m_i-1},$ such that 
dim Ker$(Id+K_0)|_{E_i}=1$. This decomposition yields the Jordan canonical form of the matrix $\Pi_1(Id+K_0)\Pi_1$ given in (\ref{J}). See Lemma \ref{Lem:decompE}.

\subsection{Hypotheses} Our first two hypotheses are about zero energy:
\\
\textbf{Hypothesis (H1)}: If  zero is a singularity for $H$ of the second kind with geometric multiplicity $k\in \bN^*$, one assumes that there exists a basis $\lbrace \phi_1, \cdots,\phi_k\rbrace $ of Ker$_{L^2}(Id+K_0)$ such that
\begin{equation}\label{Hyp:H1}
\det ( \langle{\phi_j, J\phi_i}\rangle)_{1\leq i,j\leq k}\neq 0,
\end{equation}
where $J: w\mapsto \overline{w}$ is the complex conjugation.
\\
\\
\textbf{Hypothesis (H2)}: If  zero is both an eigenvalue with geometric multiplicity $k\in \bN^*$ and a resonance of $H$, one assumes that:
\begin{enumerate}
\item There exists $1\leq i_0 \leq k+1$ such that
\begin{eqnarray}
\text{Ker}(Id+K_0)|_{E_{i_0}}&=& \text{Ker}_{L^{2,-s}}(Id+K_0)/
\text{Ker}_{L^{2}}(Id+K_0)\q \text{ and} \nonumber\\
\text{Ker}(Id+K_0)|_{E_{i}}&\subset & \text{Ker}_{L^{2}}(Id+K_0), \ \forall 1\leq i \leq k+1, i\neq i_0. \nonumber
\end{eqnarray}
\item There exists a basis $\{\phi_1,\cdots,\phi_{k}\}$ of  Ker$_{L^2}(Id+K_0)$ verifying the condition 
in (\ref{Hyp:H1}).\\
\end{enumerate}

Before stating the  hypothesis on positive
resonances we define for $\lambda >0$ the symmetric bilinear form  $B_{\lambda}(\cdot,\cdot)$ on $\h^{-1,s}\times \h^{-1,s}$ by
\begin{align}\label{BilinearForm:B}
B_{\lambda} (u, w)&= \int_{\mathbb{R}^3\times \mathbb{R}^3}  e^{i\sqrt{\lambda}|x-y|}u(x) V(x) w(y) V(y)\, dx \,dy.
\end{align}
\\
\textbf{Hypothesis (H3):} One assumes that $H$ has a finite number of outgoing positive resonances, i.e. $\sigma_r^+(H)=\lbrace \lambda_1,\cdots,\lambda_N\rbrace$. In addition, one supposes that for each $\lambda_j \in \sigma_r^+(H)$, there exist $N_j\in \mathbb{N}^*$ and
a basis $\lbrace \psi_1^{(j)},\cdots,\psi_{N_j}^{(j)}\rbrace$ in  $L^{2,-s}$  of $\text{Ker} \left(Id+K^+(\lambda_j)\right)$ such that 
\begin{equation}\label{cond:detReson}
 \det \left( B_{\lambda_j}(\psi_r^{(j)},\psi_l^{(j)})
\right)_{1\leq l,r\leq N_j} \neq 0.
\end{equation}

In Section \ref{sub:cas1} (resp. Section \ref{section:real resonance}) we find some numerical function $d(z)$ such that $(Id+K(z))$ is invertible if and only if 
$ d(z)\neq 0$ and if the condition (\ref{Hyp:H1}) (resp., (\ref{cond:detReson}))  is satisfied, then $-1$ is an eigenvalue  of  $K_0$ on $L^2$ (resp. $K^+(\lambda_0)$ on $L^{2,-s}$) with geometric multiplicity $k$ if and only if $0$ (resp. $\lambda_0$) is a zero of $d(z)$ with multiplicity $k$. In \cite{wang2017gevrey} the condition (\ref{Hyp:H1}) was used in the case when $-1$ is a semi-simple eigenvalue of $K_0$ (i.e. geometric and algebraic multiplicities are equal) to expand the function $d(z)$ in power of $z^{1/2}$ near $0$. However, we use in the present work  conditions (\ref{Hyp:H1}) and  (\ref{cond:detReson}) to compute in addition exactly  the leading terms of  resolvent expansions near $0$ and positive resonances (see proofs of Theorem \ref{thm:cas1} and Theorem \ref{Thm:realRes1}). 

In particular, under conditions 
(\ref{cond:detReson}) and (\ref{V}) for $\rho>N+1$ with $N\in \mathbb{N}^*$, we can expand the function $d(z)$ as follows
\begin{equation}\label{cond:d(z)}
 d(z)= \omega_{N_0} (z-\lambda_0)^{N_0}+ \omega_{N_0+1}(z-\lambda_0)^{N_0+1} +
 \cdots +\mathcal{O}( |z-\lambda_0|^{N_0+N})
 \end{equation}
for $z\in \C_+, |z-\la_0|<\delta$, with some $\omega_{N_0}\neq 0$, where $N_0$ is given in $(H3)$. Note that the expansion (\ref{cond:d(z)}) can hold under the analyticity condition on the potential $V$. See \cite[Remark 6.1 ]{wang2017gevrey}.\\ 

We mention that we find in \cite[Remark 5.4]{wang2012time} an example of a resonance state $\psi$ associated with   an outgoing resonance $\lambda_0$ for a perturbation of $-\Delta$ by a compactly supported complex-valued potential $V$, where it can be checked that $\psi$  satisfies
\[ \int_{\R^3\times \R^3 } e^{i\sqrt{\la_0}|x-y|} V(y)\psi(y)V(x)\psi(x)\,dxdy\neq 0. 
\]
 However, in the general case it is not clear if the condition (\ref{cond:detReson}) can be satisfied.  In section \ref{section:real resonance} we will study the resolvent expansion near $\lambda_0$ with the more general condition (\ref{cond:d(z)}).
\subsection{Main results} As first result, we establish 
asymptotic expansions for  $R(z)$ near zero and positive resonances. For $\delta>0$ small, we denote 
 \begin{equation}\label{omega delta}
\Omega_{\delta}:=\lbrace z\in
\C\setminus \R_+: \ |z|< \delta \rbrace.
 \end{equation}


\begin{thm}\label{thm:cas2}
 Suppose that $zero$ is a singularity for $H$ of the \textbf{first kind}.\\
Assume $\rho > 2\ell-1$, $\ell \in \bN$ with  $\ell\geq 2$. Let $ s >\ell -1/2$  and $z\in \Omega_{\delta}$. The expansion of $R(z)$ in $\b(-1,s,1,-s)$  has the following form:
\begin{equation}\label{R(z):cas2}
 R(z) =  \sum\limits_{j=-1}^{\ell-2} z^{j/2} R_j^{(1)} + \widetilde{R}_{\ell-2}^{(1)}(z), 
 \end{equation}
where   
$$R_{-1}^{(1)}: L^{2, s}\rightarrow L^{2,-s}, \ u \mapsto
i \langle u, J \phi \rangle  \phi,$$
 with $\phi$ is a resonance state normalized by 
 \begin{equation}\label{cond:phi}
     \frac{1}{2\sqrt{\pi}} \int_{\R^3} V(x) \phi(x) \, dx = 1. 
 \end{equation}
Furthermore, the remainder term $\widetilde{R}_{\ell-2}^{(2)}(z)$ is a $\vC^{\ell-2}$ operator-valued function of $z$ from 
$\Omega_{\delta}$ to $\b(-1,s,1,-s)$ and for $0<\lambda<\delta$
the limits
\begin{equation}\label{Rest:R(z)1+}
\lim_{\epsilon\to 0^+} \widetilde{R}_{\ell-2}^{(1)}(\lambda \pm i\epsilon):= \widetilde{R}_{\ell-2}^{(1)}(\lambda\pm i0)    
\end{equation}
exist in $\b(-1,s,1,-s)$ and satisfy
\begin{equation}\label{Rest:R(z)2}
 \Vert \frac{d^r}{d\lambda^r} \widetilde{R}_{\ell-2}^{(1)}(\lambda\pm i0) \Vert_{\b(-1,s,1,-s)} = o( |\lambda|^{\frac{\ell-2}{2}-r}),\ \forall\lambda \in ]0,\delta[,\ r=0,1,\cdots,\ell-2.    
 \end{equation}
If $\rho>2$  and $s>1/2$, we can obtain
$\displaystyle R(z) =z^{-1/2} R_{-1}^{(1)} + o(|z|^{-1/2})$.
 \end{thm}
In Theorem \ref{thm:cas2}, we have not used any implicit assumption 
on zero resonance. In particular, we do not know if zero resonance is algebraically simple. 

\begin{thm}\label{thm:cas1}
Suppose that $zero$ is a singularity for $H$ of the \textbf{second kind} and that $(H1)$ holds. Let $k\in \bN^*$ be the geometric multiplicity of the eigenvalue zero.  Assume $\rho>2\ell-3$, $\ell\in \bN$ with $\ell\geq 4.$ Let  $s>\ell-3/2$  and $z\in \Omega_{\delta}$,  $\delta>0$ small. We have
\begin{equation}\label{R(z):cas1}
R(z)=  \sum\limits_{j=-2}^{\ell-4}
z^{j/2} R_j^{(2)} + \widetilde{R}_{\ell-4}^{(2)}(z), 
 \end{equation}
as operators in $\b(-1,s,1,-s)$. Here   $R_{-2}^{(2)}=-\mathcal{P}_0^{(2)}, R_{-1}^{(2)}:L^2\to \text{Ker} (Id+K_0)$,
$$\mathcal{P}_0^{(2)}=\sum_{j=1}^k\inp{\cdot, J\mathcal{Z}_j^{(2)}}\mathcal{Z}_j^{(2)} \ \text{with} \
 \inp{\mathcal{Z}_i^{(2)}, J \mathcal{Z}_j^{(2)}}= \delta_{ij},
\ \forall 1\leq i,j\leq k,$$
where
$\{ \vZ_1^{(2)}, \cdots, \vZ_k^{(2)} \}$ is a basis of $\text{Ker} (Id+K_0).$
Moreover, the remainder term $\widetilde{R}_{\ell-4}^{(2)}(z)$ is a $\vC^{\ell -
4}$ operator-valued function of $z$ from 
$\Omega_{\delta}$ to $\b(-1,s,1,-s)$ and for $0<\lambda<\delta$ the limits $\widetilde{R}_{\ell-4}^{(2)}(\lambda\pm i0)$  (see (\ref{Rest:R(z)1+}))
exist in $\b(-1,s,1,-s)$ and satisfy  \begin{equation}\label{Rest:R(z)1} 
\Vert \frac{d^r}{d\lambda^r} \widetilde{R}_{\ell-4}^{(2)}(\lambda\pm i0) \Vert_{\b(-1,s,1,-s)} = o(|\lambda|^{\frac{\ell}{2}-2-r}),\ \forall \lambda \in  ]0,\delta[,\ r=0,1,\cdots,\ell-4.
\end{equation}
If $\rho>4$ and $s>3/2$ we can obtain $R(z)=z^{-1} R_{-2}^{(2)}+z^{-1/2}R_{-1}^{(2)} +o(|z|^{-1/2})$ and if $\rho>3$ and $s>1/2$,  $\displaystyle R(z) = z^{-1} R_{-2}^{(2)} +o(|z|^{-1})$.  
\end{thm}
Note that although there is no  natural spectral projection associated with the embedded eigenvalue 0 of $H$, our result shows that the leading term is still given by spectral projection $\vP_0^{(2)}$.


\begin{thm}\label{thm:cas3}
Suppose that $zero$ is a singularity for $H$ of the \textbf{third kind} and $(H2)$ holds.
Assume $\rho>2\ell-3, \ell \in \bN$ with $\ell\geq 4$. Then for  $s>\ell-3/2, $ and $z \in \Omega_{\delta}$, we have
\begin{equation}\label{R(z):cas3}
 R(z) =  \sum\limits_{j=-2}^{\ell-4} z^{j/2} R_j^{(3)} + \widetilde{R}_{\ell-4}^{(3)}(z), 
\end{equation}
in $\b(-1,s,1,-s)$. Here $R_{-2}^{(3)}=-\mathcal{P}_0^{(3)},
R_{-1}^{(3)}=i \inp{\cdot, J\psi}\psi + S_{-1}^{(3)},$
$$ \mathcal{P}_0^{(3)}, S_{-1}^{(3)} :L^{2}\to \text{Ker}_{L^2}(Id+K_0), \ $$ 
$$\mathcal{P}_0^{(3)}= \sum_{j=2}^{k}\inp{\cdot, J\mathcal{Z}_j^{(3)}}\mathcal{Z}_j^{(3)} \ \text{with}\ \inp{\mathcal{Z}_i^{(3)}, J \mathcal{Z}_j^{(3)}}=\delta_{ij}, \ \forall 2\leq i,j\leq k,$$
where $\{ \vZ_2^{(3)}, \cdots,\vZ_{k}^{(3)} \}$ is a basis of 
Ker$_{L^2}(Id+K_0)$ and $\psi$ is a resonance state satisfying (\ref{cond:phi}). In addition, the remainder term $\widetilde{R}_{\ell-4}^{(3)}(z)$ has the same properties as $\widetilde{R}_{\ell-4}^{(1)}(z)$ in Theorem \ref{thm:cas1}. \\
\end{thm}
 The following theorem gives the resolvent expansion  near an outgoing positive resonance $\lambda_0$ for $z$ in a set $\Omega_{\delta}^+$ given by 
 \begin{equation}
 \Omega_{\delta}^+:= \lbrace z\in \C_+:\ |z-\lambda_0| <\delta \rbrace.\label{omega+}
 \end{equation}
\begin{thm}\label{Thm:realRes1}
Suppose that $(H3)$ holds. Let $\lambda_0\in \sigma_r^+(H)$. Assume $\rho>2\ell-1, \ell\in \bN$ with $\ell\geq 2$. Then, for $s>\ell-1/2$ and $ z \in \Omega_{\delta}^+$, we have
\begin{equation}\label{dev:realRes2}
 R(z) = \frac{\mathcal{P}(\lambda_0)}{z-\lambda_0} + \sum\limits_{j=0}^{\ell-2} (z-\lambda_0)^j R_j(\lambda_0) + \widetilde{R}_{\ell-2}(z-\lambda_0),
 \end{equation}
in $\b(-1,s,1,-s)$, where 
$$ \mathcal{P}(\lambda_0)=  \sum_{j=1}^{N_0} \langle
\cdot , J \psi_j^{(\lambda_0)}\rangle 
\psi_j^{(\lambda_0)} \ \text{with}\ \frac{1}{i8\pi \sqrt{\lambda_0}} B_{\lambda_0}(\psi_i^{(\lambda_0)},
\psi_j^{(\lambda_0)})=\delta_{ij},$$
such that  $\{\psi_1^{(\lambda_0)}, \cdots, \psi_{N_0}^{(\lambda_0)}\}$
is a basis of Ker$(Id+R_0^+(\lambda_0)V)$, and 
$B_{\lambda_0}$ is the bilinear form defined in (\ref{BilinearForm:B}).
The remainder term $\widetilde{R}_{\ell-2}(z-\lambda_0)$ is analytic in $\Omega_{\delta}^+$ and for $\la>0$ with $|\la-\la_0|<\delta$ the limit
\begin{equation*}
\lim_{\epsilon\to 0^+} \widetilde{R}_{\ell-2}(\lambda-\lambda_0+i\epsilon) = \widetilde{R}_{\ell-2}(\lambda-\lambda_0+i0)
\end{equation*}
exists in the norm of  $\b(-1,s,1,-s)$ and  satisfies  
\begin{equation}\label{Rest:R(z)res}
\Vert \frac{d^r}{d\lambda^r} \widetilde{R}_{\ell-2}(\lambda-\lambda_0+i0)\Vert_{\b(-1,s,1,-s)} = o(|\lambda-\lambda_0|^{\ell-2-r}), \ r=0,1,\cdots,\ell-2.    
\end{equation}
If $\rho>2$ and $s>1/2$, we can obtain $\displaystyle R(z)=\frac{R_{-1}(\lambda_0)}{z-\lambda_0} +
o(|z-\la_0|^{-1})$.
\end{thm}

Using the preceding results, we show that  under the assumption $\rho>2$, $H$ has at most a finite number of discrete eigenvalues located in the closed upper half-plane. However, if zero is an eigenvalue of $H$ we need the stronger assumption $\rho>3$. See Proposition \ref{vp fini}.
\\

We obtain asymptotic expansions in time of the strongly continuous Schr\"odinger semigroup $(e^{-itH})_{t\geq 0}$,  as $t\to +\infty$, if zero is a resonance or an eigenvalue of $H$ taking into account the presence of outgoing positive resonances. Our main result is the following:

 \begin{thm}\label{thm:estim}
Assume that $(H3)$ holds.
\begin{enumerate}
\item[$(a)$] Suppose that zero is a singularity for $H$ of the \textbf{first kind}.
If $\rho>5$ then for  $s> 5/2$ we have 
\begin{align}
 e^{-itH} -\sum_{j=1}^p e^{-itH} \Pi_{z_j} + \sum_{j=1}^N e^{-it\lambda_i} R_{-1}(\lambda_j)& =\label{dev:solution} \\
(i\pi)^{-1/2}\langle \cdot,\phi\rangle \phi\, t^{-\frac{1}{2}} &+ o(t^{-1/2}),\qquad  t\to +\infty,\nonumber
\end{align}
in $\b(0,s,0-s)$, where $\phi$ (resp. $R_{-1}(\lambda_j)$) is given by Theorem \ref{thm:cas2} (resp. Theorem \ref{Thm:realRes1}).
 \\ 
%
%
\item[(b)]Assume $\rho>7$. Suppose that zero  is a singularity for $H$ of the \textbf{second kind} and that $(H1)$ holds.
 Then, for $s>7/2$  the expansion at the right hand side of (\ref{dev:solution}) has the following form
\begin{align}
\mathcal{P}_{0}^{(2)}&  -i(i\pi)^{-1/2}  R_{-1}^{(2)} t^{-\frac{1}{2}}
 - (4i\pi)^{-1/2} R_1^{(2)} t^{-\frac{3}{2}} + o(t^{-3/2}),\nonumber
\end{align}
where $\mathcal{P}_{0}^{(2)}$ and $R_{s}^{(2)}$ for $s=-1,1$ are  given by Theorem \ref{thm:cas1}. \label{thm:estim1}
\\
If $\rho>5$ and $s>5/2$, then (\ref{dev:solution}) holds with the right hand side replaced by $$\mathcal{P}_{0}^{(2)}  -i(i\pi)^{-1/2}  R_{-1}^{(2)} t^{-\frac{1}{2}} + o(t^{-1/2}).$$
In the above expansions 
$\displaystyle\Pi_{z_j}$ denote the Riesz projections associated with the discrete eigenvalues $\displaystyle z_j$ located in the closed upper half-plane.\label{thm:estimReson}
\end{enumerate}
\end{thm}
Note that $(a)$ holds if 
zero is both an eigenvalue and a resonance of $H$.
To obtain the above expansions we find some curve $\Gamma^{\nu}(\eta)$, for some $\eta,\nu>0$ small,  which does not intersect the real axis at zero or at points in $\sigma_r^+(H)$, such that above this curve, $H$ has a finite number of eigenvalues (see Figure 1. in Section \ref{sec:Representation}). Then the expansions are 
deduced by  representing $e^{-itH}$ as a sum of some residue terms and a Dunford integral of $R(z)$ on $\Gamma^{\nu}(\eta)$. \begin{rem}
 With regard to the case zero is a regular point for $H$, i.e it is not an eigenvalue nor a resonance of $H$, we can obtain the same result obtained in \cite[Theorem 6.1]{jensen1979spectral} for $H=-\Delta+V$ with real $V$. If $\rho>3$ and $s>3/2$, then for $z\in \Omega_{\delta}$ with $\delta>0$ small, we have
 \begin{equation}\label{R(z):regular case}
 R(z)=  R_0^{(0)} + z^{1/2} R_1^{(0)} + o(|z|^{1/2}),    
 \end{equation}
in $\b(-1,s,1,-s)$,  where 
 \begin{align*}
R_0^{(0)}&= (Id+G_0V)^{-1},\\
R_0^{(1)}&= i(Id+G_0V)^{-1}G_1(I-V(Id+G_0V)^{-1} G_0 ).
 \end{align*}
Here, $G_0, G_1$ are given in (\ref{ResLibre:dev}) and $\Omega_{\delta}$ is defined in (\ref{omega delta}). The proof of \cite[Theorem 6.1]{jensen1979spectral} can be done here because it is based on the expansion of $(I+R_0(z)V)^{-1}$ in a Neumann series which works also for non-real $V$. 
However, in  presence of positive resonances of $H$ a stronger assumption on $\rho$ is needed to obtain the expansion in time of $e^{-itH}$ with reminder $o(t^{-3/2})$ even if zero is a regular point for $H$. We obtain the following result:\\ 
Suppose that zero is a regular point for $H$ and $(H3)$ holds. If $\rho>7$ and $s>7/2$, we have 
\begin{equation*}
 e^{-itH} -\sum_{j=1}^p e^{-itH} \Pi_{z_j} + \sum_{j=1}^N e^{-it\lambda_i} R_{-1}(\lambda_j) =
 - (4i\pi)^{-1/2} R_1^{(0)} t^{-\frac{3}{2}} + o(t^{-3/2}),\quad  t\to +\infty,
\end{equation*}
in $\b(0,s,0,-s)$. (See (\ref{dev:solution})).
\end{rem}
\section{Resolvent expansions  at low energies}\label{sec:zero resonance}
Consider $H_0=-\Delta$ and the perturbed non-selfadjoint operator $H= -\Delta +V$.  
In the following, we always assume that $V$ satisfies (\ref{V}), where a stronger assumption on $\rho$ is needed for a high-order asymptotic expansion in $z^{1/2}$ of the resolvent. In this section we will use some tools developed in \cite[Section 5.4]{wang2017gevrey}.
\\
\paragraph{\textbf{Riesz projection}}
Set $E=\text{Ran}\ \Pi_1$,  where $\Pi_1$ is the Riesz projection associated with the eigenvalue $-1$ of $K_0$ on $L^{2,-s}$  for $1/2<s<\rho-1/2$ (see Section \ref{sub:model}).\\
Let $J:f\rightarrow \overline{f}$ be the operation of complex conjugation. Then we have $V^*= JVJ$, $H^{\ast} =JHJ$ and the following relations
\begin{eqnarray}
JV\Pi_1 = \Pi_1^* J V,&& \Pi_1 JG_0 = JG_0 \Pi_1^*,\\ 
JV(Id+K_0) &=& (Id+ K_0^{\ast})JV,\\
(Id+K_0) JG_0 &=& JG_0 (Id+K_0^*).
\end{eqnarray}
It follows that 
 $JV (\text{resp.}\, JG_0) :\text{Ran}\, \Pi_{1} (\text{resp.}\, \text{Ran}\, \Pi_{1}^{\ast}) \rightarrow \text{Ran}\, \Pi_{1}^{\ast}
 (\text{resp.}\, \text{Ran}\, \Pi_{1})$ is injective. 
 We deduce that the bilinear form  $\Theta(\cdot,\cdot)$ defined on $E$ by 
\begin{equation}\label{BilinForm}
\Theta(u,v)=\int_{\mathbb{R}^{3}} V(x)u(x)v(x) dx = \langle u, v^{\ast}\rangle 
\end{equation}
is non-degenerate on $E\times E$, where we denoted 
\begin{equation}\label{notation:u*}
u^{\ast} = JVu.
\end{equation}
See the proof of Lemma 5.13 in \cite{wang2017gevrey}.

Using the above non-degenerate bilinear form, we obtain the following decomposition lemma:
\begin{lem}\label{Lem:decompE}
 Assume that $-1$ is an eigenvalue of $K_0$ of geometric multiplicity $k\geq 1 $ and algebraic multiplicity $m$. Then there exist $k$ invariant subspaces of $K_0$ denoted by  $E_{1},\cdots, E_{k}$ such that
\begin{enumerate}
\item $E = E_{1}\oplus \cdots \oplus E_{k}$, where  $\forall i\neq j$: $E_{i}\perp_{\Theta} E_{j}$.\label{Lem:decomp1}
\item $\forall 1 \leq i\leq k$, there exists a basis $ \mathcal{U}_i:= \lbrace u_1^{(i)},\cdots, u_{m_i}^{(i)} \rbrace $ of $E_i$
such that 
\begin{align}
u_r^{(i)}:= ( Id+K_0)^{m_i-r}u^{(i)}_{m_i},\, \forall 1\leq r \leq m_i,\label{vect:u}\\
u_{m_i}^{(i)}\in \text{Ker}(Id+K_0)^{m_i} \ \mbox{and} \ \Theta(u_1^{(i)},u_{m_i}^{(i)})=c_i\neq 0. \nonumber 
\end{align} \label{Lem:assertion2}
\item $\forall 1\leq j\leq k$, there exists a basis $ \mathcal{W}_j:=\lbrace w_1^{(j)}, \cdots, w_{m_j}^{(j)} \rbrace$ of $E_j$ such that
\begin{equation}
w_r^{(j)} \in \text{Ker} (Id+K_0)^{m_j+1-r}, \ and \  \Theta(u_{\ell}^{(i)}, w_r^{(j)}) = \delta_{ij} \delta_{\ell r}.\label{vect:w}
\end{equation}
\label{Lem:decomp3}
\label{Lem:assertion3}

\item \text{dim ker}$(Id+K_0)|_{ E_{j}}=1$, $\forall j=1,\cdots,k$.\label{Lem:decomp2}
\end{enumerate}
Moreover, the matrix of $\Pi_{1}(Id+K_0)\Pi_{1}$ in the basis $\mathcal{U}:= \bigcup_{i=1}^k \mathcal{U}_i$ of $E$ is a $m\times m$ block diagonal  matrix of the following form
\begin{equation}\label{J}
J_m=diag (
J_{m_{1}}, J_{m_{2}}, \cdots,J_{m_{k}}),
\end{equation} 
where 
\begin{equation*} 
J_{m_{j}}=\begin{pmatrix}
0&1&0&\cdots& 0\\
0&0&1&\ddots & \vdots  \\
0&0&\ddots&\ddots&0\\
\vdots &\vdots &\ddots &\ddots & 1\\
0&0&\cdots&0 &0\\
\end{pmatrix}_{m_j \times m_j}
\end{equation*} 
is a Jordan block. We have also denoted  $m_{j}=\text{dim}\ E_{j}$ for $j=1,\cdots,k$, such that $m=m_{1}+ \cdots + m_{k}$.  
\end{lem}
The following statement is an immediate consequence of the previous lemma.
 \begin{cor}\label{cor:Riesz}
 The Riesz projection $\Pi_1$ has the following representation:
 \begin{equation*}
  \Pi_1= \sum\limits_{j=1}^k \sum\limits_{r=1}^{m_j} \inp{ \cdot, (w_r^{(j)})^{\ast}} u_r^{(j)}.    
 \end{equation*}
 \end{cor}
See \cite[Corollary 5.16]{wang2017gevrey} for the proof of the corollary in the case $k=1$. We now prove Lemma \ref{Lem:decompE}.
\begin{proof}
We proceed by induction on $k=$ dim Ker$(Id+K_0)$. Initially, for $k=1$, we refer to \cite[Section 5.4]{wang2017gevrey} for the case of geometrically simple eigenvalue  $-1$ of $K_0$. For $k=2$, we shall show the parts (\ref{Lem:decomp1}) and (\ref{Lem:assertion2}) as follows. 
First, note that $K_0$ is $\Theta$-symmetric which implies that, if $F$ is a stable subspace of $Id +K_0$ its $\Theta$-orthogonal 
$F^{\perp_{\Theta}}$ is stable and if in addition $\Theta|_{F\times F}$ is non degenerate then $E=F\oplus F^{\perp_{\Theta}}$ and $\Theta|_{F^{\perp_{\Theta}}\times F^{\perp_{\Theta}}}$ is non degenerate. Second, If $E(u_{\mu})= \text{Span}(u_k= (Id+K_0)^{{\mu}-k}u_{\mu}, 1\leq k\leq \mu)$, with 
$u_{\mu}\in$ Ker$(Id+K_0)^{\mu}$ and $\Theta(u_{\mu},u_1)=c\neq 0$, then $u_{\mu}\notin$Ker$(Id+K_0)^{\mu-1}$,
$\{u_1,\cdots,u_{\mu}\}$ is a basis of $E(u_{\mu})$, $\Theta|_{
E(u_{\mu})\times E(u_{\mu})}$ is non degenerate, $E=E(u_{\mu})\oplus E(u_{\mu})^{\perp_{\Theta}}$ and  $\Theta|_{
E(u_{\mu})^{\perp_{\Theta}}\times E(u_{\mu})^{\perp_{\Theta}}}$ is non degenerate. In addition, $\{u_1,\cdots,u_{\mu}\}$ has  $\Theta$-dual basis
$\{ w_1,\cdots,w_{\mu}\}$ with $w_{\mu}=c^{-1}u_1$ which can be constructed as in \cite[Lemma 5.15]{wang2017gevrey} since $\Theta|_{ E(u_{\mu})\times E(u_{\mu})}$ is non degenerate.
Therefore it suffices to find one $u_{m_1}$ such that the second statement holds and to find in addition  one $u_{m_2}$ such that $E(u_{m_1})^{\perp_{\Theta}}= E(u_{m_2})$.

\begin{itemize}
\item[$(i)$] Let $m_1$, $1\leq m_1\leq m$, be the smallest integer such that $E=$Ker$(Id+K_0)^{m_1}$. If $m_1=1$, $E=$Ker$(Id+K_0)$. By non degeneracy of $\Theta$, we can easily find $u_1, v_1\in E$ such that $\Theta(u_1,u_1)\neq 0,\Theta(v_1,v_1)\neq 0$ and $\Theta(u_1,v_1)= 0 $. This solves the problem for $m_1=1$.  If $m_1>1$,
set $Q=(Id+K_0)^{m_1-1}$ which is $\Theta$-symmetric. The bilinear form $B(u,v)=\Theta(Qu,v)$ is symmetric on 
$E$ and not identically zero. Otherwise $\Theta(w,v)=0$ for all $v\in E$ and $w\in $ Ran $Q$ and because $\Theta$ is non degenerate it follows that $E=($Ran $Q)^{\perp}=$Ker $Q$ ( it can be seen that 
Ker $Q\subset$ (Ran $Q)^{\perp}$ and the equality comes from the non degenracy of $\Theta$). This contradicts the definition of $\mu$. By polarization identity a symmetric bilinear form $B$ on $E$ is the null bilinear form if and only if $B(u,u)=0$ for all $u\in E$.  
Hence there exists $u_{m_1}\in E$ such that 
$\Theta ((Id+K_0)^{m_1-1}u_{m_1},u_{m_1})\neq 0$.

\item[$(ii)$] Consider the restrictions of $(Id+K_0)$ and $\Theta$ to 
 $E(u_{m_1})^{\perp_{\Theta}}$. Let $m_2= m-m_1$. Since $\text{dim Ker}((Id+K_0)|_{E(u_{m_1})^{\perp_{\Theta}}}) =1$, then $m_2$ is the smallest integer such that $E(u_{m_1})^{\perp_{\Theta}}=$Ker$((Id+K_0)|_{E(u_{m_1})^{\perp_{\Theta}}})^{m_2}$, so following $(i)$ one finds $u_{m_2}$. 
\end{itemize}

Assume now that (\ref{Lem:decomp1})  and (\ref{Lem:assertion2}) are true for $k=\ell-1$, $\ell\in \bN$, $\ell\geq 2$. For $k=\ell$, 
$E=E(u_{m_1})\oplus E(u_{m_1})^{\perp_{\Theta}}$ with 
dim Ker$(Id+K_0)|_{E(u_{m_1})^{\perp_{\Theta}}} =\ell-1$ by $(i)$. Applying the inductive hypothesis on  $E(u_{m_1})^{\perp_{\Theta}}$ we prove $(1)$ and $(2)$ for $k=\ell \in \bN^*$.\\

  Finally, for the statement about the basis $ \vW$ once the basis $\vU=\bigcup_{j=1}^k \vU_j$ is constructed, $\vW_j$, $j=1,\cdots,k$, are given by \cite[Lemma 5.15]{wang2017gevrey}. We have established (\ref{Lem:decomp1})-(\ref{Lem:decomp3}) and (\ref{Lem:decomp2}) follows directly.
\end{proof}
\begin{rem}\label{Rem:c_j}
By construction of $E_j$ in  the previous lemma, we have
$$\displaystyle \Theta(u_{1}^{(j)},u_{m_j}^{(j)})=c_j\neq 0, \forall 1\leq j\leq k,$$
and $\Theta|_{E_j\times E_j}$ is non degenerate. Then, applying Lemma 5.15 in  \cite{wang2017gevrey} to construct the basis $\displaystyle  \{w_1^{(j)},\cdots,w_{m_j}^{(j)}\}$ of 
$\displaystyle \{u_1^{(j)},\cdots,u_{m_j}^{(j)}\}$, we take $w_{m_j}^{(j)}= c_j^{-1} u_1^{(j)}$, $\forall 1\leq j\leq k$. 
\end{rem}


In order to establish the asymptotic expansion of the resolvent $R(z)$, we use the resolvent equation given in (\ref{ResEq}). We must establish the asymptotic expansion of $(I+R_0(z)V)^{-1}$. Our approach is to use the so called Grushin's method.
To avoid repetition we will introduce a Grushin problem in the more general case when dim Ker$_{L^{2,-s}}(Id+K_0)=k$, $k\in \bN^*$.  Set $\mathcal{M}(z)= Id+K(z)$. Given the decomposition of $E$ in Lemma \ref{Lem:decompE}, we can identify $E_j$ with $\mathbb{C}^{m_j}$ and 
$\mathbb{C}^m$  with $\mathbb{C}^{m_1} \oplus \cdots \oplus
 \mathbb{C}^{m_k}$ to construct the following Grushin problem.
 \\
\subsection{Grushin problem for the inverse of $(I+R_0(z)V)$}\label{sec:Gr}
We consider
\begin{equation}\label{matrix:P(z)}
\mathcal{P}(z):=\begin{pmatrix}
\mathcal{M}(z)& S \\
T & 0  \\
\end{pmatrix},
\end{equation} 
$$S:\overset{k}{\underset{j=1}{\oplus}} \mathbb{C}^{m_{j}} \longmapsto \text{Ran}\,\Pi_1 ; \ 
\zeta =\overset{k}{\underset{j=1}{\oplus}}( \zeta_{1}^{(j)},\cdots,\zeta_{m_{j}}^{(j)})\mapsto S\zeta:=\overset{k}{\underset{j=1}{\sum}}\overset{m_{j}}{\underset{i=1}{\sum}}\zeta_{i}^{(j)}
u_{i}^{(j)},$$
 $$T: \Ran \Pi_1 \longmapsto \overset{k}{\underset{j=1}{\oplus}} \mathbb{C}^{m_{j}};\
  v \mapsto Tv:= \overset{k}{\underset{j=1}{\oplus}}(\langle v,(w_{1}^{(j)})^{\ast}\rangle ,\cdots,\langle v,(w_{m_{j}}^{(j)})^{\ast}\rangle).$$
The operators $S$ and $T$ verify  $TS=I_{m}$ and $ST=\Pi_{1}$ (see Corollary \ref{cor:Riesz}), where $S$ and $T$
are chosen so that the problem $\mathcal{P}(z)$ is invertible.
Since $Id+K_0$ is injective on Ran $\Pi_{1}'$, where $\Pi_{1}' = Id- \Pi_{1}$, by the alternative Fredholm theorem  $\Pi_1^{'}(Id+K_0)\Pi_1^{'}$ is invertible  on $\text{Ran}\, \Pi_{1}'$. Then, using an argument of perturbation for
  $\delta>0$ small enough 
  $\Pi_{1}'\mathcal{M}(z)\Pi_{1}'$ is also invertible on $\text{Ran}\, \Pi_{1}'$  for all $z$ in $\Omega_{\delta}$, with inverse 
$$ E(z)=(\Pi_{1}'\mathcal{M}(z)\Pi_{1}')^{-1}\Pi_{1}'.$$
In view of (\ref{ResLibre:dev}), for $\rho>2\ell+1$, $\ell+1/2<s<\rho-\ell-1/2$, $\ell=1,2\cdots$ and $z\in \Omega_{\delta},  \delta>0$ small,  the expansion of $E(z)$ in $\b(1,-s,1,-s)$  can be written as follows:
\begin{equation}\label{E(z):dev}
E(z) =  \sum\limits_{j=0}^{\ell} z^{j/2} E_j + E_{\ell}(z), \end{equation}
where $E_0= (\Pi_{1}'( Id + K_0)\Pi_{1}')^{-1}\Pi_{1}'$,
$E_1= i E_0 G_1V \Pi_1^{'}E_0$ and other terms $E_j$, $j=2, \cdots, \ell$ can be computed directly.  Moreover, the remainder term $E_{\ell}(z)$  satisfies 
\begin{equation}\label{Rest:E(z)1}
\Vert \frac{d^r}{dz^r} E_{\ell}(z)\Vert_{\B(\h^{1,-s})} = o( |z|^{\frac{\ell}{2}-r}), \ \forall z\in \Omega_{\delta}, \ r=0,1,\cdots,\ell.
\end{equation}
In addition, for $0<\lambda<\delta$, it follows from 
(\ref{Lim ResLibre}) that the limits
\begin{equation}\label{Rest:E(z)+}
\lim_{\epsilon\to 0} E_{\ell}(\lambda\pm i\epsilon) = E_{\ell}(\lambda\pm i0)
\end{equation}
exist as operators in $\b(-1,s,1,-s)$. By taking $\epsilon\to 0^+$ in (\ref{Rest:E(z)1}) with $z=\la\pm i\epsilon$ we show that the above limits satisfy 
\begin{equation}\label{Rest:E(z)}
\Vert \frac{d^r}{d\la^r} E_{\ell}(\la\pm i0)\Vert_{\B(\h^{1,-s})} = o( |\la|^{\frac{\ell}{2}-r}),\ \forall\, 0<\la<\delta, \ r=0,1,\cdots,\ell.
\end{equation}

Then the Grushin problem is invertible with inverse
   \begin{equation}\label{Prblem:Inverse}
\begin{pmatrix}
\mathcal{M}(z)& S \\
T & 0  \\
\end{pmatrix}^{-1}=
\begin{pmatrix}
E(z) & E_{+}(z)  \\
E_{-}(z)& E_{-+}(z)  \\
\end{pmatrix}  :\  \h^{1,-s}\times \mathbb{C}^{m} \longmapsto \h^{1,-s}\times \mathbb{C}^{m},
\end{equation} 
where
\begin{align}
E_{+}(z) &= S- E(z)\mathcal{M}(z)S \nonumber \\
  E_{-}(z) &= T -T\mathcal{M}(z) E(z) \nonumber\\
  E_{-+}(z)&=-T\mathcal{M}(z)S + T \mathcal{M}(z) E(z) \mathcal{M}(z)S \label{E-+:def}
\end{align}
Therefore, $\mathcal{M}(z)$ is invertible if and only if $E_{-+}(z)$ is invertible, with
\begin{equation}\label{1+R}
\mathcal{M}(z)^{-1}=E(z)-E_{+}(z)E_{-+}(z)^{-1}E_{-}(z) \quad \text{on} \quad \h^{1,-s}. 
\end{equation}

\subsection{Zero singularity of the first kind}\label{section:cas2}
In this section, zero will be only a resonance and not an eigenvalue of  $H$, in which case $-1$ is a geometrically simple eigenvalue of the compact operator $K_0$ on $L^{2,-s}$, $\forall s>1/2$.
The same construction made in Lemma \ref{Lem:decompE} for a single subspace $E_i$ can be done for $E=\text{Ran}\, \Pi_1$ at the present case. By this lemma we find $\mathcal{U}:=\lbrace u_{1},\cdots, u_{m} \rbrace \subset L^{2,-s}$ a basis of $E$ such  that $\Theta(u_1,u_m)=c\neq 0$ and  $\mathcal{W}=\lbrace w_{1},\cdots, w_{m} \rbrace$ its $\Theta$-dual basis. In particular, $u_1\in \text{Ker}_{L^{2,-s}}(Id+K_0)$ is a resonance state.\\

Let $\delta>0$ small, we denote 
\begin{equation}\label{omega_delta}
\Omega_{\delta}=\lbrace z\in \C\setminus \R_+,\ |z|<\delta\rbrace.
\end{equation}

To calculate the singularity of $(I+K(z))^{-1}$ due to zero resonance, we must establish an asymptotic expansion of $E_{-+}(z)^{-1}$ by using the above Grushin problem with $\text{dim Ker}_{L^2,-s} (Id+K_0)=k=1$. \\ 

First, let $\rho>\ell+1$ with $\ell\in \mathbb{N}^*.$ For $z\in \Omega_{\delta}, \delta>0$, introducing the expansions in  $z^{1/2}$ of $R_0(z)$ and $E(z)$  at order $\ell$ given in (\ref{ResLibre:dev}) and (\ref{E(z):dev}) respectively, we obtain 
\begin{equation}\label{E-+:cas2}
E_{-+}(z)= N+\sqrt{z} A + \sum\limits_{j=2}^{\ell} z^{j/2} E_{-+,j} + \widetilde{E}_{-+,\ell}(z),
\end{equation}
where by using Lemma \ref{Lem:decompE} (\ref{Lem:assertion2}) and $w_{m}=c^{-1}\,u_{1}$ we have
\begin{align*}
N &:= -\big(\inp{ (Id+G_0V)u_r, JV w_{\ell}}\big)_{1\leq \ell,r\leq m}=-(\delta_{i+1 j})_{1\leq i,j\leq m},\nonumber\\
A &:= -i\Big(\inp{ G_1Vu_r, JV w_{\ell}}\Big)_{1\leq \ell,r\leq m}, \nonumber\\
E_{-+,2}&=\Big(\inp{(G_2V-G_1V E_0 G_1V)u_r, JV w_{\ell}}\Big)_{1\leq \ell,r\leq m}\\
 &+\Big( \inp{(iE_1 G_1V-E_0G_2V) u_r, JV (Id+G_0V) w_{\ell}}\Big)_{1\leq \ell,r\leq m} \nonumber.
\end{align*}
 In particular, the $(m,1)$-th entry of the matrix $A$ is non zero and is given by
\begin{align}\label{am1}
 a_{m1}& = (i c)^{-1}\langle G_{1}V u_{1},J V u_{1}\rangle
= (i4\pi\,c)^{-1} \langle u_1,JV 1\rangle^{2}.
 \end{align} 
  Moreover, the remainder $\widetilde{E}_{-+,\ell}(z)$ is a 
$\vC^{\ell}$ matrix-valued function of $z$ in $\Omega_{\delta}$ and for $0<\lambda<\delta$ the limits
\begin{equation}\label{Rest:E-+cas1+}
\lim_{\epsilon\to 0^+} \widetilde{E}_{-+,\ell}(\lambda\pm i\epsilon)= 
\widetilde{E}_{-+,\ell}(\lambda\pm i0)
\end{equation}
exist and  satisfy
\begin{equation}\label{Rest:E-+cas1}
\Vert \frac{d^r}{d\lambda^r}\widetilde{E}_{-+,\ell}(\lambda\pm i0) \Vert = o(|\lambda|^{\frac{\ell}{2}-r}), \ r=0,1,\cdots,\ell.   
\end{equation} 
Also, $E_{-+,j}, j=3,\cdots,\ell$ can be obtained explicitly but they become even more complicated. The condition $\rho>\ell +1$ is necessary to obtain the expansion of $E_{-+}(z)$ up to order
 $z^{\ell/2}$. Let us check this for some terms appearing in
  the computation of $E_{-+,\ell}$. We must obtain the terms 
  $\inp{G_{\ell} Vu_r,JVw_l}$, $\inp{(Id+G_0V)E_{\ell}
  (Id+G_0V)u_r,JVw_l}$, $\inp{(Id+G_0V)E_0G_{\ell}Vu_r,JVw_l},$ and 
$\inp{G_{\ell}V E_0(Id+G_0V)u_r,JVw_l}.$ Consider $G_{\ell} Vu_r$. 
We have $u_r\in \h^{1,-\frac{1}{2}-\epsilon}$, $0<\epsilon<1/2$, so
 $Vu_r\in \h^{-1,\rho-\frac{1}{2}-\epsilon}$ and $G_{\ell} Vu_r \in
 \h^{1,-\ell-\frac{1}{2}-\epsilon}$ because $\rho-1/2 >\ell +1/2$, where $G_{\ell}$ is given in (\ref{Gj})). Thus  $\inp{G_{\ell} Vu_r,JVw_l}$ makes sense since $Vw_l\in 
 \h^{-1,\rho-\frac{1}{2}+\epsilon}\subset \h^{-1,\ell+\frac{1}{2}
 +\epsilon}$. Consider now $(Id+G_0V)E_{\ell}(Id+G_0V)u_r$. $(Id+G_0V)u_r\in \h^{1,-\frac{1}{2}-\epsilon}$, to apply $E_{\ell}$ we must have $\rho-1/2>\ell +1/2$, and it maps to $\h^{1,-\ell-\frac{1}{2}-\epsilon}$ (see (\ref{E(z):dev})). Finally, since $(Id+G_0V)\in \B(1,-s,1,-s)$ for $1/2<s<\rho-1/2$, then  $(Id+G_0V)E_{\ell}(Id+G_0V)u_r\in \h^{1,-\ell-\frac{1}{2}-\epsilon}$  because $\ell+\frac{1}{2}<\rho-1/2 $. For the properties of the remainder term $E_{-+,\ell}$,  see the argument used for $E_{\ell}(z)$ in (\ref{E(z):dev}). 
\\

Set $\ds E_{-+,\ell}(z) = N+ \sqrt{z}A$.  Then for $\rho>2$ and $z\in  \Omega_{\delta}$, $\delta>0$, $E_{-+}(z)=E_{-+,1}(z) +o(|z|^{1/2})$. 
This yields to
\begin{equation}\label{detE:cas2 }
\det\, E_{-+}(z) = \det\, E_{-+,1}(z) +o(\sqrt{z}) = \sqrt{z}\, a_{m1} + o(|z|^{1/2})
\end{equation}
where $a_{m1}$ is the non-zero constant given in (\ref{am1}).

\begin{proof}[Proof of Theorem \ref{thm:cas2}]
It follows from the previous paragraph that $E_{-+,1}(z)$ is invertible for $z\in \Omega_{\delta}, \delta>0$ small, and we can easily check that
\begin{align}
E_{-+,1}(z)^{-1} &= \frac{^{t}Com E_{-+,1}(z)}{ \det E_{-+,1}(z) }\nonumber \\
&= \frac{1}{\sqrt{z}}\widetilde{A} + O(1)\label{e.3.20},
\end{align}
where $$\widetilde{A}=\frac{1}{a_{m1}} \begin{pmatrix}
0&\cdots &0& 1 \\
0&\cdots &0 & 0\\
\vdots & & \vdots & \vdots\\
0 & \cdots &0 & 0
\end{pmatrix}_{m\times m}.$$ 
Then, if $\rho>\ell+1$ and $z\in \Omega_{\delta},$ $\delta>0$ small, $E_{-+}(z)^{-1}$ exists, with  
\begin{align*}
E_{-+}(z)^{-1} &= ( I+ E_{-+,1}(z)^{-1}\sum\limits_{j=2}^{\ell} z^{j/2}E_{-+,j})^{-1}E_{-+,1}(z)^{-1}\\
&= \frac{1}{\sqrt{z}}\widetilde{A} + \sum_{j=0}^{\ell-2} z^{j/2} F_{j}^{(2)} + E_{-+,\ell-2}^{(-1)} (z),
\end{align*}
where $\widetilde{A}$ is the above matrix, $F_j^{(2)}$, $j=1,\cdots,\ell-2,$ can be computed explicitly.  Moreover, for 
$0<\lambda<\delta$ the limits of the remainder term  $\ds E_{-+,\ell-2}^{(-1)}(\lambda\pm i0)$ exist and satisfy
\begin{equation}\label{Rest:E-+inverscas2}
\Vert \frac{d^r}{d\lambda^r} E_{-+,\ell-2}^{(-1)}(\lambda\pm i0)\Vert = o(|\lambda|^{\frac{l}{2}-1-r}), \ \lambda \in ]0,\delta[, \quad r=0,1,\cdots,\ell-2.   
\end{equation}

Next, using the formula (\ref{1+R}) we can verify that if $\ell-1/2<s<\rho-\ell+1/2$, $\ell\in \bN$  with $\ell\geq 2$
\begin{equation} \label{(I+R)inverse:cas1}
 (Id+ R_{0}(z)V)^{-1}=  \sum_{j=-1}^{\ell-2} z^{j/2} B_j^{(2)}+ \widetilde{B}_{\ell-2}^{(2)}(z), 
\end{equation}
in $\b(1,-s,1,-s)$ where 
$$B_{-1}^{(2)}= -S\widetilde{A}T = -\frac{1}{ca_{m_1}}\inp{\cdot, JVu_1}u_1$$
and the remainder $\widetilde{B}_{\ell-2}^{(2)}(z)$ is a $\vC^{\ell}$-function from $\Omega_{\delta}$ to $\B(1,-s,1,-s)$ such that for $0<\la<\delta$ the limits 
$\widetilde{B}_{\ell-2}^{(2)}(\lambda\pm i0)$ exist and satisfy
\begin{equation}\label{estim:(I+R)cas1}
\Vert \frac{d^r}{d\lambda^r} \widetilde{B}_{\ell-2}^{(2)}(\lambda\pm i0)\Vert_{\B(1,-s,1,-s)} = o(|\lambda|^{\frac{\ell}{2}-1-r}), \ \lambda \in ]0,\delta[, \quad r=0,1,\cdots,\ell-2.   
\end{equation}
To make this computation we introduce (\ref{e.3.20}) in (\ref{1+R})
 together with the expansion (\ref{E(z):dev}) of $E(z)$ up to order $\ell-1$. We obtain the expansion of $E_+(z)E_{-+}(z)^{-1}E_{-}(z)$ up to order $\ell-2$ with remainder  $o(|z|^{\frac{\ell-2}{2}})$ which requires the assumption $\ell-1/2<s<\rho-\ell +1/2$. In addition, estimates (\ref{estim:(I+R)cas1}) can be checked from (\ref{Rest:E(z)}) and (\ref{Rest:E-+inverscas2}).\\    

 
Consequently, for $\rho>2\ell -1$ and $s>\ell-1/2$, $\ell\in \bN$ with $\ell\geq 2$, using the equation (\ref{ResEq}), (\ref{(I+R)inverse:cas1}) and (\ref{ResLibre:dev}) the expansion of $R(z)$ in $\b(-1,s,1,-s)$,  can be written as follows: 
\begin{align*}
R(z)  &=   \sum_{j=-1}^{\ell-2} z^{j/2} R_j^{(2)}+
\widetilde{R}_{l-2}^{(2)}(z), 
\end{align*}
where 
\begin{align*}
R_{-1}^{(2)}&= -\frac{1}{ca_{m_1}}\inp{\cdot, JG_0Vu_1}u_1\\
&= \frac{1}{c\,a_{m1}}\langle \cdot,J u_1\rangle  u_1\\
&= i\frac{4\pi}{\langle u_1,JV1\rangle^2}\langle \cdot,J u_1\rangle  u_1.    
\end{align*}
Finally, let $$\phi = \frac{2\sqrt{\pi}}{\langle u_1,JV1\rangle }u_1. $$
Then $\phi$ is a resonance state of $H$ satisfying (\ref{cond:phi}) and $R_{-1}^{(2)}= i\langle \cdot, \phi\rangle \phi$.
Moreover, the estimate (\ref{Rest:R(z)2}) follows from
(\ref{estim:(I+R)cas1}).
\end{proof}

\subsection{Zero singularity of the second kind}\label{sub:cas1}  In this case we assume that  $-1$  is an eigenvalue of the operator $K_0$ on $L^{2,-s}, 1/2<s<\rho-1/2,$ with geometric multiplicity $k\geq 1$. Indeed the case $k=1$ could be treated using the similar method used in \cite{wang2017gevrey} to study the situation of geometrically simple zero eigenvalue for a compactly supported perturbation
of the non-selfadjoint Schr\"odinger operator $H_0=-\Delta +V_0(x)$ with a slowly decreasing potential $V_0$. In the latter case, the matrix of $\Pi_1 (Id+K_0)\Pi_1$ on $\text{Ran}\,\Pi_1$ is consisting of one Jordan block. Therefore, the usual tools can be used to compute the singularity of the 
resolvent at threshold zero. This work is concerned with the more interesting case  $k\geq2$. Assume from now that $k\geq2$. Let $\mathcal{U}=\cup_{j=1}^k \{ u_1^{(j)},\cdots,u_{m_j}^{(j)}\}$ and 
$\mathcal{W}=\cup_{j=1}^k \{ w_1^{(j)},\cdots,w_{m_j}^{(j)}\}$ be the basis constructed in Lemma \ref{Lem:decompE}. In particular 
$\{u_1^{(1)},\cdots,u_1^{(k)}\}$ is a basis of Ker$_{L^2}(Id+K_0)$.\\

In order to prove Theorem \ref{thm:cas1}, we use the Grushin's  method introduced in Section \ref{sec:Gr} for arbitrary Jordan structure. We begin by computing the asymptotic expansion of the matrix $E_{-+}(z)$. To study this matrix we  need to decompose it conformally with the block decomposition of the matrix $J_m$ in (\ref{J}). More precisely, 
using formula (\ref{E-+:def}) and basis $\mathcal{U}_j$ and  $\mathcal{W}_i$, we have
\begin{align*}
E_{-+}(z) &= \big(E_{-+}^{(ij)}(z)\big)_{1\leq i,j\leq k} \ \text{with} \\ 
 E_{-+}^{(ij)}(z)&= (\langle E_{-+}(z)u_{r}^{(j)},J V w_{l}^{(i)}\rangle)_{1\leq l\leq m_{i},1\leq r\leq m_{j}},
\end{align*} 
where $E_{-+}^{(ij)}(z)$ denotes the $m_{i}\times m_{j}$ block entry of $E_{-+}(z)$ located in the same row as $J_ {m_ {i}}$ and in the same column as $J_ {m_ {j}}$.\\

Thus, if $\rho>\ell+1$, $\ell\in \bN$ with $\ell \geq 3$, 
using (\ref{ResLibre:dev}), (\ref{E(z):dev}) and  (\ref{E-+:def}) we obtain the following expansion of $E_{-+}(z)$ for $z \in \Omega_{\delta}$, $\delta >0$ small:
\begin{equation}\label{E-+:devcas1}
E_{-+}(z) = E_{-+,2}(z) + \sum_{j=3}^{\ell} z^{j/2} E_{-+,j} + \widetilde{E}_{-+,\ell}(z),
\end{equation}
where $E_{-+,2}^{(ij)}(z)= N^{(ij)} + z^{1/2}\  A^{(ij)} + z \ B^{(ij)}$, 
such that for all $ 1\leq \ell\leq m_i$ and $ 1\leq r\leq m_j$ we have
\begin{align*}
N^{(ij)}_{\ell r} &= -\inp{ (Id+G_0V)u_r^{(j)}, JV w_{\ell}^{(i)}},\nonumber\\
A^{(ij)}_{\ell r} &= -i\inp{ G_1Vu_r^{(j)}, JV w_{\ell}^{(i)}} + i \inp{ E_0 G_1Vu_r^{(j)}, JV (Id+G_0V) w_{\ell}^{(i)}}, \nonumber\\
B^{(ij)}_{\ell r}&= \inp{(G_2V- G_1V E_0 G_1V) u_r^{(j)}, JV w_{\ell}^{(i)}}\\
& +\inp{(iE_1 G_1V-E_0G_2V) u_r^{(j)}, JV (Id+G_0V) w_{\ell}^{(i)}} \nonumber
\end{align*}
also $E_{-+,n}^{(ij)}, n=3,\cdots,\ell,$ can be computed explicitly. Moreover,  the remainder term $\widetilde{E}_{-+,\ell}(z)$ is a $\vC^{\ell}$ matrix-valued function of $z$ in $\Omega_{\delta}$ and for $0<\lambda<\delta$ the limits
\begin{equation}\label{Rest:E-+cas1+}
\lim_{\epsilon\to 0^+} \widetilde{E}_{-+,\ell}(\lambda\pm i\epsilon)= 
\widetilde{E}_{-+,\ell}(\lambda\pm i0)
\end{equation}
exist and  satisfy
\begin{equation}\label{Rest:E-+cas1}
\Vert \frac{d^r}{d\lambda^r}\widetilde{E}_{-+,\ell}(\lambda\pm i0) \Vert = o(|\lambda|^{\frac{\ell}{2}-r}), \ r=0,1,\cdots,\ell.   
\end{equation}
We can further simplify the previous expression of the matrix $E_{-+,2}(z)$ as follows: $\forall 1\leq j\leq k$, $u_{1}^{(j)}\in$ ker$(Id+ K_0)$ implies that $N^{(ij)}_{\ell 1}=0, \forall 1\leq \ell \leq m_i$, while for all $2\leq r \leq m_j$, $N^{(ij)}_{\ell r}=-\inp{ u_{r-1}^{(j)}, JV w_{\ell}^{(i)}}= -\delta_{\ell r-1}\delta_{ij}$ by (\ref{vect:w}). Moreover, since $w_{m_i}^{(i)}= c_i u_1^{(i)}\in L^2$ for some $c_i\neq 0$ (see Remark \ref{Rem:c_j}) then  $G_1Vu_1^{(j)}=0= G_1V w_{m_i}^{(i)}$,  $\forall 1\leq i,j\leq k$, by (\ref{Char:resonState}). This implies that $A^{(ij)}_{\ell 1}=0=A^{(ij)}_{m_i r}, \forall 1\leq \ell\leq m_i, 1\leq r\leq m_j$.\\

Summing up, we obtain 
 $E_{-+,2}(z)= N + z^{1/2}\  A + z \ B$, with
\begin{equation}\label{N}
N^{(ij)}=\begin{pmatrix}
0&-\delta_{ij}& \cdots &0 \\
 0 & 0&\ddots& \vdots\\
\vdots &  \vdots &\ddots&-\delta_{ij}\\
0&0& \cdots&0
\end{pmatrix}_{m_{j}\times m_{j}},
\end{equation}
\begin{equation}\label{matrices:A,B}
A^{(ij)}= \begin{pmatrix}
0& & &   \\
0&& \widetilde{A}^{(ij)}&\\
\vdots& & &\\
0 &0 &\cdots & 0\\
\end{pmatrix}_{m_{i}\times m_{j}}, \,
 B^{(ij)}= \begin{pmatrix}
\ast &\ast &\cdots& \ast  \\
\vdots & \vdots  &   &\vdots\\
 \ast & \ast&\cdots & \ast\\
 \beta_{ij}& \ast& \cdots& \ast\\
\end{pmatrix}_{m_{i}\times m_{j}}, \ \forall 1\leq i,j\leq k,
\end{equation}
where
\begin{align}
\beta_{ij}&= -\lim_{ z\in \Omega_{\delta},z \to 0}\frac{1}{z}\langle  (Id+R_{0}(z)V)
u_{1}^{(j)}, JVw_{m_{i}}^{(i)}\rangle \nonumber\\
& =- \lim_{ z\in \Omega_{\delta},z \to 0}\frac{1}{z}\lbrace \langle (Id+ G_{0}V)u_{1}^{(j)}, J V w_{m_{i}}^{(i)}\rangle + z \langle G_{0}V u_{1}^{(j)}, J R_{0}(z)V w_{m_{i}}^{(i)}\rangle \rbrace \nonumber\\
&= - c_{i}^{-1} \langle  u_{1}^{(j)},J u_{1}^{(i)}\rangle , \quad \forall 1\leq i,j\leq k \label{bij}.
\end{align}
See the proof of (\ref{E-+:cas2}) for more details.\\
%


Unfortunately, we have  found in (\ref{E-+:devcas1}) a block matrix  $E_{-+}(z)$ of arbitrary block structure, where the usual methods of algebra are no longer practical to calculate its determinant and to explicitly develop its inverse matrix.  
To treat this matrix, we propose a method based on that of Lidskii developed in his original paper \cite{lidskii1966perturbation}, and used latter in \cite{moro1997lidskii} for the problem of eigenvalues of matrices with arbitrary Jordan structure. 
\begin{rem}\label{remq1}
Set 
\begin{equation}\label{matrix:Phi}
\phi_{k} = \big(\beta_{ij}\big)_{1\leq i,j\leq k}, \ L_k= \Big(\langle u_1^{(j)},Ju_1^{(i)}\rangle \Big)_{1\leq i,j\leq k},
\end{equation} 
where  $\beta_{ij}$ are the coefficients in (\ref{matrices:A,B}).
Then, it is seen from (\ref{bij}) and  Remark \ref{Rem:c_j} that 
\begin{equation}\label{matrices:Phi,L}
\phi_{k}= -C_k  L_k\ \mbox{with} \ \displaystyle C_k=\text{diag} (c_1^{-1},\cdots,c_k^{-1}).  
\end{equation}
\end{rem}
\begin{lem}\label{lem:detE-+cas1}
Assume that $zero$ is a singularity of the second kind of $H$ and that $(H1)$ holds. If $\rho>3$,  we have   
\begin{equation*}
\text{det}\, E_{-+}(z) = \sigma z^{k} + \mathcal{O}(|z|^{k+\epsilon}),\  \forall z \in \Omega_{\delta},
\end{equation*}
for some $0<\epsilon<1/2$, where $\sigma = \sigma' \times \text{det} (\langle u_{1}^{(j)}, J u_{1}^{(i)}\rangle )_{1\leq i,j\leq k}$ with $\sigma' \in \C^*$ and $k= \text{dim\, Ker}(Id+K_0)$.
\end{lem}

\begin{proof}
We begin by writing 
 the expansion of the matrix $E_{-+}(z)$ in (\ref{E-+:devcas1}) as follows:
 \begin{equation}\label{E(eta)}
 E_{-+}(\eta) = E_{-+,2}(\eta)+\mathcal{O}(|\eta|^{2(1+\epsilon)}),
\end{equation}
for some $0<\epsilon<1/2$.
Let $Z(\eta)= \text{det} E_{-+,2}(\eta)$.  
Then,  we reduce the computation to that of $Z(\eta)$ close to $\eta=0$. To do it we  introduce the following diagonal matrix $L(\eta)$   partitioned conformally with the block structure of the matrix $E_{-+}(\eta)$:
\begin{equation}\label{matrix:L}
L(\eta)= \text{diag} \big(L_1(\eta),\cdots, L_k(\eta)\big), \ L_i(\eta)= \text{diag} (1,\cdots,1,\eta^{-2}), \ 1\leq i\leq k,
\end{equation}
where $\eta \in \{ z\in \C_+: |z|<\delta\}$.
We now define  
\begin{equation}\label{Ftilde}
\widetilde{E}_{-+,2}(\eta) = L(\eta)E_{-+,2}(\eta), \ \widetilde{Z}(\eta) =
\text{det}\, \widetilde{E}_{-+,2}(\eta).
\end{equation}
Then, by regularity of the matrix $L(\eta)$ for $\eta\neq 0$, we see that $ \widetilde{Z}(\eta) =0$ if and only if $Z(\eta)=0.$
Also, we can show that $\widetilde{Z}(\eta)$ is polynomial in $\eta$. Indeed, we write
\begin{equation*}
\widetilde{E}_{-+,2}(\eta)= L(\eta) (N+ \eta \, A + \eta^2\, B):= \widetilde{N}(\eta) + \widetilde{A}(\eta) +\widetilde{B}(\eta),
\end{equation*}
where, by (\ref{N}), (\ref{matrices:A,B}) and (\ref{matrix:L}) we see that
\begin{enumerate}
\item[(i)] $\widetilde{N}(\eta) =  L(\eta) N = N.$
\item[(ii)]$\widetilde{A}(\eta)= \eta L(\eta) A $  \ \text{with} \ 
$\widetilde{A}^{(ij)}(\eta)= \begin{pmatrix}
0& &   \\
\vdots& \eta \ \widetilde{A}^{(ij)}&\\
& &\\
0 & \cdots&0 \\
\end{pmatrix}_{m_i\times m_j}.$

\item[(iii)]$\widetilde{B}(\eta)=\eta^{2} L(\eta)B$ \ \text{with} \ 
$$\widetilde{B}^{(ij)}(\eta)= 
 \begin{pmatrix}
\eta^2& 0& \cdots&0 \\
0&\ddots&\ddots& \vdots\\
\vdots&\ddots & \eta^2& 0\\
0&\cdots & 0& 1\\ 
\end{pmatrix}_{m_i\times m_i}  
\begin{pmatrix}
\ast &\ast &\cdots& \ast  \\
\vdots & \vdots  &   &\vdots\\
 \ast & \ast&\cdots & \ast\\
 \beta_{ij}& \ast& \cdots& \ast\\
\end{pmatrix}_{m_i\times m_j}.$$
\end{enumerate}
This shows that there is no negative powers of $\eta$ in $\widetilde{E}_{-+,2}(\eta)$. We will then examine $\widetilde{Z}(0)$.
It follows from  (i) (resp. (ii)) that $\widetilde{N}(0)= N$ 
(resp. $\widetilde{A}(0)=0$). Moreover,
\begin{equation}\label{matrix:B(0)}
\widetilde{B}^{(ij)}(0)= 
 \begin{pmatrix}
0&0 &\cdots &0 \\
\vdots&\vdots&&\vdots \\
0&0 & \cdots&0 \\
\beta_{ij}& \ast&\cdots& \ast  
\end{pmatrix}_{m_i\times m_j}, \ \forall 1\leq  i,j\leq k.
\end{equation}
Thus, 
\begin{equation}\label{F(0)}
\widetilde{E}_{-+,2}^{(ij)}(0) = \begin{pmatrix}
0&-\delta_{ij}&0&\cdots &0\\
\vdots&\ddots&\ddots&\ddots&\vdots\\
0& &\ddots & \ddots&0\\
0 & &&0 & -\delta_{ij}\\
\beta_{ij}& \ast&\cdots &\ast&0
\end{pmatrix}_{m_i\times m_j},\ \forall 1\leq i,j\leq k.
\end{equation}
We can now calculate $\widetilde{Z}(0)$ which is the determinant of the above matrix $\widetilde{E}_{-+,2}(0)$. By expanding the determinant along the rows of $\widetilde{E}_{-+,2}(0)$ that are containing only $-1$, we obtain the following: 
\begin{equation}\label{detZ(0):cas1}
\widetilde{Z}(0)=\text{det}\ \left(\beta_{ij}\right)_{1\leq i,j\leq k},
\end{equation}
(see proof of Theorem 2.1 in \cite{moro1997lidskii} for a specific example with $12\times 12$ matrix that illustrates the strategy).
Hence, there exist $\epsilon, \eta >0$ small such that for $\eta\in \{ z\in \C_+: |z|<\delta\}$ 
\begin{equation}\label{Z(eta)=Z(0)}
\widetilde{Z}(\eta) = \widetilde{Z}(0)+ \mathcal{O}(|\eta|^{2\epsilon})=  \text{det} \Phi_{k} + \mathcal{O}(|\eta|^{2\epsilon})
\end{equation}
where $ \Phi_{k}=(\beta_{ij})_{1\leq i,j\leq k}$ is defined in (\ref{matrix:Phi}). Then, it follows from (\ref{Ftilde}) and (\ref{Z(eta)=Z(0)}) with $\det L(\eta)= \eta^{-2k}$  
\begin{equation}\label{Z}
Z(\eta)=  (\text{det} L(\eta))^{-1} \widetilde{Z}(\eta)=\eta^{2k} \text{det} \Phi_{k}+ \mathcal{O}(|\eta|^{2(k+\epsilon)}).
\end{equation}
Finally, (\ref{E(eta)}) with the previous equation implies that for $\eta\in \{ \eta \in \C_+,\, |\eta|<\delta \}$ 
$$\text{det} E_{-+}(\eta) =\eta^{2k}\text{det}\Phi_{k}+  \mathcal{O}(|\eta|^{2(k+\epsilon)}),$$  
where $\text{det}(\Phi_k) =\sigma' \times \text{det} (
\langle u_{1}^{j}, J u_{1}^{i}\rangle)_{1\leq i,j\leq k}$  with some $\sigma' \neq 0$  by Remark \ref{remq1}.
\end{proof}
Now, we are able to prove Theorem \ref{thm:cas1}.
\\
\begin{proof}[Proof of Theorem \ref{thm:cas1}]

Firstly, it follows from  Lemma \ref{lem:detE-+cas1} that  $E_{-+}(z)^{-1}$ exists under the hypothesis $(H1)$. Then, we will show that for $\rho>\ell+1,$  $\ell=4,5,\dots$ and $z\in \Omega_{\delta}$, $\delta>0$ small, the expansion of $E_{-+}(z)^{-1}$ has the following form :
\begin{equation}\label{E-+inversecas3}
E_{-+}(z)^{-1}= \sum_{j=-2}^{\ell-4} z^{j/2} F_j^{(1)} + E_{-+,\ell-4}^{(-1)}(z),
\end{equation}
where  $F_{-2}^{(1)}$ is a matrix of rank $k$,  whose blocks are of the form 
\begin{equation}\label{D}
(F_{-2}^{(1)})^{(ij)}= \frac{1}{\det \Phi_k}
\begin{pmatrix}
0&\cdots&0&\gamma_{ij}\\
0&\cdots&0&0\\
\vdots& &\vdots&\vdots\\
0&\cdots&0&0\\
\end{pmatrix}_{m_i\times m_j}, \ \forall 1\leq i,j\leq k, 
\end{equation}
for some $\gamma_{ij}$ that will be determined during this proof  and the matrix $\displaystyle  F_{-1}^{(\ell)}= - F_{-2}^{(1)}E_{-+,3}F_{-2}^{(1)}$ has rank at most $k$. Moreover,  for $0<\lambda<\delta$, the limits 
\begin{equation}\label{Rest:E-+invcas1+}
\lim\limits_{\epsilon\to 0^+}  E_{-+,\ell-4}^{(-1)}(\lambda\pm i\epsilon) =  E_{-+,\ell-4}^{(-1)}(\lambda\pm i0)    
\end{equation}
exist  and satisfy
\begin{equation}\label{Rest:E-+}
\Vert \frac{d^r}{d\lambda^r} E_{-+,\ell-4}^{(-1)}(\lambda\pm i0) \Vert = o\left(|\lambda|^{\frac{\ell}{2}-2-r}\right), \  r=0,1,\cdots,\ell-4.  
\end{equation}
In order to prove (\ref{E-+inversecas3}) we consider the same notations that we have just used in the previous proof. We see that for $\Im\eta>0$ and $|\eta|<\delta$ with $\delta>0$ small, $\widetilde{E}_{-+,2}(\eta)$  can be developed in powers of $\eta$ as follows: 
\begin{equation}\label{F1}
 \widetilde{E}_{-+,2}(\eta) = \widetilde{E}_{-+,2}(0) + \eta A  + \eta^2 B_1 + o(|\eta|^{2}),
\end{equation}  
where $B_1= B- \widetilde{B}(0)$ and the matrices $A$, $B$ and $\widetilde{B}(0)$ are given in  (\ref{matrices:A,B}) and (\ref{matrix:B(0)}). \\
In addition, $\widetilde{E}_{-+,2}(0)^{-1}$ exists by (\ref{detZ(0):cas1}) under the condition (\ref{Hyp:H1}) with
\begin{equation}\label{F(0):inverse}
\widetilde{E}_{-+,2}(0)^{-1}=\frac{^{t}Com 
\widetilde{E}_{-+,2}(0)}{\det \Phi_{k}}=
 F_{-2}^{(1)}+\mathcal{E}^{(1)},   
\end{equation}
where $F_{-2}^{(1)}$ is the above matrix and  $\mathcal{E}$ is a block matrix, with  block entries of the  form
\begin{equation}\label{matrix:E}
(\mathcal{E}^{(1)})^{(ij)}=
\frac{1}{\det \Phi_{k}} \begin{pmatrix}
\ast&\ast&\cdots&\ast &0\\
\widetilde{\alpha}_{2}^{(ij)}&0&\cdots&0&0\\
 0&\ddots&\ddots&&\\
 \vdots& &\ddots &\ddots&\vdots\\
0&\cdots &0& \widetilde{\alpha}_{n_{i}}^{(ij)}&0
\end{pmatrix}_{m_i\times m_j}   \quad  \forall 1\leq i,j\leq k,\\
\end{equation}
where $\widetilde{\alpha}_{r}^{(ij)}\in \C$ are  minors of order $m-1$ of the matrix $\widetilde{E}_{-+,2}(0)$. In particular $\widetilde{\alpha}_{r}^{(ij)}= 0$ if $i\neq j$. Here, we have applied the same process  used in the previous proof to calculate the minors of order $m-1$ of the matrix $\widetilde{E}_{-+,2}(0)$. For the rest of this proof we need to give more precision on the coefficients $\gamma_{ij}$ of the above matrix $F_{-2}^{(1)}$. Let $|[M]_i^j|:=$det $[M]_i^j$ denote the minors of a matrix $M$, where $[M]_i^j$ are  the resulting matrices when the $i$-th row and the $j$-th column  of $M$ are deleted. Then
\begin{align}
\gamma_{1j} =&
 (-1)^{1+m_1+\cdots+m_j} | [\widetilde{E}_{-+,2}(0)]_{m_{1}+\cdots+m_{j}}^{1}| \ \text{and}\label{gam2} \\
\gamma_{ij} =&
 (-1)^{\mu_{ij}} | [\widetilde{E}_{-+,2}(0)]_{m_{1}+\cdots+m_{j}}^{1+m_{1}+\cdots +m_{i-1}}|, \ \mu_{ij} = 1+\sum_{l=1}^{i-1} m_{l}+\sum_{r=1}^{j} m_{r} \label{gam1}
\end{align}
for $ 2\leq i\leq k$ and  $ 1\leq j\leq k$.
Then, for $\eta\in \{ \eta\in \C_+; |\eta|<\delta\}$ with $\delta$ small enough $\widetilde{E}_{-+,2}(\eta)$ in (\ref{F1}) is invertible, as well as $E_{-+,2}(\eta)$ with
$E_{-+,2}(\eta)^{-1}= \widetilde{E}_{-+,2}(\eta)^{-1} L(\eta)$.   
More precisely,  using (\ref{Ftilde}) we obtain
\begin{align}
E_{-+,2}(\eta)^{-1} &=\Big(\widetilde{E}_{-+,2}(0)^{-1}+o(|\eta|^2)\Big) L(\eta)\nonumber\\
&=\frac{F_{-2}^{(2)}}{\eta^{2}} 
+ F_0^{(2)} + \mathcal{O}(|\eta|), \quad \forall \eta \in \Omega_{\delta}\cap \C_+,
 \label{F:inverse}
 \end{align} 
where $\displaystyle F_0^{(2)}= \mathcal{E}^{(2)} - \mathcal{E}^{(2)} B_1 F_{-2}^{(2)}$ and we can check that  $F_0^{(2)}L(\eta)= F_0^{(2)}$.\\

Let $\ds E_{-+}^I(\eta)= E_{-+}(\eta)-E_{-+,2}(\eta)$. For $\eta\in \{ \eta \in \C_+,\, |\eta|<\delta \}$, we see that
$\ds \Vert E_{-+,2}(\eta)^{-1}  E_{-+}^I(\eta)\Vert = \mathcal{O}(|\eta|)$. 
Consequently,  we deduce from (\ref{E-+:devcas1}) and (\ref{F:inverse}) that $E_{-+}(\eta)^{-1}$ exists for $\Im\eta>0$ and $|\eta|<\delta$, $\delta>0$ small enough with
\begin{align}
E_{-+}(\eta)^{-1}
&=\frac{F_{-2}^{(2)}}{\eta^{2}} - \frac{F_{-2}^{(2)}E_{-+,3}F_{-2}^{(2)}}{\eta} 
- F_0^{(2)} - F_{-2}^{(2)}E_{-+,4}F_{-2}^{(2)} \label{E-+2}\\
&+\sum\limits_{j=1}^{\ell-4} \eta^j F_{j}^{(2)} + E_{-+,\ell-4}^{(-1)}(\eta),\nonumber
\end{align}
where the estimates (\ref{Rest:E-+}) follow from 
(\ref{Rest:E-+cas1}). See the proof of (\ref{Rest:E-+cas1+}) for (\ref{Rest:E-+invcas1+}).

It is remaining to prove that $F_{-2}^{(2)}$ has rank $k$.  We shall show that    
\begin{equation}\label{rem:rank}
\Gamma_k:=(\gamma_{ij})_{1\leq i,j\leq k} = (\det\, \Phi_k ) \Phi_k^{-1}.
\end{equation}
Indeed, by (\ref{gam2}) and (\ref{gam1}) we can  check that
\begin{equation}\label{prove of Rank}
 \gamma_{ij} = \left \{\begin{array}{ccc}
(-1)^{\mu_{ij}} (-1)^{m_i-1+\cdots + m_j -1} | [\Phi_k]^i_j| & \text{if}  & i<j,\\
\\
(-1)^{\mu_{ij}} (-1)^{m_i-1} | [\Phi_k]^i_i| & \text{if} & i=j,\\
\\
(-1)^{\mu_{ij}} | [\Phi_k]^i_j| & \text{if} & i=j+1,\\
\\
(-1)^{\mu_{ij}} (-1)^{m_{j+1}-1+\cdots + m_{i-1} -1} 
| [\Phi_k]^i_j| & \text{if}   & i>j+1,
\end{array}\right.
\end{equation}
where $|[\Phi_k]^i_j|$ is the $(j,i)$-th minor of the invertible matrix 
$\Phi_k$ defined in (\ref{matrix:Phi}). Substituting 
the values of $\mu_{ij}$ given in (\ref{gam1}) yields 
$$ \gamma_{ij} = (-1)^{i+j} | [\Phi_k]^i_j|,$$
which is the $(j,i)$-th cofactor of the matrix $\Phi_k$ for all 
$1\leq i,j\leq k.$\\

Secondly, if $\ell-3/2<s<\rho-\ell-3/2$, we obtain the expansion of $(I+R_0(\eta^2)V)^{-1}$ up to order $\ell-4$ with  remainder 
$o(|\eta|^{\ell-4})$ by using the formula (\ref{1+R}), the expansion (\ref{E(z):dev}) up to order $\ell-2$ and the expansion of $E_{-+}^{-1}(\eta)$ in (\ref{E-+2}).\\

Finally, if $\rho>2\ell-3$ and $s>\ell-3/2$, using the identity $R(z)= (Id+R_0(z)V)^{-1}R_0(z)$, the expansion of the resolvent in $\b(-1,s,1,-s)$ has the following form:
\begin{equation}\label{R(z):cas3proof}
R(z)=  \sum\limits_{j=-2}^{\ell-4}
z^{j/2} R_j^{(2)} + R_{\ell-4}^{(2)}(z), 
 \end{equation}
where $R_{-1}^{(2)}= -SF_{-1}^{(2)}TG_0$, $R_0^{(2)}= E_0G_0 + SF_{-2}^{(2)}TG_2 - SF_0TG_0$, $S$ and $T$ are defined in (\ref{matrix:P(z)}) and for $f\in \h^{-1,s}$ 
\begin{align}
R_{-2}^{(2)}f &= -SF_{-2}^{(2)}TG_0f\nonumber\\
&=\frac{-1}{\det \phi_k}
\overset{k}{\underset{i=1}{\sum}}\overset{k}{\underset{j=1}{\sum}} \gamma_{ij} \langle f , JG_0V w_{m_j}^{(j)}\rangle u_{1}^{(i)}\nonumber\\
&=\frac{1}{\det \phi_k}
\overset{k}{\underset{i=1}{\sum}}\overset{k}{\underset{j=1}{\sum}}
c_j^{-1} \gamma_{ij} \langle f , J u_1^{(j)}\rangle u_{1}^{(i)}\nonumber\\
&= \sum\limits_{i=1}^k \langle f, Jv_i\rangle  u_{1}^{(i)}\label{v_j}.
\end{align}
 Let
 $$V= \begin{pmatrix}
v_1\\
\vdots\\
v_k
\end{pmatrix}, \quad U= \begin{pmatrix}
u_1^{(1)}\\
\vdots\\
u_1^{(k)}
\end{pmatrix}, \quad \Gamma_k=(\gamma_{ij})_{1\leq i,j\leq k}.$$
Then, we have 
\begin{align}\label{vectorV}
V &= \frac{1}{\det \Phi_k}\Gamma_k C_k U.
\end{align}
Thus, using (\ref{rem:rank}) we obtain
\begin{equation}\label{matrix:Q}
V=\Phi_k^{-1} C_k U= - (C_k L_k)^{-1} C_k U= -L_k^{-1} U= -\, ^tQ Q U,    \end{equation}
where $C_k=\text{diag}(c_1^{-1},\cdots,c_k^{-1})$, $L_k=(\langle u_1^{(j)}, Ju_1^{(i)}\rangle )_{1\leq i,j\leq k}$ is an invertible complex symmetric matrix with (\ref{matrices:Phi,L}) and $Q=(q_{ij})_{1\leq i,j\leq k}$ is the upper triangular matrix obtained by the Cholesky decomposition of the matrix $L_k^{-1}$ (cf. \cite[Proposition 25]{Christian2017Matrix}). Thus by returning to (\ref{v_j}), we get 
 \begin{equation*}
  R_{-2}^{(2)}f = -\sum\limits_{i,j=1}^k \sum_{\ell=1}^k q_{\ell i}q_{\ell j} \langle f,Ju_1^{(j)}\rangle u_1^{(i)} =- \mathcal{P}_0^{(2)}f
  \end{equation*}
 where 
 \begin{equation*}
   \mathcal{P}_0^{(2)}= \sum_{\ell=1}^k\langle \cdot,J \mathcal{Z}_{\ell}^{(2)}\rangle \mathcal{Z}_{\ell}^{(2)} \ \text{with}\   \mathcal{Z}_{\ell}^{(2)}= \sum_{i=1}^k q_{\ell i}
   u_1^{(i)},
 \end{equation*}
and we see that  $\mathcal{P}_0^{(2)}$ is a projection of rank $k$ since for all $1\leq i,j\leq k$ we  have 
$$ \langle \mathcal{Z}_{i}^{(2)}, J\mathcal{Z}_{j}^{(2)}\rangle =\sum_{\ell,m}^k q_{i\ell}q_{jm} \langle u_{1}^{(\ell)}, Ju_1^{(m)}\rangle  = \sum_{l=1}^k q_{i\ell} (QL_k)_{j\ell}= (QL_k\, ^tQ)_{ji}= \delta_{ij}. $$
Moreover, other terms $R_j^{(2)}, \ j= 1,\cdots,\ell-4,$ can be obtained explicitly. Finally, the  estimate  (\ref{Rest:R(z)1})  can be seen from (\ref{Rest:E(z)}) and (\ref{Rest:E-+}), also (\ref{Rest:R(z)1+}) follows from (\ref{Lim ResLibre}), (\ref{Rest:E(z)+}) and (\ref{Rest:E-+invcas1+}).
\end{proof}

\subsection{Zero singularity of the third kind}\label{sub:cas3}
In this section we discuss the case when zero is both an embedded eigenvalue and a resonance of $H$. We assume that the eigenvalue zero has geometric multiplicity $k-1$, $k\in \bN$ with $k\geq 2$.
 In this case, we have $\text{dim ker}_{L^{2,-s}} (Id+K_0) =k$ and  dim $\text{ker}_{L^{2,-s}} (Id+K_0) /\text{ker}_{L^{2}} (Id+K_0)=1$.
Set $m=\text{rank}\, \Pi_1.$
 Let $ \mathcal{V}_i = \lbrace v_1^{(i)}, \cdots, v_{m_i}^{(i)} \rbrace$ be a basis of $E_i$ and  $\mathcal{X}_i = \lbrace \chi_1^{(i)}, \cdots,\chi_{m_i}^{(i)} \rbrace$ be its $\Theta$-dual basis (see Lemma \ref{Lem:decompE}). Under the hypothesis $(H2)$ one can consider that  
$\text{Ker}_{L^{2,-s}}(Id+K_0) / \text{Ker}_{L^{2}}(Id+K_0)= \text{Span}\{v_1^{(1)}\} $ and $\text{Ker}_{L^{2}}(Id+K_0)=\text{Span}\{ v_1^{(2)}, \cdots, v_1^{(k)}\}$.
This means that the matrix of $ \Pi_1(Id+K_0)\Pi_1$ in the basis $\vX$ is a $k\times k$ block diagonal matrix with Jordan block on the diagonal such that only the first Jordan block corresponds to the resonant state  and the other blocks  correspond to the eigenvectors that are the solutions in $ L^2$ of $ (Id+ K_0) g = 0.$\\

 
Since $-1$ is an eigenvalue of the operator $K_0$ of geometric multiplicity $k\geq 1$, the same computations made to develop $E_{-+}(z)^{-1}$ in Section \ref{sub:cas1} can be done here with a slight difference. 
Indeed, (\ref{E-+:devcas1}) holds with some difference in the block entries of the matrix $A$ as follows:
$$A^{(ij)}_{m_i 1} = -i\langle G_1Vv_{1}^{(j)}, JV\chi_{m_i}^{(i)}\rangle = (i4\pi c_i)^{-1} \langle v_1^{(i)},JV1\rangle \langle v_1^{(j)},JV1\rangle$$
do vanish for all $i,j$ except that for $i=j=1$ (see (\ref{Char:resonState})), so we have
\begin{equation}\label{cst:a}
A^{(11)} =  \begin{pmatrix}
\ast& & &   \\
\vdots& & \widetilde{A}^{(11)}&\\
\ast& & &  \\
 a &*& \cdots & *\\
\end{pmatrix}_{m_{1}\times m_{1}},\ a= (i4\pi c_1)^{-1}\langle v_1^{(1)},JV1\rangle^2\neq 0,
\end{equation}
$A^{(i1)}$, $i=2\cdots,k, $ (resp., $A^{(1j)}$, $j=2,\cdots,k$) are sub-matrices with last row zero
(resp.,  first column), and
$$A^{(ij)} =  \begin{pmatrix}
0& & &   \\
\vdots& & \widetilde{A}^{(ij)}&\\
0& & &  \\
 0 &0& \cdots &0 \\
\end{pmatrix}_{m_{i}\times m_{j}},\ \forall 2\leq i,j\leq k.$$

We now  check the invertibility  of $E_{-+}(z)$ by following the same steps as before. We define $$ \Phi_{k-1} = (b_{ij})_{2\leq i,j\leq k}.$$ 
Here $b_{ij}:=-c_i\langle v_1^{(j)} ,v_1^{(i)} \rangle$, $2\leq i,j\leq k$, are computed in the same way as  coefficients $\beta_{ij}$ given in (\ref{bij})). Let $$L_{k-1}=(\langle v_1^{(j)}, J v_1^{(i)}\rangle)_{2\leq i,j\leq k}.$$ Then, we have
\begin{align}\label{det(k-1)}
\Phi_{k-1}= -   C_{k-1} L_{k-1} \ \text{with}\ \quad C_{k-1}= \text{diag} (c_2^{-1},\cdots,c_{k}^{-1}).
\end{align}
\begin{lem}\label{prop1.3} 
Assume $\rho>3$. Suppose that $zero$ is a singularity of the third kind of $H$ such that the eigenvalue zero has geometric multiplicity 
$k-1$, $k\in \bN$ with $k\geq 2$. We assume in addition that  $(H2)$ holds. Then 
\begin{equation} \label{detE-+3}
\det \,E_{-+}(z) =  \sigma_k\, z^{k-1/2} + o(|z|^{k-1/2}), 
\end{equation}
for $z \in \Omega_{\delta}$, $\delta>0$ small,
where $\sigma_k =  a \times  \det\, \Phi_{k-1} \neq 0$ with $a$ is the constant given in (\ref{cst:a}).
\end{lem}
\begin{proof}
We  proceed in the same way as in the proof of Lemma \ref{lem:detE-+cas1}. To avoid repetition we omit details.
First, we  define 
\begin{equation}\label{F:cas3}
E_{-+,2}(\eta) = N + \eta A + \eta^2 B,  \ Z(\eta) = \det E_{-+,2}(\eta),
\end{equation} 
for $\Im \eta>0$ and $|\eta|<\delta$. Now, we introduce the matrix 
$$\widetilde{L}(\eta) = \text{diag}(\widetilde{L}_1(\eta), \cdots, \widetilde{L}_k(\eta) )$$
with 
$$\ \widetilde{L}_1(\eta) = \text{diag}(1,\cdots,1, \eta^{-1}) \ \text{and} \
\widetilde{L}_i(\eta) =\text{diag}(1,\cdots,1, \eta^{-2}), \ i= 2,\cdots,k.$$
We denote 
\begin{equation}\label{F*L,Z(eta)}
\widetilde{E}_{-+,2}(\eta) = \widetilde{L}(\eta) E_{-+,2}(\eta)=
N + \widetilde{A}(\eta) +\widetilde{B}(\eta) \ \text{and} \
\widetilde{Z}(\eta) = \det \widetilde{E}_{-+,2}(\eta).
\end{equation} 
Then, it is non difficult to see that 
\begin{enumerate}
\item[(i)] $\widetilde{L}(\eta) N =N.$
\item[(ii)] $\widetilde{A}^{(ij)}(\eta) = 
\left \{
\begin{array}{cccc}
\begin{pmatrix}
0&0& \cdots&0\\
\vdots&\vdots & &\vdots\\
0&0 &\cdots&0 \\
a \, \delta_{ij} & \ast & \cdots & \ast\\
\end{pmatrix}&+\mathcal{O}(|\eta|) & if & i=1,1\leq j \leq k,\\
\\
 \mathcal{O}(|\eta|) & &if & 2\leq i \leq k, 1\leq j \leq k.
\end{array}
\right.$ 
\\
\\
\item[(iii)]$\widetilde{B}^{(ij)}(\eta)) =
\left \{
\begin{array}{cccc}
\begin{pmatrix}
0&0& \cdots&0\\
\vdots&\vdots & &\vdots\\
0&0 &\cdots&0 \\
b_{ij} & \ast & \cdots & \ast\\
\end{pmatrix}& +\mathcal{O}(|\eta|^2) & if & 2\leq i\leq k, \, 1\leq j\leq k,\\
\\
\mathcal{O}(|\eta|)&& if & i =1, 1\leq j \leq k.
\end{array}
\right.$ 
\end{enumerate}
  Let us calculate  
$\widetilde{Z}(0)$. We see  from $(i),(ii)$ and $(iii)$ that the block entries $\widetilde{E}_{-+,2}^{(ij)}(0)$, $1\leq i,j\leq k$, of the block matrix $\widetilde{E}_{-+,2}(0)$ have the following forms 
\begin{eqnarray}
\widetilde{E}_{-+,2}^{(1j)}(0)&=&
\begin{pmatrix}
0&-\delta_{ij}&0&\cdots &0\\
\vdots&\ddots&\ddots&\ddots&\vdots\\
0& &\ddots & \ddots&0\\
0 & &&0 & -\delta_{ij}\\
a\delta_{1j}& \ast&\cdots &\ast&0
\end{pmatrix}_{m_1\times m_j}, 
\ 1\leq j\leq k,\\
\widetilde{E}_{-+,2}^{(ij)}(0) &=&
\begin{pmatrix}
0&-\delta_{ij}&0&\cdots &0\\
\vdots&\ddots&\ddots&\ddots&\vdots\\
0& &\ddots & \ddots&0\\
0 & &&0 & -\delta_{ij}\\
b_{ij}& \ast&\cdots &\ast&0
\end{pmatrix}_{m_i\times m_j},\  2\leq i\leq k, 1\leq j\leq k .
\end{eqnarray}
This yields 
\begin{align}
\label{Z(0):cas3}
\widetilde{Z}(0)&= det \begin{pmatrix}
a&0&\cdots &0\\
b_{21}&b_{22}&\cdots& b_{2k}\\
\vdots &\vdots & & \vdots\\
b_{k1}& b_{k2}&\cdots & b_{kk}
\end{pmatrix} 
 =a\times \det (b_{ij})_{2\leq i,j\leq k} =
a \times \det\Phi_{k-1}.\nonumber 
\end{align}
Thus $\widetilde{Z}(0)\neq 0$ since $a\neq 0$ and in view of (\ref{det(k-1)}) det$\Phi_{k-1}$ does not vanish if the condition (2) of (H2) is satisfied. 
Finally, using det$L(\eta)= \eta^{-2k+1}$, we obtain
$$ \det\, E_{-+}(\eta) = \det\, E_{-+,2}(\eta) + o(|\eta|^{2k-1})= \left( a \times \det\,\Phi_{k-1}\right) \eta^{2k-1} + o(|\eta|^{2k-1}).$$
\end{proof}

\begin{lem}\label{prop:E-+cas3}
Let $\rho>\ell+1$, $\ell\in \bN$ with $\ell\geq 4$. Assume  that the hypotheses in the previous lemma hold.  Then
\begin{equation}\label{E-+dev:cas3}
E_{-+}(z)^{-1}=  \sum\limits_{j=-2}^{\ell-4} z^{j/2} F_j^{(3)} + E_{-+,\ell-4}^{(-1)}(z),
\end{equation}
for $z\in \Omega_{\delta}$, $\delta>0$ small, where $F_{-2}^{(3)}$ is a matrix of rank $k-1$ with
\begin{equation}\label{F-23}
 (F_{-2}^{(3)})^{(ij)} =
 \begin{pmatrix}
0&\cdots&0&\alpha_{ij}\\
0&\cdots&0&0\\
\vdots& &\vdots&\vdots\\
0&\cdots&0&0\\
\end{pmatrix}, \ 1\leq i,j\leq k, 
\end{equation}
 such that $\alpha_{1j}=\alpha_{i1}=0$, $\forall 1\leq i,j\leq k$,
and  $F_{-1}^{(3)}$ is a matrix of rank at most $ k$, where
\begin{equation}\label{F-13}
(F_{-1}^{(3)})^{(ij)}=   \begin{pmatrix}
0&\cdots&0& \mu_{ij}\\
0&\cdots&0&0\\
\vdots& &\vdots&\vdots\\
0&\cdots&0&0\\
\end{pmatrix}, \quad \forall 1\leq i,j\leq k,
\end{equation}
such that $\displaystyle \mu_{11}= \alpha_{11}= a^{-1}$ and $\mu_{1j}=0$ for $j=2,\cdots,k$. 
Moreover, the remainder term $E_{-+,l-4}^{(-1)}(z)$ satisfies the estimates in (\ref{Rest:E-+invcas1+}).
\end{lem}
\begin{proof}
 We have showed in the proof of Lemma \ref{prop1.3} 
that $\widetilde{E}_{-+,2}(0)$ is invertible if the hypothesis $(H2)$ is satisfied. We calculate that 
\begin{equation}\label{F(0)inverse:3}
\widetilde{E}_{-+,2}(0)^{-1} =  \frac{^tCom \widetilde{E}_{-+,2}(0)}{\det \widetilde{ E}_{-+,2}(0)}= \widetilde{F}_{-2}^{(3)} + \mathcal{E}^{(3)} ,
\end{equation}
where 
$$ (\widetilde{F}_{-2}^{(3)})^{(ij)} =\left \{
 \begin{array}{ccc}
\begin{pmatrix}
0&\cdots&0&\alpha_{ij}\\
0&\cdots&0&0\\
\vdots& &\vdots&\vdots\\
0&\cdots&0&0\\
\end{pmatrix} & ,& i=j=1\ \text{or}\ 2\leq i\leq k, 1\leq j\leq k,\\
\begin{pmatrix}
0
\end{pmatrix} &,& i=1, 2\leq j\leq k.
\end{array}\right. $$
 In particular $\alpha_{11}= a^{-1}.$ The matrix $\vE^{(3)}$ has the same form as $\vE^{(2)}$ in (\ref{F(0):inverse}).
 
On the other hand, for $\Im\eta>0, |\eta|<\delta$ with $\delta>0$ small, the matrix $\widetilde{E}_{-+,2}(\eta)$ can be developed as follows:
\begin{equation}\label{Ftilde:cas3}
 \widetilde{E}_{-+,2}(\eta) = \widetilde{E}_{-+,2}(0) + \eta \widetilde{E}_{-+,2}^{\sharp }(0) + \mathcal{O}(|\eta|^2),
 \end{equation}
where the $m_i\times m_j$ block entries of $\widetilde{E}_{-+,2}^{\sharp }(0)$ have the following form:
\begin{equation}\label{F1:cas3}
 (\widetilde{E}_{-+,2}^{\sharp }(0))^{(ij)}= \left \{ \begin{array}{cc}
\begin{pmatrix}
\ast& & &   \\
\vdots& &  \widetilde{A}^{(1j)}&\\
\ast& & &  \\
 b_{1j} &\ast& \cdots & \ast\\
\end{pmatrix}, & i=1, \, 1\leq j\leq k,\\
\\
\begin{pmatrix}
\ast& & &   \\
\vdots& &  \widetilde{A}^{(ij)}&\\
\ast & & &  \\
 0 &0& \cdots & 0\\
\end{pmatrix},&   2\leq i\leq k, 1\leq j\leq k.
\end{array}\right.\end{equation}
Then, it follows from (\ref{E-+:devcas1}), (\ref{F*L,Z(eta)}), (\ref{F(0)inverse:3}) and the regularity of $L(\eta)$ that for $\Im \eta>0$ and $|\eta|<\delta$, $\delta>0$ small enough, $E_{-+,2}(\eta)^{-1}$ exists with
\begin{align}\label{E-+2cas3}
E_{-+,2}(\eta)^{-1}&=  \frac{F_{-2}^{(3)}}{\eta^2} + \frac{F_{-1}^{(3)}}{\eta} + \mathcal{O}(1), 
\end{align}
where $(F_{-2}^{(3)})^{(ij)}= (\widetilde{F}_{-2}^{(3)})^{(ij)}$ if $2\leq i,j\leq k$ and  $(\widetilde{F}_{-2}^{(3)})^{(ij)}=0$ elsewhere. Also, 
$F_{-1}^{(3)} =\widetilde{F}_{-2}^{(3)}- F_{-2}^{(3)} - \widetilde{F}_{-2}^{(3)} \widetilde{E}_{-+,2}^{\sharp }(0) \widetilde{F}_{-2}^{(3)}$.\\

The rest of the proof follows directly from (\ref{E-+:devcas1}) and (\ref{E-+2cas3}).
Moreover, the same argument used to prove that the matrix $F_{-2}^{(2)}$ in (\ref{E-+inversecas3}) is of rank $k$ can be applied to the matrix $F_{-2}^{(3)}$ to show that it has rank $ k-1 $, unless in the present case we show  that $(\alpha_{ij})_{2\leq i,j\leq k}$ is a $(k-1)\times (k-1)$ invertible matrix. Indeed, we can check that
$$ \alpha_{ij} = \left(\det \Phi_{k-1}\right)^{-1}(-1)^{i+j}\times |(\Phi_{k-1})_{j-1}^{i-1} |, \ \forall 2\leq i,j \leq k,$$
which  denotes the $(i-1,j-1)$-th entry  of the invertible matrix $\Phi_{k-1}$ (see (\ref{det(k-1)})).
Thus 
\begin{equation}\label{phi_k inverse:cas3}
 (\alpha_{ij})_{2\leq i,j\leq k}= \Phi_{k-1}^{-1}.    
\end{equation}
\end{proof}
We end this section by proving Theorem \ref{thm:cas3}.
\\
\begin{proof}[Proof of Theorem \ref{thm:cas3}]
 Since the same proof of Theorem \ref{thm:cas1} can be done here, we will omit  details. Using the asymptotic expansion of $E_{-+}(z)^{-1}$  established in the previous lemma, for $\rho>2\ell-3, s>\ell-3/2$ and $z\in \Omega_{\delta}$, $\delta>0$ small, the expansion of $R(z)$ in $\B(-1,s,1,-s)$ is written
\begin{align*}
R(z)&= -\frac{\mathcal{P}_0^{(3)}}{z} 
 + \frac{1}{\sqrt{z}} \lbrace  i\langle \cdot,J\psi \rangle \psi + \sum_{i=2}^{k}\langle \cdot, J Y_i^{(3)}\rangle 
v_1^{(i)}\rbrace  \\&+
\sum\limits_{j=0}^{\ell-4} z^{j/2} R_j^{(3)}  + \widetilde{R}^{(3)}_{\ell-4}(z),
\end{align*}
where, by the help of the matrices defined in (\ref{det(k-1)}) and (\ref{phi_k inverse:cas3}), we have
\begin{align*}
\mathcal{P}_0^{(3)}&= \sum_{i=2}^{k} \langle \cdot, J \mathcal{Z}_i^{(3)}\rangle \mathcal{Z}_i^{(3)},\\
 \mathcal{Z}_i^{(3)} &=\sum_{j=2}^k q_{ij} v_{1}^{(j)}, \ 
\mbox{with} \ 
  \langle \mathcal{Z}_i^{(3)}, J\mathcal{Z}_j^{(3)} \rangle=\delta_{ij }, \ \forall 2\leq i,j\leq k,
 \end{align*}
 with  $Q_{k-1}:=(q_{ij})_{2\leq i,j\leq k}$ is such that $L_{k-1}^{-1}=\, ^tQ_{k-1} Q_{k-1}$. In addition 
\begin{align*}
\psi &= (-i c_1^{-1}\mu_{11})^{1/2} v_{1}^{(1)} = \frac{2\sqrt{\pi}}{\inp{v_1^{(1)},JV1}}v_1^{(1)},\ \text{and} \\
Y_{i}^{(3)}&=- \sum\limits_{j=1}^k \mu_{ij} G_0V\chi_{m_j}^{(j)}- \sum_{j=2}^k i \alpha_{ij} G_1V\chi_{m_j}^{(j)}= \sum\limits_{j=1}^k c_j^{-1} \mu_{ij}  v_{1}^{(j)}, \  i= 1,\cdots,k,
\end{align*}
where $\alpha_{ij}$ and $\mu_{ij}$ are respectively entries of the matrices $F_{-2}^{(3)}$ and $F_{-1}^{(3)}$ given in 
(\ref{F-23}) and (\ref{F-13}). Here we used $ G_1V\chi_{m_j}^{(j)}=0$ for $j=2,\cdots, k$ (see Remark \ref{Rem:c_j} and (\ref{Char:resonState})) and $G_0V v_1^{(j)}=- v_1^{(j)}, \ \forall 1 \leq j \leq k$. Set $S_{-1}^{(3)}= \sum_{i=2}^{k}\langle \cdot, J Y_i^{(3)}\rangle v_1^{(i)}$. Thus the resolvent expansion in (\ref{R(z):cas3}) is proved. 
Moreover, we refer to the proof of Theorem \ref{thm:cas1} for the properties of the remainder term $\widetilde{R}^{(3)}_{l-4}(z)$.
\end{proof}
\section{Resolvent expansions near positive  resonances}\label{section:real resonance}
In this section, we prove Theorem \ref{Thm:realRes1}. First, we use the condition (\ref{cond:detReson}) of the hypothesis $(H3)$ to  establish the asymptotic expansion of the
resolvent $R(z)=(H-z)^{-1}$ near an outgoing positive
resonance. Note that the study of  incoming 
positive resonance can be done in a similar way. \\

Let $\lambda_0 \in \sigma_r^{+}(H)$, for $\delta>0$ we denote 
\begin{equation*}
\Omega_{\delta}^+:= \lbrace z\in \C_+:\ |z-\lambda_0 | <\delta\rbrace \ \text{and} \ \overline{\Omega}_{\delta}^+:= \lbrace z\in \Bar{\C}_+,\ |z-\lambda_0 | <\delta\rbrace. 
\end{equation*}
Let us begin with the known results on the behavior of the free resolvent $R_0(z)$ on the boundary of the right half-plane (the real half axis $]0,+\infty[$). 
Taking the analytic continuation of the integral kernel $R_0(z)(x,y)$ to $\Bar{\mathbb{C}}_+\setminus \{0\}$, the expansion of $R_0(z)$ at order $r$ for every $r\in \mathbb{N}$ and for $z\in \Omega_{\delta}^+, \delta>0,$ is written
\begin{equation}\label{devRes:real}
R_0(z) = \sum\limits_{j=0}^r (z-\lambda_0)^j G_j^{+} + o(|z - \lambda_0|^r),
\end{equation}
where
\begin{align}
G_0^{+}&: \mathbb{H}^{-1,s} \longrightarrow  \mathbb{H}^{1,-s'}; \quad s,s'>1/2, \label{G_0^+}\\
G_j^{+}&: \mathbb{H}^{-1,s} \longrightarrow  \mathbb{H}^{1,-s'}; \quad s,s'>j+1/2,\label{G_1^+}, \ j=1,\cdots,r, 
\end{align}
are integral operators  with corresponding kernels
$$r_j^+(x,y,\lambda_0) := \lim\limits_{z\to \lambda_0, z\in \C_+} \displaystyle \frac{d^j }{dz^j} \frac{e^{+i\sqrt{z}
\vert x-y \vert} }{4\pi \vert x-y \vert}, \ j=0,1,\cdots,r.$$
For simplicity, we will use  in the following the variable $\xi= z-\lambda_0$.
We denote by $R_0(\lambda+i0)$ the boundary value of $R_0(z)$.\\

Let $\lambda_0$ be an outgoing positive resonance of the operator $H=-\Delta +V$. Note that the two subspaces Ker$(Id+K^+(\lambda_0))$ and $\{ \psi \in \text{Ker}(H-\lambda_0);\, \psi \, \text{satisfies the radiation} \text{ condition
(\ref{Behavior:psi}) with sign +} \}$ coincide in $H^{1,-s}$, $1/2<s<\rho-1/2$ (see Section \ref{sub:model}). Denote $ \text{dim Ker}(Id+K_+(\lambda_0))=N_0$. 
Let $\Pi_1^{\lambda_0}$ be the Riesz projection associated with the eigenvalue $-1$ of the operator $K^+(\lambda_0)$   
\begin{equation*}
\Pi_1^{\lambda_0} = \frac{1}{2i\pi} \int_{|w+1|=\epsilon} (w-K^+(\lambda_0))^{-1} dw, \ \epsilon>0,
\end{equation*}
we also denote $E_{\lambda_0}^+=\text{Ran}\, \Pi_1^{\lambda_0}$ and $m= \text{rank}\, \Pi_1^{\lambda_0}$.\\

The same strategy used in Section \ref{sec:zero resonance} to prove Theorem \ref{thm:cas1} and Theorem \ref{thm:cas3} will be followed in this section. 
First, note that the decomposition made in Lemma \ref{Lem:decompE}  can be done in the present case for $E_{\lambda_0}^+$ with just a change of notation $E_j^+(\lambda_0)$ instead of $E_j$  in (\ref{Lem:decomp1}). Let $\mathcal{U}_j(\lambda_0)=\{ u_1^{(\lambda_0,j)}, \cdots, u_{m_j}^{(\lambda_0,j)}\}$ denotes the given basis of $E_j^+$ in Lemma \ref{Lem:decompE}(2.) and  $\mathcal{W}_{j}(\lambda_0)=\{ w_1^{(\lambda_0,j)}, \cdots, w_{m_j}^{(\lambda_0,j)}\}$ its dual with respect to the non degenerate bilinear form $\Theta$. In particular, Ker$(Id+K^+(\lambda_0))$ is the subspace of $L^{2,-s}$ generalized by $\{u_1^{(\lambda_0,1)},\cdots,u_1^{(\lambda_0,N_0)}\}$.
Then, we will study
the Grushin problem for the  operator $Id+K(z)$, that we have constructed in Section \ref{sub:cas1}. 

For the proof of the theorem we start by the following lemma, where we refer to Section \ref{sub:cas1} for details.

\begin{lem}\label{Lem:E-+:res}
 For $\rho>l+1$,  $l=2,3,\cdots,$ and $z\in \Omega_{\delta}^+,\delta>0$ small, we have the following expansion of $E_{-+}(z)$: 
\begin{equation}\label{E-+:res}
E_{-+}(z) = N+ (z-\lambda_0) A(\lambda_0) + \sum\limits_{j=2}^l (z-\lambda_0)^j E_{-+,j}(\lambda_0) + E_{-+,l}(z-\lambda_0),
\end{equation}
where $N$ is the block matrix as defined in (\ref{N}) and
\begin{equation}\label{matrices:N A}
A^{(ij)}(\lambda_0)= \begin{pmatrix}
& & &\\
&&\widetilde{A}^{(ij)}(\lambda_0) &\\
&&&\\
a_{ij }(\lambda_0)& \ast& \cdots&\ast
\end{pmatrix}_{m_i\times m_j}, \ \forall 1\leq i,j\leq N_0,
\end{equation}
where 
\begin{align*}
 a_{ij }(\lambda_0)&=-\inp{G_1^+Vu_1^{(\lambda_0,j)}, JVw_{m_i}^{(\lambda_0,i)}}\\
 &= \frac{1}{i8\pi\sqrt{\lambda_0}c_i(\la_0)}\int_{\R^6} e^{+i\sqrt{\lambda_0}|x-y|}V(x)u_1^{(\lambda_0,i)}(x) V(y) u_1^{(\lambda_0,j)}(y) \, 
 dx dy,
 \end{align*}
such that $c_i(\lambda_0)\neq 0$ by Remark \ref{Rem:c_j} and $G_1^+$ is the integral operator defined in (\ref{G_1^+}).\\
 Moreover 
$E_{-+,l}(z-\lambda_0)$ is analytic in $\Omega_{\delta}^+$ and for $\la>0, |\la -\la_0|<\delta$, the limit $\lim\limits_{\epsilon\to 0^+}E_{-+,l}(\lambda-\lambda_0+i\epsilon)$ exists and satisfies
\begin{equation*}
\Vert \frac{d^r}{d\lambda^r} E_{-+,l}(\lambda-\lambda_0+i0) \Vert =
o(|\lambda-\lambda_0|^{l-r}),\ \forall\, |\la -\la_0|<\delta,\ r=0,1,\cdots,l.
\end{equation*}
\end{lem}
The expansion (\ref{E-+:res}) can be obtained directly by introducing in (\ref{E-+:def}) the expansion (\ref{devRes:real}) of $R_0(z)$ in $\mathbb{B}(0,s,0,-s)$ (see Section (\ref{sub:cas1}) for the details). \\

Before proving Theorem \ref{Thm:realRes1}, we establish the expansion of $E_{-+}(z)^{-1}$. 
Let $\rho>\ell+1, \ell=2,3,\cdots,$ and assume that the hypothesis $(H3)$ holds. We introduce the block diagonal matrix 
 \begin{equation}\label{matrix:L:res}
 \vL(\xi) = \text{diag} (\vL_1(\xi),\cdots, \vL_{N_0}(\xi) ), \ 
 \vL_i(\xi) = \text{diag} (1,\cdots,1,\xi^{-1}), \, \forall 1\leq i\leq N_0,
 \end{equation}
for $\xi \in \{\xi\in \C_+:|\xi|<\delta\}$. Then,  by proceeding in the same way as in the proof of Lemma \ref{lem:detE-+cas1} we obtain 
\begin{equation}\label{det:a(lambda)}
\det \left(N+\xi  A_0(\lambda_0)\right)= \xi^{N_0} \det \left(a_{ij}(\lambda_0)\right)_{1\leq i,j\leq N_0}+ \mathcal{O}(|\xi|^{N_0+1}),
\end{equation}
 where $a_{ij}(\lambda_0)$, $1\leq i,j\leq N_0$, are given above  and 
 $N_0=\text{dim ker}(Id+K^+(\lambda_0))$.
It follows that, if the condition (\ref{cond:detReson}) is satisfied then (\ref{E-+:res}) together with (\ref{det:a(lambda)}) gives
\begin{prop} 
If $\rho>\ell+1$, $\ell\in \bN$ with $\ell\geq 2$, we have
\begin{align*}
\det\, E_{-+}(z)&=  \sum\limits_{j=0}^{\ell} a_j (\lambda_0) (z-\lambda_0)^{N_0+j} +o( |z-\lambda_0|^{N_0+\ell}),
 \end{align*}  
 for $z\in \Omega_{\delta}^+$, where $a_0(\lambda_0)=\det \left(a_{ij}(\lambda_0)\right)_{1\leq i,j\leq N_0}\neq 0$.
\end{prop}
Since $E_{-+,1}(\xi,\lambda_0):= N+\xi A_0(\lambda_0)$ has the same form as $E_{-+,2}(\xi)$ defined in the proof of Theorem \ref{thm:cas1}, then the same computation  can be done here.
We   obtain
\begin{align}\label{F(z):res}
E_{-+,1}(\xi,\lambda_0)^{-1} &= \frac{^t \text{Com} E_{-+,1}(\xi,\lambda_0)}{\det E_{-+,1}(\xi,\lambda_0)} 
= \frac{1}{\xi} F_{-1}(\lambda_0) +  F_0(\lambda_0),  
\end{align}
where $F_{-1}(\lambda_0)$ is a $N_0\times N_0$ block matrix of rank $N_0$ with block entries
\begin{equation}\label{b_ij:reson}
F_{-1}^{(ij)}(\lambda_0)= \frac{1}{a_0(\lambda_0)}
\begin{pmatrix}
0&\cdots&0&b_{ij}(\lambda_0)\\
0&\cdots&0&0\\
\vdots & &\vdots &\vdots\\
0& \cdots& 0&0
\end{pmatrix}_{m_i\times m_j}
\end{equation}
where $b_{ij}(\lambda_0)$
are the $(j,i)$-th cofactors of the invertible matrix $(a_{ij}(\lambda_0))_{1\leq i,j\leq N_0}$ (see Remark \ref{prove of Rank}). 
On the other hand, by (\ref{E-+:res}), for $\ell\in \mathbb{N}^*, \rho>\ell+1$ and $z\in \Omega_{\delta}^+$, we have
\begin{align*}
E_{-+}(\xi,\lambda_0)&= E_{-+,1}(\xi,\lambda_0) \times \\
&\left[ I+ \left(E_{-+,1}(\xi,\lambda_0)\right)^{-1}\left( \sum\limits_{j=2}^{\ell} \xi^j E_{-+,j}(\lambda_0) + E_{-+,\ell}(\xi,\lambda_0)\right) \right]. 
\end{align*}
Then, by (\ref{F(z):res}), for $\xi\in \C_+$ and $|\xi|<\delta$ with $\delta>0$ small enough, $E_{-+}(\xi,\lambda_0)^{-1}$ exists,  with   
\begin{equation}\label{E-+:inverse:res}
 E_{-+}(\xi,\lambda_0)^{-1}=  \frac{1}{\xi} F_{-1}(\lambda_0) + 
 \sum\limits_{j=0}^{\ell-2} \xi^j \widetilde{E}_{-+,j}(\lambda_0) + \widetilde{E}_{-+,\ell-2}(\xi).
\end{equation}
where $\widetilde{E}_{-+,0}(\lambda_0)= F_0(\lambda_0) - F_{-1}(\lambda_0) E_{-+,2}(\lambda_0)F_{-1}(\lambda_0)$, the other terms $\widetilde{E}_{-+,j}(\lambda_0)$, $j=1,\cdots,\ell-2$, can be also directly found and the remainder $\widetilde{E}_{-+,\ell-2}(\xi)$ is analytic in $\{ \xi\in \C_+: |\xi|<\delta\}$ and for $\la>0$, $|\la-\la_0|<\delta$
\begin{equation*}
\Vert \frac{d^r}{d\lambda^r} \widetilde{E}_{-+,\ell-2}(\lambda-\lambda_0+i0) \Vert = o(|\lambda-\lambda_0|^{\ell-2-r}),\ r=0,1,\cdots, \ell-2.  
\end{equation*}

\begin{rem}\label{rem:A_k} 
\begin{enumerate}
\item We see that the coefficients $a_{ij}(\lambda_0)$, $1\leq i,j\leq N_0$, defined in Lemma \ref{Lem:E-+:res} can be expressed in the following matrix:
\begin{equation}\label{A_k}
\left(a_{ij}(\lambda_0) \right)_{1\leq i,j\leq N_0}= 
  \frac{1}{i8\pi \sqrt{\lambda_0}} C_{N_0} \mathcal{A}_{N_0}(\lambda_0),
\end{equation} 
where 
\begin{equation*}
\mathcal{A}_{N_0}(\lambda_0)= \left( B_{\lambda_0}( u_1^{(\la_0,j)}, u_1^{(\la_0,i)})\right)_{1\leq i,j\leq N_0}, \  C_{N_0}=\text{diag} \big(\frac{1}{c_1(\lambda_0)},\cdots,\frac{1}{c_{N_0}(\lambda_0)}\big),       
\end{equation*}
 and the bilinear form $B_{\lambda_0}$ is defined in (\ref{BilinearForm:B}).\\
 
\item The entries $b_{ij}(\la_0)$ of the matrix $F_{-1}(\lambda_0)$ in (\ref{b_ij:reson})  can be expressed
as follows
\begin{equation}\label{A_k:inverse}
  \left(b_{ij}(\lambda_0)\right)_{1\leq i,j\leq N_0}= i8\pi \sqrt{\lambda_0}\, a_0(\lambda_0) \mathcal{A}_{N_0}(\lambda_0)^{-1}C_{N_0}^{-1}.
\end{equation}
\end{enumerate}
\end{rem}
\begin{proof}[Proof of Theorem \ref{Thm:realRes1}]
Applying Grushin problem (\ref{Prblem:Inverse})
to $\mathcal{M}(z) :=Id+R_0(z)V$, it follows from 
(\ref{E-+:inverse:res}) that $\mathcal{M}(z)$ is invertible
for $z\in \Omega_{\delta}^+$, with 
$$\mathcal{M}(z)^{-1}= E(z)-(I-E(z)\mathcal{M}(z))S E_{-+}(z)^{-1}T(I-\mathcal{M}(z)E(z)).$$
 Moreover, for $\ell-1/2<s<\rho-\ell+1/2$, we have the following expansion
in $\B(\h^{1,-s})$ 
\begin{equation}\label{I+R:res}
(Id+R_0(z)V)^{-1} = -\frac{1}{z-\lambda_0} SF_{-1}(\lambda_0)T +  \sum_{j=0}^{\ell-2} (z-\lambda_0)^j \widetilde{F}_j(\lambda_0) + \widetilde{F}(z-\lambda_0)    
\end{equation}
where $F_{-1}(\lambda_0)$ is the matrix of rank $N_0$ given in (\ref{F(z):res}), so that for $g\in \H^{1,-s}$ 
\begin{equation}\label{SFT}
SF_{-1}(\lambda_0)Tg= \frac{1}{a_0(\lambda_0)}\sum_{j=1}^k  b_{ij}(\lambda_0)  \inp{g,JV w_{m_j}^{(\la_0,j)}} u_1^{(\la_0,j)},   
\end{equation}
and $ \widetilde{F}_0(\lambda_0)= -S\widetilde{E}_{-+,0}(\lambda_0)T + E(0)$
with  $$E(0)= \left( (I-\Pi_1^{\lambda_0})(Id+G_0^+V)
(I-\Pi_1^{\lambda_0})\right)^{-1}(I-\Pi_1^{\lambda_0}).$$
The remainder $\widetilde{F}(z-\lambda_0)$ is analytic in $\Omega_{\delta}^+$ and continuous up to $\R_+$ verifying the following estimates:
\begin{equation}\label{rest:F(z)}
 \Vert \frac{d^r}{dz^r}\widetilde{F}(z-\lambda_0)\Vert_{\B(\h^{1,-s}) }= o(|z-\lambda_0|^{\ell-2-r}), \ \forall z \in  \Omega_{\delta}^+, r=0,\cdots,\ell-2.  
\end{equation}
Moreover, for $\lambda \in \overline{\Omega}_{\delta}^+\cap\R_+$, the limit
\begin{equation*}
\lim_{\epsilon \to 0^+} \widetilde{F}(\lambda-\lambda_0+i\epsilon)= \widetilde{F}(\lambda-\lambda_0+i0)    
\end{equation*}
exists as operator in $\B(\h^{1,-s})$ and still satisfy the estimates in (\ref{rest:F(z)}).\\
Consequently, if $\rho>2\ell-1$ and $s>\ell-1/2$ using $R(z)=(Id+R_0(z)V)^{-1}R_0(z)$ and expansions 
(\ref{ResLibre:dev}), (\ref{I+R:res})  together with (\ref{A_k:inverse}) and (\ref{SFT}) , it follows that for $f\in \h^{-1,s}$ 
\begin{align*}
R(z)f&= \frac{R_{-1}(\lambda_0)f}{z-\lambda_0}+ \sum_{j=0}^{\ell-2} R_j(\lambda_0)f+ \widetilde{R}_{\ell-2}(z-\lambda_0)f
\end{align*}
in $\h^{1,-s}$, where 
\begin{align*}
R_{-1}(\lambda_0)f&=- \frac{1}{a_0(\lambda_0)}\sum_{i=1}^{N_0}\sum_{j=1}^{N_0}  \frac{  b_{ij}(\lambda_0)}{c_j(\la_0)} \inp{ f,JG_0^+ V u_{1}^{(\la_0,j)}} u_1^{(\la_0,i)}\\
&= \frac{1}{a_0(\lambda_0)}\sum_{i=1}^{N_0}\sum_{j=1}^{N_0}\frac{  b_{ij}(\lambda_0)}{c_j(\la_0)}  \inp{ f,J u_{1}^{(\la_0,j)}} u_1^{(\la_0,i)}\\
&= i 8\pi \sqrt{\lambda_0} \sum_{i=1}^{N_0} \inner f{J\phi_i(\lambda_0)} u_1^{(\la_0,i)},    
\end{align*}
such that by Remark \ref{rem:A_k} and a simple computation we have
$$\begin{pmatrix}
\phi_1(\lambda_0)\\
\vdots\\
\phi_{N_0}(\lambda_0)
\end{pmatrix} = \mathcal{A}_{N_0}(\lambda_0)^{-1} \begin{pmatrix}
u_1^{(\la_0,1)}\\
\vdots\\
u_{1}^{(\la_0,N_0)}
\end{pmatrix} , $$
where $\mathcal{A}_{N_0}(\lambda_0)^{-1}$ is given in (\ref{A_k:inverse}). Thus, to simplify the above sum we  decompose 
$$\mathcal{A}_{N_0}(\lambda_0)^{-1}=\, ^t Q(\lambda_0) Q(\lambda_0)$$ 
(see (\ref{matrix:Q})). Hence, $R_{-1}(\lambda_0)$ can be written in the following form
$$ R_{-1}(\lambda_0)f= \sum_{i=1}^{N_0} \langle f, J\psi_i(\lambda_0)\rangle  \psi_i(\lambda_0),$$ 
where 
$$\begin{pmatrix}
\psi_1(\lambda_0)\\
\vdots\\
\psi_{N_0}(\lambda_0)
\end{pmatrix}=(i 8\pi \sqrt{\lambda_0})^{1/2} Q(\lambda_0)\begin{pmatrix}
u_1^{(\la_0,1)}\\
\vdots\\
u_1^{(\la_0,N_0)}
\end{pmatrix}$$
with 
$$\frac{1}{i 8\pi \sqrt{\lambda_0}}B_{\lambda_0} (\psi_i(\lambda_0),\psi_j(\lambda_0))=\delta_{ij},$$
$B_{\lambda_0}(\cdot,\cdot)$ is the bilinear form defined in (\ref{BilinearForm:B}).  In addition,
$$R_0(\lambda_0)f=-\sum_{i=1}^{N_0}  \langle f,JG_1^+V \psi_i(\lambda_0)\rangle  u_1^{(\la_0,i)}+
\widetilde{F}_0(\lambda_0)G_0f.$$
Moreover, we see from (\ref{Lim ResLibre}) and (\ref{rest:F(z)}) that $\widetilde{R}_{l-2}(z-\lambda_0)f$ can be continuously extended to $\overline{\Omega}_{\delta}^+$ with estimates (\ref{Rest:R(z)res}).
\end{proof}

\paragraph{\textbf{General situation}} We end this section with a more general result if the condition (\ref{cond:detReson}) does not hold.
\begin{thm}\label{Thm:realRes2}
Let $\lambda_0$ be an outgoing positive  resonance of $H$.
Suppose that there exists an integer $\mu_0>0$ and a small $\delta>0$
such that 
\begin{equation}\label{cond:detE-+_Reson}
d(z):=\det \,E_{-+}(z)= (z-\lambda_0)^{\mu_0} g(z), \quad \forall  z\in \Omega_{\delta}^+,
\end{equation}
where $g$ is an analytic function on $\overline{\Omega}_{\delta}^+$ such that $g(\lambda_0)\neq 0$. 
Assume $\rho>2\ell+1$, $s>\ell+1/2$ with $\ell=\mu_0-N_0+q, q\in \mathbb{N}^*$.  We have 
\begin{equation}\label{R(z):res1}
R(z)= \frac{\mathcal{R}(\lambda_0)}{(z-\lambda_0)^{\mu_0-N_0
+1}} + \sum_{j=-\mu_0+N_0}^{q-1}(z-\lambda_0)^j R_{j}(\lambda_0)+ \widetilde{R}_{q}(z-\lambda_0)    
\end{equation}
in $\mathbb{B}(-1,s,1,-s)$, where $N_0$ is the geometric multiplicity of $-1$ as eigenvalue of $K^+(\lambda_0)$ and
$$ \mathcal{R}(\lambda_0) : L^{2,s} \rightarrow 
\text{Ker}( Id+K^+(\lambda_0)) \subset L^{2,-s}.$$
Moreover, the remainder term $\widetilde{R}_{q}(z-\lambda_0) $ is analytic in $\Omega_{\delta}^+$ and satisfies the estimates
\begin{equation}\label{estim:R(z)_Reson}
\Vert \frac{d^r}{d\lambda^r}\widetilde{R}_{q}(\lambda-\lambda_0+i0)\Vert_{\b(-1,s,1,-s)}= o(|\lambda-\lambda_0|^{q-1-r}),  
\end{equation} 
for $|\lambda-\lambda_0|<\delta$ and $r=0,1,\cdots,q-1.$\\
If $\rho>3$ and $s>3/2$, then we can obtain
\begin{equation}\label{dev:realRes1}
\displaystyle R(z) = \frac{\mathcal{R}(\lambda_0)}{(z-\lambda_0)^{\mu_0-N_0+1}} + 
\mathcal{O}(|z-\lambda_0|^{-\mu_0+N_0}), 
 \quad \forall  z\in \Omega_{\delta}^+.
\end{equation}
\end{thm}
\begin{proof}
Under the condition $(\ref{cond:detE-+_Reson})$ there exists 
$\delta>0$ such that for $z\in \Omega_{\delta}^+$
the $m\times m$ matrix $E_{-+}(z)$ is invertible and
$$ E_{-+}(z)^{-1} = \frac{M(z)}{d(z)}$$
where the entries of the matrix $M(z):= (a_{ij}(z))_{1\leq i,j\leq m}$ are polynomials of the entries of $E_{-+}(z)$  which are analytic in $z\in \Omega_{\delta}^+$.
Moreover, taking the expansion (\ref{E-+:res}) up to order $\ell$, $\ell\geq \mu_0-N_0+1$, we obtain
\begin{align*}\label{adjointE-+}
E_{-+}(z)^{-1} &=\frac{1}{d(z)}\left[
\sum_{j=0}^{\ell}
(z-\lambda_0)^{j+N_0-1} B_j(\lambda_0) 
\right] 
+\frac{1}{d(z)} E_{-+,\ell}^{-1}(z-\lambda_0)\\
&= \sum_{j=0}^{\ell} (z-\lambda_0)^{j-(\mu_0-N_0+1)} \widetilde{B}_j(\lambda_0)
+ \widetilde{E}_{-+,\ell}^{-1}(z-\lambda_0)
\end{align*}
where $N_0=$dim Ker $(Id+K^+(\lambda_0))$, 
\begin{equation*}
(\widetilde{B}_0(\lambda_0))_{ij} = \begin{pmatrix}
0&\cdots & 0 &\beta_i^j(\lambda_0)\\
0&\cdots &0 & 0\\
\vdots & &\vdots & \vdots \\
0&\cdots & 0 & 0
\end{pmatrix}, \quad 1\leq i,j\leq N_0,   
\end{equation*}
such that $\beta_i^j(\lambda_0)$ are polynomials of the entries of $A(\lambda_0)$ and the remainder term $\widetilde{E}_{-+,q}^{-1}(z-\lambda_0)$ is analytic in $z\in \Omega_{\delta}^+$ and satisfying
\begin{equation}\label{estim:restE-+Reson}
\Vert \frac{d^r}{d\lambda^r} \widetilde{E}_{-+,\ell}^{-1}(\lambda-\lambda_0+i0)\Vert = o(|\lambda-\lambda_0|^{q'-r}), 
 \end{equation}  
for $ |\lambda-\lambda_0|<\delta$, $r=0,1,\cdots,q'$ and $q'=\ell-(\mu_0-N_0+1)$.

In the rest of the proof we can proceed in the same way as in the previous proof. We obtain the leading term
$$\mathcal{R}(\lambda_0)= 
\sum_{i=1}^{N_0}\langle \cdot,\widetilde{\psi}_1(\lambda_0)\rangle  u_1^{(\la_0,i)}\ \text{on} \  L^{2,-s}, $$
where
$$\begin{pmatrix}
\widetilde{\psi}_1(\lambda_0)\\
\vdots\\
\widetilde{\psi}_{N_0}(\lambda_0)
\end{pmatrix}= \mathcal{B}_{N_0}(\lambda_0)
\begin{pmatrix}
u_1^{(\la_0,1)}\\
\vdots\\
u_1^{(\la_0,N_0)}
\end{pmatrix}, \quad \mathcal{B}_{N_0}(\lambda_0)= (\beta_{ij}(\lambda_0))_{1\leq i,j\leq N_0}C_{N_0},$$
(see (\ref{A_k})). Moreover, the estimate (\ref{estim:R(z)_Reson}) can be seen from (\ref{estim:restE-+Reson}).
\end{proof}
\section{Large-time expansion of the semigroup $(e^{-itH})_{t\geq 0}$}\label{sec:Representation}
In this section, we look at the asymptotic expansion in time, as $t\to +\infty$, of solutions to the Schr\"odinger equation (\ref{Eq:Schr}). We use the preceding results for the resolvent behavior on a contour surrounding the positive resonances in the upper half-plane, encircling the origin  and  down to the lower half-plane.\\ 

First, 
since $V$ satisfies the condition (\ref{V:decay}) for $\rho>2$, we can check that for  arbitrary small $\epsilon>0$, there exists $R_{\epsilon}>0$ large enough such that the numerical range of $H$ denoted by $\mathcal{N}(H)$ is included in an angular sector 
\begin{equation}\label{N(H)}
\mathcal{N}(H) \subseteq \{ z\in \C:\ \Re z\geq -R_{\epsilon},\ |\arg(z+R_{\epsilon})|\leq \frac{\epsilon }{2}\}.   
\end{equation}

We recall that
$\sigma_d(H)$ denotes the set of discrete eigenvalues of finite algebraic multiplicities (see Section
\ref{sub:model}). Denote $\sigma_d^+(H)=\sigma_d(H)\cap \bar{\mathbb{C}}^+$, whose accumulation points  can  exist only in $\sigma_r^+(H)\cup \{ 0\}$ (see Definition \ref{def:reson}).\\

We deduce from our previous main results that $H$ has a finite number of discrete eigenvalues in the closed upper half-plane.
\begin{prop}\label{vp fini}
Assume that $\rho>2$ and the hypothesis $(H3)$ on positive resonances holds.  If zero is an eigenvalue of $H$ we assume in addition that $\rho>3$ and $(H1)$ or $(H2)$ holds.  Then $$ \sigma_d^+(H) \, \text{is finite.}$$
\end{prop}
\begin{proof}
Since the eigenvalues located in the closed upper half-plane can only accumulate at points of
$\{0\}\cup \sigma_r^+(H)$, then to show that there is at most a finite number of these eigenvalues it suffices to prove that zero and the outgoing positive resonances of $H$ are not accumulation points of $\sigma_d^+(H)$ even though zero is an eigenvalue or/and a resonance.
Indeed, in view of Theorem \ref{Thm:realRes1} for each $\lambda_j\in \sigma_r^+(H)$ there is small $\delta>0$ such that if $\rho>2$ and $s>1/2$ $\langle x\rangle^{-s}R(z)\langle x\rangle^{-s}$ is uniformly bounded in $z$ on every compact set in $\Omega_{\delta}^+\subset \C_+$ given in (\ref{omega+}).  Then we have found a set $\Omega_{\delta}^+\subset \C_+$ which does not contain any pole of $\langle x\rangle^{-s}R(z)\langle x\rangle^{-s}$. Thus $\Omega_{\delta}^+\cap \sigma_d(H)= \emptyset$. Hence $\lambda_j$ is not an accumulation point of $\sigma_d(H)\cap \Bar{\C}_+$. 
Also, if zero is a resonance of $H$ then it is seen from Theorem \ref{thm:cas2} that the same argument can be done.  Moreover, if 
zero is an eigenvalue of $H$ and the hypothesis $(H1)$ holds, using the same argument, it follows 
from Theorem \ref{thm:cas1} with $\rho>3$ that  $\sigma_d(H)\cap \Omega_{\delta} =\emptyset$ for some $\delta>0$ small enough, where
$\Omega_{\delta}$ is given in (\ref{omega delta}).  We can check also that zero is not an accumulation point of  $\sigma_d(H)\cap \Bar{\C}_+$ if it is both an eigenvalue and a resonance of $H$ under the hypothesis $(H2)$.  
\end{proof}
In addition, we check  the existence of the limiting absorption principle for the non-selfadjoint Schr\"odinger operator $H$ on each subinterval of $\R_+$ which does not contain any outgoing positive  resonance. 
We also establish high energy estimates of the derivatives of the resolvent.\\

In the following, we denote by $\mathcal{C}^j(\Omega, F)$ the set of all functions $f: \Omega \subset E \longrightarrow F$ that is of class $C^j$ on $\Omega$, where $E$ and $F$ denote normed vector spaces.
And for $a>0$, we define the open set $\Lambda_{a}$ and its closure $\Bar{\Lambda}_{a}$ by
 $$\displaystyle \Lambda_{a}=\bigcap\limits_{j=0}^N \{ z\in \C_+: | z -\lambda_j|>a\},\  \Bar{\Lambda}_{a}=\bigcap\limits_{j=0}^N \{ z\in \Bar{\C}_+: | z-\lambda_j|>a\},$$ 
 where $ \lambda_0=0$ and $\lambda_j\in \sigma_r^+(H), \forall 1\leq j\leq N$.\\

The following proposition gives the high energy resolvent estimate as $|z|\to +\infty$, when $R(z)$ is extended through the upper half-plane to $\bar{\Lambda}_{a}$ for $a>0$. It can be proved in the same way as in \cite[Theorem 9.2]{jensen1979spectral} (see also \cite[Theorem 3.8]{komech2013dispersive}). 
\begin{prop}\label{PALlem}
Assume that $(H3)$ holds.
Let $\ell \in \mathbb{N}$. 
If $\rho> \ell+1$, then for $s>\ell+\frac{1}{2}$ and $f,g\in L^{2,s}$:
$$\Lambda_a \ni z \mapsto \langle R(z)f,g\rangle  \textit{ can be continuously extended to a function in } \ \mathcal{C}^{\ell}\left(\bar{\Lambda}_a; \mathbb{C}\right).$$ Moreover, the boundary values $\langle R(\lambda+i0)f,g\rangle $ satisfy the following  estimates:
\begin{equation}\label{estim-res-2}
\vert \langle \frac{d^{\ell}}{d\lambda^{\ell}} R(\lambda + i0)f,g\rangle  \vert \leq 
\frac{C_a}{|\lambda|^{\frac{\ell+1}{2}}} \Vert f\Vert_{0,s} \Vert g\Vert_{0,s},\  \lambda \in \Bar{\Lambda}_a\cap \R_+, \ \lambda\to +\infty,
\end{equation}\label{part:i}
for some constant $C_a>0$.
\end{prop}
The existence of the above limit is a direct consequence of the two main theorems \ref{thm:cas3} and \ref{Thm:realRes1} and the following known results: 
for $f\in L^{2,-s}$, $h\in L^{2,s}$ and $1/2<s<\rho-1/2$ the functions $z \mapsto \langle (Id+R_0(z)V)^{-1}f,h\rangle $ and  $z\mapsto \langle R_0(z)Vf,h\rangle$  
can be continuously extended to uniformly bounded functions on $\bar{\Lambda}_a$ (see\cite[Lemma 9.1]{jensen1979spectral}). In addition to the following estimates in 
\cite[Theorem 8.1]{jensen1979spectral}
$$|\frac{d^r}{d\lambda^r} \langle R_0(\lambda\pm i0)f,g\rangle |\leq \frac{C}{|\lambda|^{\frac{r+1}{2}}} \Vert f\Vert_{0,s} \Vert g\Vert_{0,s},\ r\in \mathbb{N},\ \text{as} \, \lambda \to +\infty. $$

 Before proving Theorem \ref{Thm:representation}, we must establish a representation formula for the Schrodinger semigroup $e^{-itH}$, as $t\to +\infty$, generated by the non-selfadjoint operator $H$.
Our representation formula is based on 
the Dunford-Taylor integral (cf. \cite[Section IX.1.6]{kato1966perturbation}), which is valid for m-sectorial operator, which means whose numerical range is a subset of a sector $\{ | \arg z |\leq \theta < \frac{\pi}{2}\}$. 
 More precisely, we will find a curve $\Gamma^{\nu}(\eta)$ such that
 $\Gamma^{\nu}(\eta) \cap 
\left( \sigma_d^+(H) \cup \sigma_r^+(H)\cup \{0\}\right) =\emptyset$ and
$$\Gamma^{\nu}(\eta) := \Gamma_-^{\nu}(\eta) \cup\Gamma_0(\eta) \cup \Gamma_1(\eta) 
\cup \Gamma_+,$$
 where
\begin{align*}
\Gamma_-^{\nu}(\eta)&= \lbrace z=  a(\eta)-\lambda e^{i\nu}, \ \lambda \geq 0 \rbrace,\ a(\eta)= 2\eta -i \eta \sin\eta, \\
\sigma_-(\eta)& = \{ z=\lambda -i \eta \sin\eta, \ \eta \cos \eta \leq \lambda\leq 2\eta \},\\
\Gamma_0(\eta) &= \lbrace z= \eta  e^{i(2\pi-\theta)},  \ \eta <  \theta
< 2\pi \rbrace, \\
\Gamma_1(\eta) &= \cup_{j=1}^N \left(\sigma_j(\eta) \cup 
\gamma_j(\eta) \right),\ 
\sigma_1(\eta) = \lbrace z= \lambda +i0, \
\eta  \leq \lambda \leq  \lambda_1 - \eta \rbrace, \\
\sigma_j(\eta) &= \lbrace z= \lambda +i0, \ \lambda_{j-1}
+\eta \leq \lambda \leq  \lambda_j - \eta \rbrace, \ j=2,\cdots ,N,\\ 
\gamma_j(\eta) &= \lbrace z= (\lambda_j + \eta  
e^{i(\pi-\theta)}),\  0 <  \theta < \pi \rbrace,\ j=1,\cdots,k ,\\
\Gamma_+&= \lbrace z= \lambda +i0 , \ \lambda \geq 
\lambda_N+ \eta \rbrace,
\end{align*}
for some $\eta>0$ and $\nu\in ]0,\frac{\pi}{2}[$ chosen so that there are no
eigenvalues of $H$ between  $\sigma_-(\eta)\cup\Gamma_0 (\eta)\cup \Gamma_1(\eta)$ and
the real axis, nor between  $\Gamma_-^{\nu}(\eta)$ and  the negative real axis. See Figure 1.\\
 
\begin{figure}[h!]
\centering
\begin{tikzpicture}[scale=0.7]
	[domain=0:5]
	\draw[->,very thin, gray] (-4,0) -- (11,0) node[right] {\textcolor{black}{$\R$}};
	\draw[->,very thin, gray] (0,-2) -- (0,2) node [above,black] {$i\R$};
	\draw[-] (0.5,0) -- (1.5,0)node[midway,above, sloped]{\small $\sigma_1(\eta)$} ;
	\draw[-] (2.5,0) -- (3.5,0) ;
	\draw[-]  (4.5,0) -- (7.5,0) node[midway, sloped, gray] {$\small >$} ;
	\draw[->] (8.5,0) -- (11,0);
	\draw[-] (1,-0.25) -- (-2,-2) node[above]{$\small \Gamma_-^{\nu}(\eta)$} node[midway, sloped, gray] {$>$};
\draw[-] (-35:0.5) -- (1,-0.25) node[below]{ $\small \sigma_-(\eta)$};
	\draw(0.5,0) arc (0:325:0.5) node[midway,above left]{\small $\Gamma_0(\eta) $} node[midway, sloped, gray] {$>$};
	\draw (2.5,0) arc (0:180:0.5)node[midway, sloped, gray] {\small $>$} ;
	\draw (4.5,0) arc (0:180:0.5)node[midway,above]{\small $\gamma_2(\eta) $} ;
	\draw (8.5,0) arc (0:180:0.5) ;
	\draw[thick] plot[mark=*] (2,0)
	node[below]{$\small \lambda_1$};
	\draw[thick] plot[mark=*] (4,0) node[below]{$\small \lambda_2 $};
	\draw[thick] plot[mark=*] (8,0)node[below]{$\small\lambda_N$};
	\draw (0,0) node{$\bullet$};
\draw[thick] plot(-1.5,0) node{$\times$} ;
    \draw[very thick] plot(1,2) node{$\times$} ;
    \draw[thick] plot(4,2.5) node{$\times$};
    \draw[thick] plot(7,2) node{$\times$};
    \draw[thick] plot(6.5,1) node{$\times$};
    \draw[thick] plot(9,2) node{$\times$};
    \draw[thick] plot(-3.5,1.5) node{$\times$};
    \draw[thick] plot(-2,2) node{$\times$};  
\end{tikzpicture}
\caption{
The curve $\Gamma^{\nu}(\eta)$. $\la_1,\cdots,\la_N$ are the outgoing positive resonances of $H$.}

\end{figure}
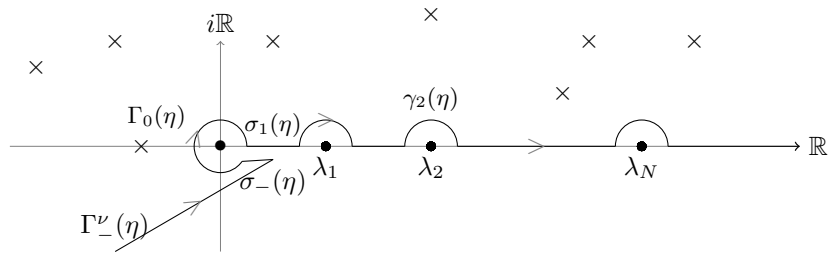
We now state the intermediate theorem
\begin{thm}\label{Thm:representation}
Let $\rho>3$ and $s>5/2$. Assume that hypothesis $(H3)$ on positive resonances holds. If zero is an eigenvalue of $H$  we assume in addition  that $(H1)$ or $(H2)$ holds.
Then, for $f$ and $g$ $\in L^{2,s}$,
we have the following representation formula:
\begin{align}\label{RepresFormula} 
 \langle e^{-itH}f,g\rangle  &= \sum\limits_{z_j\in \sigma_d^+(H)} \langle e^{-itH} \Pi_{z_j}f,g\rangle \\
 &+ \frac{1}{2i\pi} \int_{\Gamma^{\nu}(\eta)} e^{-itz}\langle  (H-z)^{-1}f,g\rangle \,dz, \quad \forall t>0,\nonumber 
\end{align}
 where $\sigma_d^+(H)$ is the finite set of discrete eigenvalues of $H$ located in the closed upper half-plane with associated Riesz projections $\lbrace\Pi_{z_j}\rbrace_{j}$ and $\Gamma^{\nu}(\eta)$ is the above curve (see Figure 1.).
\end{thm}

\begin{proof}
We proceed in 3 steps:\\
\\
\underline{First step}:
Let $\epsilon>0$. Let $P_{\epsilon}
=i(e^{-i\epsilon}H -i\epsilon R_{\epsilon})$. Then it can be seen from (\ref{N(H)}) that
$$N(P_\epsilon) \subseteq \lbrace  i(e^{ -i\epsilon} z -i\epsilon R_{\epsilon}), \Re z \geq - R_{\epsilon},
\vert \arg(z+R_{\epsilon}) \vert \leq \frac{\epsilon}{2} \rbrace \subset \bar{S}_{\theta_{\epsilon}},$$
 where $S_{\theta_{\epsilon}}$ denotes the open sector with angle $\theta_{\epsilon}$. Let $\theta_{\epsilon}= (\pi - \arctan{\epsilon})/2 \in ]0, \pi/2[$.
Moreover, for
$\lambda_{\epsilon}:=e^{i(\pi/2+\epsilon)}(\lambda+R\epsilon) \in \mathbb{C}\setminus N(H)$ with $\lambda>0$ large,  it can be seen that 
$(P_{\epsilon}+\lambda)=ie^{-i\epsilon} (H-\lambda_{\epsilon})$ is  a bijection of $\mathbb{H}^2$ into $L^2$. This shows that $P_{\epsilon}$ is  m-sectorial with semi-angle $\theta_{\epsilon}$. Hence,  $-P_{\epsilon}$ is the unique generator of the semigroup $(e^{-tP_{\epsilon}})_{t\geq 0}$, which is bounded by $\Vert e^{-tP_{\epsilon}}\Vert \leq 1$ (cf. \cite[Theoreme IX.1.24]{kato1966perturbation}).

Therefore, there exists a closed  curve $\Gamma$, oriented in the anticlockwise sense, included in the resolvent set of $-P_{\epsilon}$ and enclosing the numerical range of $-P_{\epsilon}$ in its interior,  where  $\Gamma = \lbrace \lambda
e^{-i(\pi-\theta_{\epsilon}- \delta)}, \lambda \geq 0\rbrace \cup \lbrace \lambda e^{i(\pi-\theta_{\epsilon}- \delta)}, \lambda \geq 0\rbrace $ for some $0< \delta < \frac{\pi}{2} - \theta_{\epsilon}$,
such that the semigroup integral representation  is written:
\begin{equation}\label{Rep1}
e^{-tP_{\epsilon}}u = \frac{1}{2i\pi} \int_{\Gamma}
e^{tz}(P_{\epsilon}+z)^{-1}udz,\  \forall u \in L^2,\
\ \forall t> 0,
\end{equation}
and we have the following estimate   
\begin{equation}\label{estimU(t)}
\Vert e^{-ite^{-i\epsilon}H} u \Vert_{0} \leq e^{tR_{\epsilon} \epsilon}\Vert u \Vert_0,
\ t \geq 0. 
\end{equation}
Denote by  $\Pi_{z_j}:L^2\to L^2$ the Riesz  Projection (\ref{Riesz}) associated with eigenvalue $z_j\in \sigma_d^+(H)$.  $\sigma_d^+(H)$ is a finite set by Proposition \ref{vp fini} and $H_j:=H\Pi_{z_j}$ defines a bounded operator on the finite dimensional subspace Ran $\Pi_{z_j}$, with $\sigma(H_j)=
\lbrace z_j\rbrace$ (see \cite[p. 178-179]{kato1966perturbation}).
Let $H_{\epsilon} = e^{-i\epsilon}H$ and $z_{\epsilon}=
e^{-i\epsilon}z$. By analytic deformation in $\rho(H)$ of the curve $\Gamma$, we can find a set of curves $\cup_{j=1}^p \Sigma_{j}$ around the eigenvalues $ z_1,\cdots, z_p \in \sigma_d^+(H)$ located  above a curve $\Gamma^{\nu}(\eta,\epsilon)$ (defined below) oriented in the anti-clockwise sense such that
\begin{align}
 U_{\epsilon}(t)u&:= e^{-itH_{\epsilon}}u \nonumber\\ 
 &= -\frac{1}{2i\pi}  \sum_{z_j\in \sigma_d^+(H)}
\int_{\Sigma_j} e^{-itz_{\epsilon}} (H_j-z)^{-1}u dz \nonumber \\
 &+ \frac{1}{2i\pi} \int_{\Gamma^{\nu}(\eta,\epsilon)}
 e^{-itz_{\epsilon}} (H-z)^{-1}u dz  \label{intg}\\
 & =   \sum_{z_j\in \sigma_d^+(H)} 
 e^{-i t H_{\epsilon}} \Pi_{z_j} u +
 \frac{1}{2i\pi} \int_{\Gamma^{\nu}(\eta,\epsilon)} e^{-i t z_{\epsilon}}(H-z)^{-1}
 u \ dz, \quad \forall t > 0,\nonumber
 \end{align}
where $\Gamma^{\nu}(\eta,\epsilon)$ is a closed curve oriented from $-\infty$ to $+\infty$ and $\Gamma^{\nu}(\eta,\epsilon) = \Gamma_-^{\nu}(\eta) \cup \Gamma_0(\eta,\epsilon) \cup \Gamma_1(\eta,\epsilon) 
\cup \Gamma_+(\epsilon)$:
\begin{align}
\Gamma_-^{\nu}(\eta)&= \lbrace z=  a(\eta)-\lambda e^{i\nu}, \ \lambda \geq 0 \rbrace,\ a(\eta)= 2\eta -i \eta \sin\eta,\nonumber \\
\sigma_-(\eta)& = \{ z=\lambda -i \eta \sin\eta, \ \eta \cos \eta \leq \lambda\leq 2\eta \}\nonumber\\
\Gamma_0(\eta,\epsilon)&=\{ z= \eta e^{i(2\pi-\theta)}, \ 
\eta <\theta< 2\pi -\epsilon_{\eta} \}, \ \epsilon_{\eta}= \arcsin (\epsilon/\eta),\label{curve:gamma}\\
\Gamma_1(\eta,\epsilon)&= \cup_{j=1}^N \left(\sigma_j(\eta,\epsilon) \cup \gamma_j(\eta,\epsilon) \right),\nonumber\\
\sigma_1(\eta,\epsilon) &= \lbrace z= \lambda +i\epsilon , \
\eta \cos{\epsilon_{\eta}} \leq \lambda \leq  \lambda_1 - \eta \rbrace, \nonumber\\
\sigma_j(\eta,\epsilon) &= \lbrace z= \lambda +i\epsilon , \ \lambda_{j-1}
+\eta \leq \lambda \leq  \lambda_j - \eta \rbrace,  j=2,\cdots,N,\nonumber \\ 
\gamma_j(\eta,\epsilon) &= \lbrace z= \lambda_j +i\epsilon + \eta e^{i(\pi-\theta)},\  0 <  \theta < \pi \rbrace, \ j=1,\cdots, N, \nonumber\\
\Gamma_+(\epsilon)&= \lbrace z= \lambda  +i\epsilon, \ \lambda \geq
\lambda_N+ \eta \rbrace,\nonumber
\end{align}
for some fixed $0<\nu<\frac{\pi}{2}$ and $\eta,\epsilon>0$ small chosen so  that $\Gamma^{\nu}(\eta, \epsilon) \cap \sigma (H) = \emptyset$ and there is no eigenvalues of $H$ between $ \Gamma_0(\eta,\epsilon) \cup  \Gamma_1(\eta,\epsilon) \cup \Gamma_+(\epsilon)$ and the positive real axis, nor between $\Gamma_-^{\nu}(\eta)$ and the negative real axis.\\ 
\\
\underline{Second step}: Let $f,g$ $\in L^{2,s}$. We define
\begin{align*}
 \langle  U(t)f,g\rangle &:=\sum_{z_j \in \sigma_d^+(H)} \langle e^{-itH} \Pi_{z_j}f,g\rangle +\frac{1}{2i\pi} \int_{\Gamma^{\nu}(\eta)} e^{-itz}\langle  (H-z)^{-1}f,g\rangle  dz. \end{align*}
In this step we will show that 
\begin{equation}\label{limU(t)}
\langle U(t)f,g\rangle = \lim\limits_{\epsilon \to 0^+} \langle  U_{\epsilon}(t) f, g\rangle , \quad \forall f,g \in L^{2,s},\quad \forall t>0.  
\end{equation}
Let us show the convergence of the integral in (\ref{intg}) as $\epsilon\to 0^+$ by decomposing it onto  two parts:
$ \Gamma^{\nu}(\eta,\epsilon) \cap \lbrace |z|\leq R_1 \rbrace $ and $ \Gamma^{\nu}(\eta,\epsilon) \cap
\lbrace |z|> R_1 \rbrace $, where $R_1> \lambda_N + 1$ with $\lambda_N := \max
\sigma_r^+(H)$.\\

On one hand, if zero is an eigenvalue of $H$ Theorem \ref{thm:cas1} gives the uniformly boundedness of the resolvent on $\Gamma_0(\eta)$ in $\b(0,s,0,-s)$ if $\rho>3$ and $s>3/2$. In addition,
Theorems \ref{Thm:realRes1} shows that the resolvent is uniformly
bounded in $z$ on each semicircle $\gamma_j(\eta)$ surrounding the
singularity $\lambda_j$ on the positive real axis, where a weaker
assumption is required, so that $\rho>2$. Moreover, by
Proposition \ref{PALlem} the integrand
$g_{\epsilon}(z,t):=e^{-ite^{-i\epsilon} z}\langle
(H-z)^{-1}f,g\rangle $ is continuously extended  to an uniformly
bounded function  in $z$ on $\sigma_-(\eta)\cup \bigcup\limits_{j=1}^N \sigma_j(\eta)$ for $f,g\in L^{2,s}, 1/2<s<\rho-1/2$. Then in view of these
results
\begin{align}\label{E2}
\int_{\Gamma^{\nu}(\eta,\epsilon) \cap\lbrace |z|\leq R_1 \rbrace} g_{\epsilon}(z,t) \, dz & \underset{\epsilon \to 0^+}{\longrightarrow}
\int_{\Gamma^{\nu}(\eta) \cap\lbrace |z|\leq R_1 \rbrace} 
e^{-itz}\langle  (H-z)^{-1}f,g\rangle  dz& .
\end{align}
On the other hand, we have
\begin{equation}\label{decomposition_Integ}
\int_{\Gamma^{\nu}(\eta,\epsilon) \cap\lbrace |z|> R_1 \rbrace}= 
\int_{ \Gamma_-^{\nu}(\eta) \cap\lbrace |z|> R_1 \rbrace} + \int_{\Gamma_+(\epsilon)\cap\lbrace |z|> R_1 \rbrace }. 
\end{equation}
Since $g_{\epsilon}(z,t)$ is uniformly bounded in $\epsilon\in]0,\frac{\nu}{2}[$ on $\Gamma_-^{\nu}(\eta) \cap\lbrace |z|> R_1 \rbrace$ with
\begin{align*}
\vert e^{-it a(\eta)e^{-i\epsilon}} e^{it\lambda e^{i(\nu-\epsilon)}}
\langle (H -(a(\eta) &-\lambda e^{i\nu}))^{-1}f,g\rangle  \vert 
\\
&\leq C_{\eta,R_1} \Vert f\Vert_{0,s} \Vert g\Vert_{0,s} e^{-t\lambda \sin(\frac{\nu}{2})}
\lambda^{-1/2},
\end{align*}
$\displaystyle \forall \lambda\in [R_1,+\infty[$, $\displaystyle \eta>0$ and $ t>0$ where the function at the right-hand side is integrable on
$[R_1,+\infty[$  for all $\eta>0$ small and $t>0$. Then we deduce by Lebesgue's dominated convergence theorem 
\begin{equation}\label{E4}
\lim_{\epsilon\to 0^+}\int_{\Gamma_-^{\nu}(\eta) \cap\lbrace |z|> R_1 \rbrace} g_{\epsilon}(z,t) \, dz =\int_{\Gamma_-^{\nu}(\eta) \cap\lbrace |z|> R_1 \rbrace} 
e^{-itz}\langle  (H-z)^{-1}f,g\rangle  dz.
\end{equation}
However,  the integrand of the second integral tends to 
$$\displaystyle e^{-it\lambda} \langle  (H-(\lambda+i0))^{-1}
f,g\rangle \quad \text{as} \quad \displaystyle \epsilon \rightarrow 0^+$$
which by Proposition \ref{PALlem} belongs to 
$\displaystyle \mathcal{C}^{2}([R_1, +\infty[, \mathbb{C}),$ 
with the following estimate:
$$ \vert \frac{d^2}{d\lambda^2} \langle (H-(\lambda+i\epsilon))^{-1}f,g\rangle   \vert \leq \frac{C_{R_1}}{\langle \lambda\rangle^{3/2}}
\Vert f\Vert_{0,s} \Vert g\Vert_{0,s}, \ \forall\epsilon>0,$$
that requires $\rho>3$ and $s>5/2$.
Then, for $t>0$ fixed we integrate twice by parts to obtain 
\begin{align} 
\int_{R_1}^{+\infty} (-i t e^{-i\epsilon})^{-2}&
 e^{-it\lambda e^{-i\epsilon}} \frac{d^2}{d\lambda^2}
 \langle (H- (\lambda+i\epsilon))^{-1}f,g\rangle  d\lambda 
 \nonumber\\
&+ \mathcal{O}(t^{-2} |e^{-i(t R_1 e^{-i\epsilon}+\epsilon) }|) \Vert f\Vert_{0,s} \Vert g\Vert_{0,s}\nonumber\\
&:= \int_{R_1}^{+\infty} f_{\epsilon} (t,\lambda) d\lambda  + \mathcal{O}(t^{-2} |e^{-i(t R_1 e^{-i\epsilon}+\epsilon) }|) \Vert f\Vert_{0,s} \Vert g\Vert_{0,s}\label{Int by Part}.
\end{align}
Since for $\rho>3$, $s>5/2$ and $t>0$, $f_{\epsilon}(t,.)$ is uniformly bounded  in small $\epsilon>0$  as  
$$ \vert f_{ \epsilon} (t,\lambda)\vert \leq  \frac{C_{R_1}}{t^2}  \frac{1}{\langle \lambda\rangle^{3/2}}
\Vert f\Vert_{0,s} \Vert g\Vert_{0,s},$$
 then by Lebesgue's dominated convergence theorem   $\displaystyle \int_{R_1}^{+\infty}f_{ \epsilon} (t,\lambda) d\lambda$ converges as $\epsilon \rightarrow 0$, also the second term at the right-hand side of (\ref{Int by Part})
is uniformly bounded in $\epsilon>0$ by $\displaystyle \mathcal{O}(t^{-2})\Vert f\Vert_{0,s}
\Vert g\Vert_{0,s}$. This shows that  the second integral at the right-hand side of (\ref{decomposition_Integ}) is uniformly convergent in $\epsilon >0$ for all $t>0$.
%
Consequently, this together with  (\ref{E2}) and (\ref{E4}) implies
$$ \int_{\Gamma^{\nu}(\eta,\epsilon)} g_{\epsilon}(z,t) \, dz
\underset{\epsilon \to 0^+}{\longrightarrow}
\int_{\Gamma^{\nu}(\eta)} e^{-itz}\inp{ (H-z)^{-1}f,g} dz,$$
 which must establish (\ref{limU(t)}).\\
 
Finally, we will show at the third step that 
for all $f$ and $g$ $ \in L^{2,s}$ we have the following convergence
$$  \inp{ e^{-itH_{\epsilon}}f - e^{-itH}f, g }
 \underset{\epsilon \to 0^+}{\longrightarrow} 0, \quad 
\forall t >0.$$
\\
\underline{Third step:} Let $\phi$ be a test function in $\mathcal{C}^{\infty}_0(\R^3)$. We write
 \begin{align*}
 e^{-itH_{\epsilon}}\phi - e^{-itH}\phi &= 
\int_0^t \frac{d}{dr} \left( e^{-irH_{\epsilon} } e^{-i(t-r)H} \phi\right) dr \\
 &=  (-i)(e^{-i\epsilon}-1) \int_0^t
  e^{-irH_{\epsilon} }H e^{-i(t-r)H}  \phi \ dr. 
 \end{align*} 
 Now let $t>0$ be fixed. By (\ref{estimU(t)}) and the  previous equality, we see that
\begin{align*}
\Vert e^{-itH_{\epsilon}}\phi - e^{-itH}\phi \Vert_{0} & \leq C e^{tR\epsilon} 
|e^{-i\epsilon} -1| \int_0^t \Vert e^{-i(t-r)H}  H \phi \Vert_{0} \ dr \underset{\epsilon \to 0^+}{\longrightarrow} 0,
\end{align*}
 i.e. $e^{-itH_{\epsilon}}\phi$  
  converges in $L^2$ norm  to  $e^{-itH}\phi$ as $\epsilon\to 0$ for all $\phi\in \mathcal{C}^{\infty}_0(\R^3) $. This  by uniqueness of the weak limit in
 (\ref{limU(t)}) gives 
 $$\inp{e^{-itH}\phi,\psi} = \inp{U(t)\phi,\psi}, \quad \forall \phi,\psi \in \mathcal{C}^{\infty}_0(\R^3), \quad \forall t>0. $$ 
 Finally, by density of $\mathcal{C}^{\infty}_0(\R^3)$ in  $L^{2,s}$, 
we conclude that  
$$ \langle U(t)f,g \rangle =\langle  e^{-itH}f ,g \rangle, \quad \forall  f,g \in L^{2,s}, \quad \forall t>0,$$
 which establishes the desired representation formula.
\end{proof}
Next we quote a lemma for some generalized integrals given in \cite[Section II.2.]{gelfand1964generalized}. 
\begin{lem}\label{Lem:GenInt}
\begin{enumerate}
\item  The limit of the function $ \lambda \mapsto (\lambda+i\mu)^{-1}$ as $\mu\to 0^+$, is the generalized function $(\lambda+i0)^{-1}$  defined in the following sense:\\
For every test function $\phi\in \mathcal{C}^{1}_0(\R)$
\begin{align*}
 \left( (x+i0)^{-1}, \phi(x)\right) = \int_{|x|\leq 1} \frac{\phi(x)-\phi(0)}{x}\, dx + \int_{|x|>1} \frac{ \phi(x)}{x} \, dx - i\pi \phi(0).
\end{align*}
Moreover, for $t>0$  we have the following generalized integral
$$\int_{\R} \frac{e^{-it\lambda}}{\lambda+i0}  \,d\lambda =
-i 2 \pi.$$
\item For $t>0$ and $j=-1,0,1,\cdots$ we have 
\begin{align*}
\int_{0}^{+\infty} \lambda^{j/2} e^{-it \lambda} \,d\lambda = \Gamma(\frac{j}{2}+1) (-it)^{-\frac{j}{2}-1},
\end{align*}
where $\displaystyle \Gamma(\frac{j}{2}+1) = \int_0^{+\infty} t^{j/2} e^{-t} \, dt$.
\end{enumerate}
\end{lem}
Now we are able to prove Theorem \ref{thm:estim}.
Before starting the proof, let us rewrite the representation formula in (\ref{intg}) as follows:
\begin{align}
\inp{e^{-itH}f,g} = & \sum_{j=1}^p \inp{e^{-itH} \Pi_{z_j}f,g} \nonumber \\
+&\lim\limits_{\eta \to 0} \lim\limits_{\epsilon \to 0} \frac{1}{2i\pi} \int_{\widetilde{\Gamma}^{\nu}(\eta,\epsilon)} e^{-ite^{-i\epsilon}z}\inp{ (H-z)^{-1}f,g} dz, \label{RepresFormula1}
\end{align}
 after some analytic deformation of the curve $\Gamma^{\nu}(\eta,\epsilon)$ in the following sense: \\
$\widetilde{\Gamma}^{\nu}(\eta,\epsilon) \cap 
\left( \sigma_d^+(H) \cup \sigma_r^+(H)\right) =\emptyset$ and
$$\widetilde{\Gamma}^{\nu}(\eta,\epsilon) = \widetilde{\Gamma}_-^{\nu}(\epsilon)\cup \sigma_-^{\nu}(\epsilon) \cup C_0(\eta,\epsilon) \cup \Gamma_1(\eta,\epsilon) 
\cup \Gamma_+(\epsilon),$$ where
\begin{align*}
\widetilde{\Gamma}_-^{\nu}(\epsilon)&= \lbrace z=  \nu_{\epsilon} -\lambda e^{i\nu}, \ \lambda \geq 0 \rbrace,\ \nu_{\epsilon}= \nu -i \epsilon, \\
C_0(\eta,\epsilon)&= \{z= \eta e^{i(2\pi-\theta)}, \  \epsilon_{\eta} <\theta<2\pi - \epsilon_{\eta} \}, \ \epsilon_{\eta}= \arcsin (\epsilon/\eta),\\
\sigma_-^{\nu}(\epsilon)& = \{ z=\lambda -i \epsilon, \ \eta(\epsilon)\leq \lambda\leq \nu \}, \ \eta(\epsilon)= \eta \cos{\epsilon_{\eta}},
\end{align*}
the curves  $\Gamma_1(\eta,\epsilon)$ and $ \Gamma_+(\epsilon)$
are given in (\ref{curve:gamma}). See Figure 2.
\begin{figure}[h!]
\begin{tikzpicture}[scale=0.7]
	[domain=0:5]
	\draw[->,very thin, gray] (-4,0) -- (11,0) node[right] {\textcolor{black}{$\R$}};
	\draw[->,very thin, gray] (0,-2) -- (0,2) node [above,black] {$i\R$};
	\draw[-] (32:0.5) -- (1.5,0.25)node[midway,above, sloped]{\small $\sigma_1(\eta,\epsilon)$} ;
	\draw[-] (2.5,0.25) -- (3.5,0.25) ;
	\draw[-]  (4.5,0.25) -- (7.5,0.25) node[midway, sloped, gray] {$\small >$} ;
	\draw[->] (8.5,0.25) -- (11,0.25);
	\draw[-] (1,-0.25) -- (-2,-2) node[above]{$\small \widetilde{\Gamma}_-^{\nu}(\epsilon)$} node[midway, sloped, gray] {$>$};
\draw[-] (-27:0.5) -- (1,-0.22) node[below]{ $\small \sigma_-^{\nu}(\epsilon)$};
	\draw(32:0.5) arc (35:333:0.5) node[midway,above left]{\small $C_0(\eta,\epsilon) $} node[midway, sloped, gray] {$<$};
	\draw (2.5,0.25) arc (0:180:0.5)node[midway, sloped, gray] {\small $>$} ;
	\draw (4.5,0.25) arc (0:180:0.5)node[midway,above]{\small $\gamma_2(\eta,\epsilon) $} ;
	\draw (8.5,0.25) arc (0:180:0.5) ;
	\draw[thick] plot[mark=*] (2,0)
	node[below]{$\small \lambda_1$};
	\draw[thick] plot[mark=*] (4,0) node[below]{$\small \lambda_2 $};
	\draw[thick] plot[mark=*] (8,0)node[below]{$\small\lambda_N$};
	\draw (0,0) node{$\bullet$};
\draw[thick] plot(-2.5,0) node{$\times$} ;
    \draw[very thick] plot(1,2) node{$\times$} ;
    \draw[thick] plot(4,2.5) node{$\times$};
    \draw[thick] plot(7,2) node{$\times$};
    \draw[thick] plot(6.5,1) node{$\times$};
    \draw[thick] plot(9,2) node{$\times$};
    \draw[thick] plot(-3.5,1.5) node{$\times$};
    \draw[thick] plot(-2,2) node{$\times$};
\end{tikzpicture}
\caption{The curve $\widetilde{\Gamma}^{\nu}(\eta,\epsilon)$.}
\end{figure}
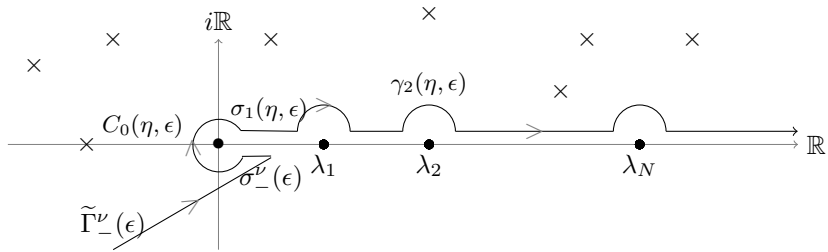 
\begin{proof}[Proof of Theorem \ref{thm:estim}]
Let $\chi:\R \to \R$ be a cutoff function such that $\chi(\lambda)=1$ for $|\lambda|\leq \nu/2$ and $\chi(\lambda)=0$ for $|\lambda|\geq \nu$.
 For $j=0,1,\cdots,N$, we define $\chi_j(\lambda)= \chi(\lambda-\lambda_j)$, where $\lambda_0=0$ and $\lambda_j\in \sigma_r^+(H)$, $\forall j=1,\cdots,N$. Set $$g_{\epsilon}(z,t)=e^{-ite^{-i\epsilon}z} \langle (H-z)^{-1}f,g\rangle. $$

First, we  prove the part $(b)$ of the theorem. We begin by introducing at the second member of (\ref{RepresFormula1})
 the resolvent expansions near zero energy  and positive resonances which are obtained in theorems \ref{thm:cas1} and \ref{Thm:realRes1} respectively
 \begin{align}
\int_{\widetilde{\Gamma}^{\nu}(\eta,\epsilon)} g_{\epsilon}(z,t)\, dz
& = \sum_{s=-2}^1 \langle R_s^{(2)}f,g\rangle  \int_{\widetilde{\Gamma}^{\nu}(\eta,\epsilon)} z^{s/2}\,e^{-ite^{-i\epsilon}z} \chi_0(\Re z)\, dz\nonumber \\
&+\sum_{j=1}^N \sum_{s=-1}^{l-2} \langle R_{s}(\lambda_j)f,g\rangle \int_{\widetilde{\Gamma}^{\nu}(\eta,\epsilon)}(z-\lambda_j)^j
e^{-ite^{-i\epsilon}z} \chi_j(\Re z)\, dz\nonumber\\
&+ \int_{\widetilde{\Gamma}^{\nu}(\eta,\epsilon)} e^{-ite^{-i\epsilon}z}  \langle\widetilde{R}^{(2)}_1(z)f,g \rangle \chi_0(\Re z)\, dz\nonumber  \\
&+\sum_{j=1}^N \int_{\widetilde{\Gamma}^{\nu}(\eta,\epsilon)} e^{-ite^{-i\epsilon}z}  \langle\widetilde{R}_{\ell}(z-\lambda_j)f,g \rangle \chi_j(\Re z)\, dz \nonumber\\
&+ \int_{\widetilde{\Gamma}^{\nu}(\eta,\epsilon)} e^{-ite^{-i\epsilon}z}  \langle R(z)f,g \rangle \chi_+(\Re z)\, dz \nonumber\\
&:= I_0^{\epsilon,\eta}(t) + \sum_{j=1}^N \sum_{s=-1}^{\ell-2} I_s^{\epsilon,\eta}(t,\lambda_j) + J_0^{\epsilon,\eta}(t) \label{Int1} \\
&+ \sum_{j=1}^N J^{\epsilon,\eta}(t,\lambda_j)+ J_{\pm}^{\epsilon}(t),\label{Int2}
 \end{align}
 where $\chi_+= 1 -\sum_{j=0}^N \chi_i$.\\ 
Let  $t>0$ be fixed. We see that
 \begin{align*}
 I_0^{\epsilon,\eta}(t)&=\langle R_{-2}^{(2)}f,g\rangle  \left[\int_{\omega_0} \frac{e^{-ite^{-i\epsilon}z}}{z}  \chi(\Re z)\, dz
 -\int_{\nu+i\epsilon}^{\nu-i\epsilon} \frac{e^{-ite^{-i\epsilon}z}}{z}  \chi(\Re z)dz\right]\\
 &+ \langle R_{-1}^{(2)}f,g\rangle \int_{\eta(\epsilon) }^{+\infty} \left(\frac{e^{-ite^{-i\epsilon}(\lambda+i\epsilon)}}{\sqrt{\lambda+i\epsilon}} - \frac{e^{-ite^{-i\epsilon}(\lambda-i\epsilon)}}{\sqrt{\lambda-i\epsilon}}\right) \chi(\lambda)d\lambda
 \\
 &+ I_1^{\epsilon,\eta}(t),
 \end{align*}
where $\omega_0:=
[\nu-i\epsilon, \eta\cos{\epsilon_{\eta}}-i\epsilon]
\cup
\vC_0(\eta,\epsilon)
\cup
[\eta \cos{\epsilon_{\eta}}+i\epsilon, \nu+i\epsilon]
\cup
[ \nu+i\epsilon, \nu-i\epsilon]$ 
is a closed curve enclosing zero traveled in the clockwise sense. 
It follows that $I_0^{\epsilon,\eta}(t)-I_1^{\epsilon,\eta}(t)$  converges as $\epsilon \to 0^+$, $\eta \to 0^+$ respectively, in the sense of generalized functions of $t$,  to 
\begin{align*}
- 2i\pi \langle R_{-2}^{(2)}f,g\rangle & +2 \langle R_{-1}^{(2)}f,g\rangle \left[\int_0^{+\infty} \frac{e^{-it\lambda}}{\sqrt{\lambda}}d\lambda + \int_{\nu/2}^{+\infty} \frac{e^{-it\lambda}}{\sqrt{\lambda}}( \chi(\lambda) -1) d\lambda\right] \\
 &=-2i\pi \langle R_{-2}^{(2)}f,g\rangle +2 (-i\pi)^{1/2} 
 \langle R_{-1}^{(2)}f,g\rangle t^{-1/2}+ \mathcal{O}(t^{-2}), 
 \end{align*}
 where the decay rate $t^{-1/2}$ can be seen from  Lemma \ref{Lem:GenInt}. In addition, 
 the integral $I_1^{\epsilon,\eta}(t)$ can be easily estimated using Lemma \ref{Lem:GenInt} as follows
\begin{align}
\lim\limits_{\eta\to 0^+} \lim\limits_{\epsilon\to 0^+} I_1^{\epsilon,\eta}(t)&= 
2\inp{R_1^{(2)}f,g} \int_0^{+\infty} \sqrt{\lambda} e^{-it\lambda} \chi(\lambda)\, d\lambda\nonumber \\
&=-(i\pi)^{1/2}\inp{R_1^{(2)}f,g} t^{-3/2}+ \mathcal{O}(t^{-2})\label{Integ:I_1}.
\end{align}
 Moreover, for $j=1,\cdots,N$, we can show that $I_{-1}^{\epsilon,\eta}(t,\lambda_j)$ in (\ref{Int1}) converge as $\epsilon\to 0^+$ to integrals
\begin{equation*}
 I_{-1}^{\eta}(t,\lambda_j)= \langle R_{-1}(\lambda_j)f,g \rangle e^{-it\lambda_j}\int_{L_{\eta}}
\frac{ e^{-it\xi}}{\xi}\, \chi(\Re \xi)\, d\xi,  \ \xi=z-\lambda_j, 
 \end{equation*}
along the contour $L_{\eta}= ]-\infty,-\eta]\cup \{ \xi=\eta e^{i(\pi-\theta)}, \ 0<\theta<\pi\}\cup [\eta,+\infty[$ traveled from $-\infty $ to $+\infty$.
Then, the limit integrals $I_{-1}^{\eta}(t,\lambda_j)$ converge as $\eta\to 0$, in the sense of generalized function of $t$, to
\begin{align*}
I_{-1}(t,\lambda_j):=&  \langle R_{-1}(\lambda_j)f,g\rangle  e^{-it\lambda_j} \int_{-\infty}^{+\infty}\frac{ e^{-it\lambda}}{\lambda+i0}\chi(\lambda)d\lambda\\
=& -2i\pi \langle R_{-1}(\lambda_j)f,g\rangle  e^{-it\lambda_j} 
+ \mathcal{O}(t^{-2}), \quad \forall t>0.
\end{align*}
However, for $r=0,1,\cdots,l-2$, we have  
$$I_r^{\eta,\epsilon}(t,\lambda_j) \underset{\eta\to 0^+,\epsilon \to 0^+}{\longrightarrow} \langle R_{r}(\lambda_j)f,g\rangle  e^{-it\lambda_j} \int_{+\infty}^{-\infty} \lambda^r e^{-it\lambda}\chi(\lambda)d\lambda, \quad \forall t>0,$$
where the right member decays rapidly at infinity as the $r$-$th$ derivative of the Fourier transform of the cut off function $\chi$ such that the convergence holds in the sense of regularized function.\\ 
Let now estimate $J_{\pm}^{\epsilon}(t)$ and $J_j^{\epsilon,\eta}(t)$ as $t\to +\infty$ given in (\ref{Int2}). 
We decompose $J_{\pm}^{\epsilon}(t)$ as follows:
\begin{eqnarray}
J_{\pm}^{\epsilon}(t) &= &
\int_{\Gamma_1(\eta,\epsilon) \cup\Gamma_+(\epsilon)} g_{\epsilon}(z,t) \chi_+(\Re z) \, dz 
+ \int_{\Gamma_-^{\nu}(\epsilon)}g_{\epsilon}(z,t) \chi_+(\Re z)\, dz\\
&:=& J_+^{\epsilon}(t) + J_-^{\epsilon}(t).
\end{eqnarray}
The integral  $J_-^{\epsilon}(t)$ has an exponential time-decay $\mathcal{O}(e^{-tc_{\nu}})$ for some  $c_{\nu}>0$ independent on $\epsilon$ and $t$.
Regarding the integral $ J_+^{\epsilon}(t)$, it follows from (\ref{Int by Part})  that $$
 \lim\limits_{\epsilon\to 0^+}|J_+^{\epsilon}(t)|=
\mathcal{O}(t^{-2}) \Vert f\Vert_{0,s}  \Vert g\Vert_{0,s},\  \text{as} \ t\to +\infty.$$
Next, we have to estimate $J_0^{\epsilon}(t)$. By Theorem \ref{thm:cas1}, if $\rho>7$, $z\mapsto \widetilde{R}_1^{(2)}(z)$ can be continuously extended to $\mathcal{C}^2 (\{|z|<\delta,\ \pm \Im z\geq 0\}, \b(-1,s,1,-s) )$ for $s>7/2$, such that 
$$
\Vert \frac{d^r}{d\lambda^r}\widetilde{R}_1^{(2)}(\lambda\pm i0)\Vert_{\b(-1,s,1,-s)}=o(\lambda^{\frac{1}{2}-r}),\ 0<\lambda<\delta, \ r=0,1,2.$$
Then, it follows from  Lemma 10.2 in \cite{jensen1979spectral} that 
\begin{equation}\label{estimate:J}
\lim\limits_{\eta\to 0}\lim\limits_{\epsilon\to 0} J_0^{\epsilon}(t)= o(t^{-\frac{3}{2}})  \quad \text{as} \ t\to +\infty.    
\end{equation}
Finally, let $j=1,\cdots,N$. We have  by Theorem \ref{Thm:realRes1} that if $\rho>2\ell-1$ and $s>\ell-1/2$, $\ell\in \bN$ with $\ell\geq 2,$ the remainder term $\widetilde{R}_{\ell-2}(z-\lambda_j)$ can be continuously extended to 
\begin{equation}\label{R1:boundary}
\widetilde{R}_{\ell-2}(\lambda-\lambda_j+i0) \in  \mathcal{C}^{\ell-2} \left( \{ \lambda>0, \ |\lambda-\lambda_j|<\delta\}, \b(-1,s,1,-s)\right).
\end{equation}
Then, as the above integrals $I_{r}^{\eta,\epsilon}(t,\lambda_j)$,
 $J^{\epsilon,\eta}(t,\lambda_j)$ converges as $\epsilon,\eta \to 0$ in the sense of regularized functions to the Fourier transform in $t$ of the regular and compactly supported function $\lambda \mapsto \widetilde{R}_{\ell-2}(\lambda+i0)\chi(\lambda)$. This in view of (\ref{R1:boundary}) gives
\begin{equation}\label{estim:Jres}
\lim\limits_{\eta\to 0} \lim\limits_{\epsilon \to 0} J^{\epsilon,\eta}(t,\lambda_j)
= o(t^{-\ell+2})\Vert f\Vert_{0,s}  \Vert g\Vert_{0,s},\  t\to +\infty.
\end{equation} 
See \cite[Lemma 10.1]{jensen1979spectral}.\\

 We have proved the part $(b)$. 
In this theorem, the stronger condition $\rho>7$ and $s>7/2$ is required to obtain $o(t^{-2})$ in (\ref{estim:Jres}) ($\ell=3$) and to get  $o(t^{-3/2})$ in  (\ref{estimate:J}).  But this condition can be relaxed to $\rho>4$ to obtain (\ref{R(z):cas2}) with remainder estimate $o(|z|^{-1/2})$ and then to get the remainder $o(t^{-1/2})$ in (\ref{estimate:J}). \\

Now, we return to the proof of the part $(a)$ of Theorem \ref{thm:estim}. We have only to compute the integral $I_0^{\epsilon, \eta}(t)$ which does not have the same behavior as in the previous proof. Indeed, when zero is  a resonance and not an eigenvalue,  $I_0^{\epsilon,\eta}(t)$ is replaced by
\begin{align*}
 \widetilde{I}_0^{\epsilon,\eta}(t)&= \sum_{j=-1}^1 \langle R_{j}^{(1)}f,g\rangle  \int_{C_0(\eta,\epsilon)} e^{-ite^{-i\epsilon}z} z^{j/2} \, dz \\
 &+ \langle R_{-1}^{(1)}f,g\rangle \int_{\eta(\epsilon) }^{+\infty}\left(\frac{e^{-ite^{-i\epsilon}(\lambda+i\epsilon)}}{\sqrt{\lambda+i\epsilon}}-  \frac{e^{-ite^{-i\epsilon}(\lambda-i\epsilon)}}{\sqrt{\lambda-i\epsilon}}\right) \chi(\lambda)\, d\lambda + \widetilde{I}_1^{\epsilon,\eta}(t),
\end{align*}
 where $R_{-1}^{(1)}$ is the one rank operator defined in Theorem \ref{thm:cas2}.
It is easily to  show that the first integral at the right hand side vanishes as $\eta\to 0$.  However, the second integral tends as $\epsilon\to 0^+$, in the sense of generalized functions, to 
\begin{align*}
2\int_{0}^{+\infty} \frac{e^{-it\lambda}}{\sqrt{\lambda}} \, d\lambda
+2\int_{\nu/2}^{+\infty} \frac{e^{-it\lambda}}{\sqrt{\lambda}} (\chi(\lambda)-1)\, d\lambda
 &=2(-i\pi)^{1/2} t^{-\frac{1}{2}}+ \mathcal{O}(t^{-2}).
\end{align*}
 Also, see (\ref{Integ:I_1}) for the estimate of $\widetilde{I}_1^{\epsilon,\eta}(t)$. 
 \\
 
 Regarding the estimate of $J_{0}^{\epsilon}(t)$ in (\ref{estimate:J}), 
to obtain the decay $o(t^{-1/2})$ it is required that $\rho>3$ and $s>3/2$. However, in presence of positive resonances our assumption $\rho>5$ and $s>5/2$ is needed to obtain $o(t^{-1})$ in (\ref{estim:Jres}), according to Theorem \ref{Thm:realRes1}.\\ 

\end{proof}
We end the paper by the following remark:
\begin{rem}
If we assume that the condition (\ref{cond:d(z)}) is satisfied instead of (\ref{cond:detReson}), then using the expansion 
(\ref{R(z):res1}) of $R(z)$ near $\lambda_j\in \sigma_r^+(H)$,
the oscillating term $\displaystyle \sum_{j=1}^N e^{-it\lambda_j} R_{-1}(\lambda_j)$ in (\ref{dev:solution}) will be replaced 
by 
$$\displaystyle \sum_{j=1}^N e^{-it\lambda_j} R_{j}(t,\lambda_j),$$ 
where 
$R_{j}(t,\lambda_j)$ is a polynomial of $t$ of degree at most
$\mu_j-N_j$ with values in $\b(0,s,0,-s)$ and with leading term 
$\displaystyle -\frac{(-it)^{\mu_j-N_j}}{(\mu_j-N_j)!} \mathcal{R}(\lambda_j)$. See 
Theorem \ref{Thm:realRes1}. Indeed, for $\displaystyle \ell=0,1,\cdots, \mu_j-N_j$, 
$ (\lambda+i0)^{-\ell-1}$ is defined as the $\ell$-th derivative, in the sense of generalized functions, of 
$\frac{(-1)^{\ell}}{\ell!}(\lambda+i0)^{-1}$.
Then, by integrating by part we have 
\begin{align*}
\int_{\R} \frac{e^{-it\lambda}}{(\lambda+i0)^{\ell+1}} \chi(\lambda) \, d\lambda &= \frac{(-it)^{\ell}}{l!}\int_{\R} \frac{e^{-it\lambda}}{(\lambda+i0)} \chi(\lambda) d\lambda \\
&+ \sum_{j=1}^{\ell }  \frac{(-1)^{\ell}}{j!(\ell-j)!}(it)^{\ell-j}
\int_{\frac{\nu}{2}<|\lambda|<\nu} \frac{e^{-it\lambda}}{(\lambda+i0)} \frac{d^j \chi}{d\lambda^j}(\lambda) d\lambda\\
&= a_{\ell} t^{\ell} + a_{\ell-1} t^{\ell-1} + \cdots + a_1 t +a_0,
\end{align*}
with $a_{\ell}=(-i)^{\ell+1}  \frac{2\pi}{\ell!}$.
\end{rem}





\subsection*{Acknowledgment}
I would like to express my gratitude to my supervisor Xue Ping Wang for helpful discussions and his useful comments on this research work.

\end{document}